\documentclass{article}
\usepackage{hyperref, graphicx}
\usepackage{amsmath}
\usepackage{tocloft}
\usepackage{amssymb}
\usepackage{slashed}
\usepackage{yfonts}
\usepackage[T1]{fontenc}
\usepackage[utf8]{inputenc}
\usepackage[english]{babel}
\usepackage{amsthm}
\usepackage{tikz-cd}
\usepackage{quiver}
\usepackage{mathabx}
\usepackage{stmaryrd}
\usepackage{braket}
\usepackage{physics}
\usepackage{adjustbox}
\usepackage[most]{tcolorbox}
\usepackage{geometry}
\usepackage[
backend=biber,
style=numeric-comp,
sorting=none,
]{biblatex}
\newtheorem{theorem}{Theorem}[section]
\newtheorem{lemma}[theorem]{Lemma}

\newtheorem{proposition}[theorem]{Proposition}
\newtheorem{conjecture}[theorem]{Conjecture}

\usepackage{authblk}
    \allowdisplaybreaks

\hypersetup{
pdfstartview = {FitH},
}
\hypersetup{
	colorlinks=true,         
	linkcolor=blue,          
	citecolor=red,        
	urlcolor=blue            
}

\DeclareMathAlphabet{\mathpzc}{OT1}{pzc}{m}{it}

\makeatletter
\def\@font@info#1{}
\makeatother

\theoremstyle{remark}
\newtheorem{remark}{Remark}[section]
\newenvironment{rmk}{\begin{remark}}{\hfill$\Diamond$\end{remark}}

\theoremstyle{definition}
\newtheorem{definition}[theorem]{Definition}

\usepackage{mathtools}

\usepackage{mathrsfs}

\numberwithin{equation}{section}

\makeatletter
\newcommand{\ostar}{\mathbin{\mathpalette\make@circled\star}}
\newcommand{\make@circled}[2]{%
  \ooalign{$\m@th#1\smallbigcirc{#1}$\cr\hidewidth$\m@th#1#2$\hidewidth\cr}%
}
\newcommand{\smallbigcirc}[1]{%
  \vcenter{\hbox{\scalebox{0.77778}{$\m@th#1\bigcirc$}}}%
}
\makeatother

\newcommand{\beq}{\begin{eqnarray}}
\newcommand{\eeq}{\end{eqnarray}}

\def\g{\mathfrak g}

\def\h{\mathfrak h}

\def\G{\mathfrak{G}}

\def\C{\mathfrak{C}}

\def\bbZ{\mathbb{Z}}
\def\R{\mathrm{P}}
\def\bbC{\mathbb{C}}

\newcommand{\cA}{{\cal A}}
\newcommand{\cC}{{\cal C}}
\newcommand{\cD}{{\cal D}}
\newcommand{\cE}{{\cal E}}
\newcommand{\cF}{{\cal F}}

\newcommand{\cH}{{\cal H}}
\newcommand{\cM}{{\cal M}}

\newcommand{\cO}{{\cal O}}
\newcommand{\cV}{{\cal V}}

\makeatletter
\newsavebox{\@brx}
\newcommand{\llangle}[1][]{\savebox{\@brx}{\(\m@th{#1\langle}\)}%
  \mathopen{\copy\@brx\kern-0.5\wd\@brx\usebox{\@brx}}}
\newcommand{\rrangle}[1][]{\savebox{\@brx}{\(\m@th{#1\rangle}\)}%
  \mathclose{\copy\@brx\kern-0.5\wd\@brx\usebox{\@brx}}}
\makeatother

\newcommand{\id}{\operatorname{id}}

\begin{document}

\title{Combinatorial quantization of 4d 2-Chern-Simons theory I: \\ the Hopf category of higher-graph operators}
\author[1]{{ \sf Hank Chen}\thanks{hank.chen@uwaterloo.ca}\thanks{chunhaochen@bimsa.cn}}

\affil[1]{\small Beijing Institute of Mathematical Sciences and Applications, Beijing 101408 , China}

\maketitle

\bigskip

\begin{abstract}
2-Chern-Simons theory, or more commonly known as 4d BF-BB theory with gauged shift symmetry, is a natural generalization of Chern-Simons theory to 4-dimensional manifolds. It is part of the bestiary of higher-homotopy Maurer-Cartan theories. In this article, we present a framework towards the combinatorial quantization of 2-Chern-Simons theory on the lattice, taking inspiration from the work of Aleskeev-Grosse-Schomerus three decades ago. The central geometric input is a "2-graph" $\Gamma^2$ embedded in a 3d Cauchy slice $\Sigma$, which has equipped the structure of a discrete 2-groupoid. Upon such 2-graphs, we model the extended Wilson surface operators in 2-Chern-Simons holonomies as Crane-Yetter's \textit{measureable fields}. We show that the 2-Chern-Simons action endows these 2-graph operators --- as well as their quantum 2-gauge symmetries --- the structure of a Hopf category, and that their associated higher $R$-matrix gives it a categorical quasitriangularity structure, which we call the {\it cobraiding}. This is an explicit realization of the categorical ladder proposal of Baez-Dolan, in the context of Lie group 2-gauge theories on the lattice. Moreover, we will also analyze the lattice 2-algebra on the graph $\Gamma$, and extract the observables from it.
\end{abstract}

\newpage

\tableofcontents

\newpage

\section{Introduction}
Over the past century, it was discovered that there is a very interesting interplay between low-dimensional geometry topology and physics. In particular, the work of Witten \cite{WITTEN1990285} revealed that the Wilson line observables in the 3-dimensional Chern-Simons theory (as well as its boundary integrable field theory \cite{KNIZHNIK198483}) computed 3-manifold invariants associated to knot complements. On the other hand, the geometry of framed knots and ribbons up to isotopy --- also known as \textit{skein theory} --- are well-known \cite{SHUM199457,FREYD1989156} to admit a description in terms of the so-called \textit{ribbon categories}. These are purely algebraic data, defined by monoidal categories equipped with additional rigidity and braiding structures. The computation of polynomial knot invariants from such purely algebraic input has also been formalized \cite{Kauffman1987StateMA,Turaev:1992}.

The stage set by this "low-dimensional triangle", between 3d topological quantum field theories (TQFTs)/2d integrable systems, knot invariants and categorical homotopy algebra, has a central player: the theory of \textit{quantum group Hopf algebras} \cite{drinfel1988,Woronowicz1988,Majid:1996kd} and the (unitary) modular ribbon category of its representations. The seminal works of Reshetikhin-Turaev \cite{Reshetikhin:1991tc,Reshetikhin:1990pr} in particular explained in great detail how the structure of quantum group Hopf algebras --- particularly those of the Drinfel'd-Jimbo type coming from quantum deformations \cite{Jimbo:1985zk,Drinfeld:1986in}, such as $U_q\mathfrak{sl}_2$ arising out of the $SU(2)_k$ Chern-Simons theory --- gave rise to invariants of framed knots and tangles. This formulation came to be known as the "Reshetikhin-Turaev functor"; the idea that, conversely, \textit{all} 3-2-1 functorial TQFTs for a given target \cite{Atiyah:1988,lurie2008classification} are determined by such ribbon functors is known as the \textit{(1-)tangle hypothesis} \cite{Baez:1995xq}. These ideas have also been applied very successfully to quantize (2+1)-dimensional gravity with cosmological constant \cite{Mizoguchi:1991hk,Bonzom:2014bua,Pranzetti:2014xva,Livine:2016vhl,Majid:2008iz}, which are known classically to be equivalent to a certain type of Chern-Simons theory \cite{Witten:1988hc,Freidel:2004nb,Meusburger:2003ta,Dupuis:2020ndx}.

\medskip

Direct computations of the 3-manifold quantum invariants involved in the above story, on the other hand, is a notoriously difficult problem itself. A way to make this problem less challenging came in the form of \textit{combinatorial state sum models} by taking a piecewise linear (PL) approximation of the underlying manifold\footnote{This is due to a classic theorems of Whitehead, which states that smooth manifolds have a unique PL structure given by its triangulation.}. This procedure computes the TQFT partition function by breaking it into \textit{local} pieces of "admissible" algebraic/categorical data \cite{Cui_2017}, which are invariant under the so-called combinatorial Pachner moves \cite{Pachner1991Pachner,Dijkgraaf1991}. This idea has been very successfully applied to not only compute the quantum scattering amplitudes in 3d Regge gravity \cite{GirelliOecklPerez:2001PachnerSpinFoam,Crane:2003ep,FREIDEL2000237}, but also to characterize topological phases in condensed matter theory \cite{Levin_2005,Wen2016,Lan2013,KitaevKong_2012}. 

It is known that, in the case of the Turaev-Viro TQFT\footnote{This is related to the Reshetikhin-Turaev TQFT through the Drinfel'd centre of its input category: $Z_{RT}^\cC = Z_{TV}^{Z_1\cC}$.} with the quantum group $U_q\g$ and its representation category as algebraic input, these combinatorial local pieces in the corresponding Barrett-Westbury state sum model \cite{Barrett1993} are given by the \textit{quantum $6j$-symbols} \cite{Turaev:1992hq}. The relationship of these $6j$-symbols to the 3d chain mail invariants in skein theory has also been studied in \cite{ROBERTS}. On the other hand, these $6j$-symbols can also be obtained as scattering amplitudes in a {\it discrete} version of Chern-Simons theory --- that is, we have a way to compute the combinatorial 3-simplex amplitudes directly \textit{without} prior knowledge of skein theory and surgery theory. This is thanks to the foundational works of Alekseev-Grosse-Schomerus \cite{Alekseev:1994pa,Alekseev:1994au}, where the full combinatorial Hamiltonian quantization of discrete Chern-Simons holonomies was pinned down. These works serve as the inspiration of this paper.

\subsection{Motivation}
The success of the above relationship between physics, categorical algebra and topology begs the question of how these correspondences would look like in higher dimensions. Based on the categorical ladder proposal of Baez-Dolan \cite{Baez:1995ph}, as well as the cobordism hypothesis proven in \cite{lurie2008classification}, it is expected that higher dimensional physics and geometry is described by a certain "higher-dimensional algebra". Each of the corners of the above triangle has seen such a "categorification" in recent years,
\begin{enumerate}
    \item categories $\rightarrow$ weak $n$-categories \cite{Lurie:2009,Kong:2014qka,Johnson-Freyd:2020usu,Baez:1997},
    \item knot polynomials $\rightarrow$ knot homology \cite{Khovanov:2000,Elias2010ADT,webster2013knot,Rouquier2005CategorificationOS},
    \item 3d Chern-Simons theory $\rightarrow$ \textbf{4d 2-Chern-Simons theory} \cite{Zucchini:2021bnn,Soncini:2014,Chen:2022hct},
\end{enumerate}
and it has been postulated that a "categorical quantum group" --- with the structure of a Hopf monoidal category \cite{Crane:1994ty,Baez:1995xq,Pfeiffer2007,Green:2023qqr,Chen:2023tjf,BAEZ1996196,neuchl1997representation} --- governs their correspondences. However, how these ideas are related have not yet been made clear: the key issue seems to be that each of these corners have their own different notions of "higher-dimensional algebra": respectively, they are (1) the 2-vector spaces of Kapranov-Voevodsky \cite{Kapranov:1994}, (2) the Soergel bimodules \cite{liu2024braided}, and finally (3) the 2-vector spaces of Baez-Crans \cite{Baez:2003fs}. 

Though, it is known from homotopy theory that any algebraic description of framed 2-tangles must have some "higher categorical" flavour \cite{getzler1998higher,BAEZ2003705}. Indeed, the discovery of the crossed-complex model for a higher categorical version of groups, called \textit{2-groups/categorical groups}, dates back to the 40's by Whitehead \cite{Whitehead:1941}, in the context of homotopy 2-types \cite{Brown,Ang:2018rls}.
\begin{definition}
    A \textbf{(Lie) 2-group} $\mathbb{G}=(G,\mathsf{H},t,\rhd)$ is a (Lie) group crossed-module \cite{Baez:2003fs,Chen:2012gz,chen:2022}, consisting of a pair of (Lie) groups $\mathsf{H},G$, a (Lie) group homomorphism $t:\mathsf{H}\rightarrow G$ and a (smooth) action $\rhd :G\rightarrow\operatorname{Aut}\mathsf{H}$ satisfying the following algebraic conditions
\begin{equation*}
    t(x\rhd y) = xt(y)x^{-1},\qquad t(y)\rhd y' = yy'y^{-1},\qquad \forall~x\in G,~y\in\mathsf{H}.
\end{equation*}
Several equivalent formulations of Lie 2-groups can be given; more will be explained in \textit{Remark \ref{smooth2grp}}.
\end{definition}
\noindent Applications of 2-groups, both Lie and finite, to physics have also been recently studied extensively \cite{Cordova:2018cvg,Benini_2019,Chen:2022hct,Delcamp:2023kew,Bartsch:2022mpm,Bartsch:2023wvv,Chen2z:2023,huang2023tannaka,Huang:2024}.

Based on this current state of affairs, it is then natural to study the geometry of principal Lie 2-group $\mathbb{G}$-bundles (with connection) and its associated higher-gauge theory \cite{Baez:2004in,Wockel2008Principal2A,Nikolaus2011FOUREV,Schommer_Pries_2011}. Over a 4-dimensional manifold $X$, the topological field theory arising from such a categorical gauge principle is known as the \textbf{4d 2-Chern-Simons theory},
\begin{equation*}
    S_{2CS}[A,B]= 2\pi k\int_X \langle B,F(A)-\frac{1}{2}\mu_1B\rangle,
\end{equation*}
whose fundamental fields are given by a polyform of degree-1, $A\in\Omega^1(X) \otimes\g, B\in\Omega^2(X)\otimes\h$, valued in the Lie 2-algebra $\operatorname{Lie}\mathbb{G}=\G=\h\xrightarrow{\mu_1} \g$ corresponding to the underlying (complex, connected, simply-connected) Lie 2-group $\mathbb{G} = \mathsf{H}\xrightarrow{t}G$. Here, $\langle-,-\rangle:\G^{\otimes 2}\rightarrow\bbC[1]$ is a degree-1 non-degenerate pairing form. See \cite{Martins:2010ry,Chen:2022hct,Zucchini:2021bnn,Soncini:2014,Song_2023,Mikovic:2016xmo,Radenkovic:2020weu} for more details on 2-Chern-Simons theory. 

It is worth mentioning that they are part of the bestiary of higher-dimensional \textit{homotopy Marer-Cartan theories} \cite{Jurco:2018sby,Chen:2024axr,Schenkel:2024dcd}, which are \textit{higher derived} generalizations of Chern-Simons theory. See \S \ref{weak2gauthy} for a brief overview of the \textit{weak} version of $S_{2CS}$. Other 4d higher-gauge theories (which may {not} be topological) have also appeared in various guises throughout theoretical physics \cite{Baez:2002highergauge,Sati:2009ic,walker2012,Bullivant:2016clk,Delcamp:2017pcw,Dubinkin:2020kxo,Song:2021}.

\medskip

The current series of papers is dedicated towards answering the following: 
\begin{quote}
\centering
{\em {\large How much does the 4d 2-Chern-Simons theory know about the geometry of 2-tangles and the topology of 4-manifolds?}} 
\end{quote}
\noindent The motivation for starting the quantization of $S_{2CS}$ from the combinatorial perspective of the \textit{discrete} holonomies is that it preps us for an explicit computation of its 4-simplex scattering amplitudes and invariants on a lattice, without having first knowing how to do handlebody surgery theory on 4-manifolds. This would provide an explicit state sum model for a 4d topological 2-gauge theory, which can be understood as a Lie 2-group generalization of the Yetter-Dijkgraaf-Witten TQFT encompassed by the seminal work of \cite{Douglas:2018}. 

While the idea of using higher-dimensional gauge groups \cite{Baez:2004} and homotopy 2-types \cite{Martins:2006hx} to produce 4d TQFTs is not new (it dates back to the 90's \cite{Yetter:1993dh}), many of the explicit examples constructed in the literature so far had only used \textit{finite} 2-groups (see eg. \cite{Kapustin:2013uxa,,Mikovic:2011si,Bullivant:2016clk,Bullivant:2017sjz,Bochniak_2021,Chen2z:2023}), so the resulting TQFTs are always of Yetter-Dijkgraaf-Witten type. These are known to be too simple to produce any exotic 4d invariants \cite{Reutter:2020bav}.

\subsubsection*{A conjecture on the 4d Crane-Yetter TQFT.} An additional motivation for this work is the following. It was argued in \cite{Baez:1995ph} that 4d BF-BB theory with Lie group $G=SU(2)$, which is a special case of 2-Chern-Simons theory on the \textit{inner automorphism 2-group} $\operatorname{Inn}G=G\xrightarrow{\id}G$ (see eg. \cite{Kim:2019owc}), quantizes to a(n oriented) theory which is equivalent to the Crane-Yetter-Broda TQFT \cite{Crane:1994ji}. The formal argument given there, however, was flawed,\footnote{To be more precise, the argument in \cite{Baez:1995ph} was based on a formal path integration over configurations of \textit{only} the 2-form gauge field $B$. This cannot be done in 2-gauge theory, however, since the tuple $(A,B)$ forms a $G$-multiplet, and hence the path integral measure $D[A,B]\sim D[A]D[B]$ cannot be split up without fixing some gauge.} hence a detailed study of the higher-representation theory associated to the 2-Chern-Simons observables is necessary in order to shed light on Baez's claim.

To be more precise, recall that Crane-Yetter TQFT is based on the input pre-modular 2-category $\operatorname{Mod}(\operatorname{Rep}U_q\frak{sl}_2)$; see also \S 3.4.1 in \cite{Douglas:2018}. Secondly, since it is known that the 4d Crane-Yetter-Broda TQFT admits a state sum construction in terms of the so-called "15$j$-symbols" \cite{Crane:1993if}, a direct verification of Baez's conjecture can also be obtained by computing the lattice scattering amplitudes in 4d BF-BB theory. 
\begin{conjecture}\label{baezconjecture}
    \textbf{(Implicitly made in \cite{Baez:1995ph}).}
    \begin{itemize}
    \item \textbf{Algebraic version}:     There is a (ribbon/pre-modular) equivalence of 2-categories
    \begin{equation*}
        \operatorname{Mod}(\operatorname{Rep}U_q\frak{sl}_2) \simeq \operatorname{2Rep}(\mathbb{U}_{q}\operatorname{inn}(\frak{sl}_2)),
    \end{equation*}
    where $\mathbb{U}_{q}\operatorname{inn}(\frak{sl}_2)$ is the Hopf category corresponding to the quantization of $\operatorname{Inn}SU(2)$.\footnote{We will give in \S \ref{2gthopf} a definition of the "categorical quantum enveloping algebra $\mathbb{U}_q\G$" associated to a Lie 2-group $\mathbb{G}$ in the context of the current framework.}
    \item \textbf{Piecewise-linear version}: Given a closed 4-simplex $T^4$, the lattice 2-Chern-Simons scattering amplitudes on $T^4$ coincides with the 15$j$ symbols of Crane-Yetter.
\end{itemize}
\end{conjecture}


\subsection{Overview and results}
In \S \ref{1cs}, we will give a brief review of the quantization of Chern-Simons on the lattice by adapting the formalism of \cite{Alekseev:1994pa} to the language groupoids and functors. This "coherent" setting serves as the template for the framework that we shall develop in \S \ref{2cs}, in which the \textit{discrete} degrees-of-freedom of $S_{2CS}$, living on the edges and faces of a lattice in a codimension-1 Cauchy slice $\Sigma$, is described. 

Then in \S \ref{quantumcategory}, based on the semiclassical Lie 2-bialgebra \cite{Bai_2013,Baez:2003fs}/Poisson-Lie 2-group \cite{Chen:2012gz,Chen:2013} symmetries of $S_{2CS}$ \cite{Chen:2022hct}, we deduce the deformation quantization of the underlying structure Lie 2-group $\mathbb{G}$ and develop the configuration space of discrete 2-Chern-Simons theory --- as well as its categorical quantum gauge symmetries --- by taking inspiration from \cite{Alekseev:1994pa}. 

Since we are dealing with \textit{infinite} Lie 2-groups, these Hopf categories are in some sense infinite-dimensional. As such, the usual theory of finite-dimensional 2-Hilbert spaces $\mathsf{2Hilb}$ \cite{Baez1996HigherDimensionalAI} is not enough. Here, we will develop our framework using an infinite-dimensional version of $\mathsf{2Hilb}$, given by the bicategory $\mathsf{Meas}$ of Crane-Yetter  \textit{measureable categories} \cite{Crane:2003ep,Yetter2003MeasurableC,Baez:2012}. This work is therefore a marriage of both higher-categorical algebra and functional analysis. In the companion paper \cite{Chen:2025?}, the author shows that the 2-category $\operatorname{2Rep}(\tilde{\cC};\tilde R)$ of linear finite semisimple $\tilde{\cC}$-module categories inherits a {\it rigid tensor} structure from this *-operation.

In \S \ref{lattice2alg}, we construct the \textbf{lattice 2-algebra} $\mathscr{B}^\Gamma$ as a categorical semidirect product \cite{Fuller:2015}. It allowed us to extract the 2-holonomy observables on the lattice. The geometry and orientation of the  2-graph $\Gamma^2$ \cite{stehouwer2023dagger,ferrer2024daggerncategories} is then shown to induce a certain *-operation on the lattice 2-algebra $\mathscr{B}^\Gamma$.  

\subsubsection*{Results}
We will prove that the fundamental degree-of-freedom in the quantum lattice 2-gauge theory has equipped a certain Hopf categorical structure, described most naturally in the framework of \textit{internal categories} \cite{douglas2016internalbicategories}.
\begin{theorem}
    Let $\Gamma^2$ denote the 2-groupoid of 2-graphs associated to a lattice $\Gamma\subset \Sigma$ embedded in a 3-dimensional Cauchy slice $\Sigma$ of $X$. 
    \begin{enumerate}
        \item The 2-graph operators $\cC$ on $\Gamma$ has the structure of a {Hopf \textit{co}category} internal to $\mathsf{Meas}_q$.
        \item Under certain regularity conditions, the quantum 2-gauge transformations $\tilde{\cC}$ on $\cC$ has the structure of a {Hopf category} internal to $\mathsf{Meas}_q$.
    \end{enumerate}
    Moreover, both are equipped with a "coraiding", which is a higher categorical analogue of a quasitriangulairty structure (see \textbf{Definition \ref{cobraidingdef}}).
\end{theorem}
\noindent Here, "$\mathsf{Meas}_q$" is a certain "non-commutative" version of the Crane-Yetter measureable categories $\mathsf{Meas}$, described in more detail in \S \ref{hopfopalg}. It can be thought of as the "sheafy" counterpart of the bicommutant categories of Henrique-Pennys \cite{Henriques2017-gm}.

Now it is known that the 2-Chern-Simons action $S_{2CS}$ admits symmetries described by a Lie 2-bialgebra $(\G;\delta)$ \cite{Bai_2013}, and that connected, simply-connected Poisson-Lie 2-groups are in bijection with Lie 2-bialgebras \cite{Chen:2012gz}. By extracting the Hopf structures explicitly, we are then able to prove the following semiclassical limit.
\begin{theorem}
    Assuming the conditions in \textbf{Definition \ref{hypH}}, the category $\cC=\mathfrak{C}_q(\mathbb{G})$ on a PL 2-disc $\Gamma=D^2$ admits a Poisson-Lie 2-group $(C(\mathbb{G});\{-,-\})$ as a semiclassical limit.
\end{theorem}
\noindent We will state this result more clearly and prove it in \S \ref{semiclassical}. The author finds it important to emphasize here (and also in \textit{Remark \ref{hypHconstruction}}) that the conditions in \textbf{Definition \ref{hypH}} need not be invoked in constructions of 4d TQFTs from \textit{finite} 2-groups or Hopf categories.

\medskip

Briefly, the  technical assumptions in \textbf{Definition \ref{hypH}} can be understood (see the end of \S \ref{coprod} as well as \S \ref{hopf2algproblems}) as the existence of a certain \textit{decategorification} map arising from higher derived quantization schemes \cite{Freed:1994ad,Zucchini:2021bnn,Jurco:2018sby,Gaiotto:2024gii,Calaque:2021sgp}. It posits a fundamental relationship between Baez-Crans 2-vector spaces $\mathsf{2Vect}^{BC}=\operatorname{Cat}_\mathsf{Vect}$ \cite{Baez:2003fs} and (an infinite-dimensional analogue of) the Kapranov-Voevodsky 2-vector spaces $\mathsf{2Vect}^{KV}$ \cite{Kapranov:1994,Baez1996HigherDimensionalAI}). The former model the $L_\infty$-algebra structure in higher-gauge theory, while the latter model topological orders and topological defects.

\subsubsection*{Acknowledgement} 
The author would also like to thank Yilong Wang, Jinsong Wu, Hao Zheng, Zhi-Hao Zhang, and Florian Girelli for enlightening discussions throughout the completion of this work.

\section{Graph operators pr{\'e}cis}\label{1cs}
Let us first briefly recall the discrete quantization of Chern-Simons/BF theory. Let $X$ be a framed smooth 3-manifold and let $G$ be a compact Lie group, assumed to be simple. The Chern-Simons partition function is of course written as
\begin{equation*}
    Z(X) = \int D[A] e^{ i2\pi kS_{CS}[A]}=\int D[A] e^{i2\pi k\int_X \langle A,dA+\frac{1}{3}[A,A]\rangle},
\end{equation*}
where $A$ is a $G$-connection on $X$ and $\langle-,-\rangle$ is the Killing form on $\g=\operatorname{Lie}G$. One way to make sense of this partition function is to perform a discretization procedure: we triangulate $X$ and truncate/localize the connection data onto the oriented edges of the dual cells, by defining the holonomy degrees-of-freedom
\begin{equation*}
    h_e = P\exp \int_e A,\qquad e\in T^1_X.
\end{equation*}
Taking an arbitrary 2d Cauchy surface $\Sigma\subset X$ equipped with an induced triangulation $T_\Sigma$, we consider the graph $\Gamma$ Poincar{\'e} dual to the triangulation $T_\Sigma$. An admissible $G$-decorated graph $G^\Gamma$ is then defined by a functor $\Gamma^1\rightarrow G$ such that $\prod_{e\in  c}g_e=1$ on a closed graph cycle $c$ (ie. satisfying the flatness condition). 

We shall in the following consider $\Gamma$ to be planar and non-self-intersecting. Following the philosophy of \cite{Delcamp:2016yix}, the physical Hilbert space on $\Gamma$ is given by the linear span of $\mathbb{C}$-valued functions $\psi: G^{\Gamma^1} \rightarrow \mathbb{C}$ equipped with a certain well-defined (ie. convergent) inner product,
\begin{equation*}
    \mathcal{H} = L^2(G^{\Gamma^1})/\sim,
\end{equation*}
modulo gauge transformations $g_{(01)} \mapsto a_1g_{(01)}a_0^{-1}$, where $a: \Gamma^0\rightarrow G$ are $G$-valued gauge parameters localized on the vertices of $\Gamma$, denoted by $G^{\Gamma^0}$. Treating $L^2(G^{\Gamma^1})$ as a left-regular representation of $G^V$ such that
\begin{equation*}
    a\rhd \psi(\{h_e\}_e) = \psi(\{a_{s(e)}h_ea_{t(e)}^{-1}),\qquad s,t:\Gamma^1\rightarrow \Gamma^0, 
\end{equation*}
this gauge invariance condition can be enforced by a gauge-averaging procedure: 
\begin{equation*}
    \Psi = \int_G \big[\prod_{\Gamma^0} da\big] a\rhd \psi.
\end{equation*}
The quotient modulo gauge transformations, $L^2(G^{\Gamma^1})//G^V$, serves as the basis for the spin network construction \cite{Baez:1994hx}.

\subsection{Coproduct and the antipode} 
We begin with the classical treatment. Given a fixed directed graph $\Gamma\subset \Sigma$, there is a group product on the decorated graphs
\begin{equation*}
    (\{h_e\}_e,\{h'_e\}_e)\mapsto \{h_eh'_e = (hh')_e\}_e
\end{equation*}
which fuses the decorations on each edge $e\in\Gamma^1$. Pulling back yields a coproduct on the function algebra $C(G^\Gamma)$ --- however, we wish to make this coproduct sensitive to the composition of the holonomies, as well as the geometry of $\Gamma^1$. 

We do this through the following construction. Let $\psi_e$ denote the \textit{local graph operators} $\psi_e(\{h_{e'}\}_{e'}) = h_e$, which outputs the (matrix elements of the) $G$-holonomy at $e$. Suppose the edges $e_1,e_2\in \Gamma$ are composable, such that they can be attached into another edge $e=e_1\cup e_2\in\Gamma$, then we can write in Sweedler notation
\begin{equation}
    (\Delta\psi_e) = \sum\delta_{t(e_1),s(e_2)}(\psi_{(1)})_{e_1}\otimes (\psi_{(2)})_{e_2},\qquad e=e_1\cup e_2,\label{partialcoprod}
\end{equation}
where $s(e),t(e)$ denote the source and target vertices of an edge $e$, and $\psi_{(1)},\psi_{(2)}$ are functions on $G$ for which
\begin{equation*}
    \sum \psi_{(1)}(g)\psi_{(2)}(h) = \psi(gh),\qquad g,h\in G,.
\end{equation*}
By the groupoid structure of the graph complex $\Gamma = \Gamma^1\rightrightarrows \Gamma^0$, the formula \eqref{partialcoprod} extends to a map\footnote{Note that for each edge $e=v\to v'\in \Gamma^1$, it can be viewed as a "composite" with respect to the identity, constant edge, $e=e\cup 1_{v}=1_{v'}\cup e$. For more details, see \S \ref{2graphdef}.}
\begin{equation*}
\Delta: C(G^{\Gamma}) \to C(G^\Gamma)\,\bar\otimes\, C(G^\Gamma)\,,
\end{equation*}
where $\bar\otimes$ denotes the completed topological tensor product \cite{Woronowicz1988}.

\begin{rmk}\label{truecoproduct}
    Consider the following two cases: (i) $\Gamma$ denotes a graph complex $\Gamma$ composed of a single edge $e$ starting and ending on a single vertex. In this case, $C(G^\Gamma)=C(G)$ is the usual $C^*$-algebra. (ii) we take a colimit over refinements of the graph $\Gamma$, such that for every edge $e\in\Gamma^1$ there exists a graph $\Gamma'$ and an embedding $\Gamma\hookrightarrow \Gamma'$ such that $e = e'_1\cup e_2'$ for some $e_1',e_2'\in\Gamma';$ in this case $\lim_\Gamma C(G^\Gamma)$ approaches the character variety of $G$; see also \textit{Remark \ref{fockrosly}}.
\end{rmk}

A counit for this coproduct can be seen to be clearly given by the trivial decorated graph $\epsilon(\psi) = \psi(\{1_e\}_e)$. This geometric interpretation for $\Delta$ also endows $C(G^\Gamma)$ with an antipode $S:C(G^\Gamma)\rightarrow C(G^\Gamma)$, given in the classical case by
\begin{equation*}
    (S\psi)(\{h_e\}_e) = \psi(\{h_e^{-1}\}_e) = \psi(\{h_{\bar e}\}_e),
\end{equation*}
which can be interpreted as an orientation reversal operation, such that the usual coalgebra axioms
\begin{equation*}
     (S\otimes 1)\circ\Delta =\epsilon = (1\otimes S)\circ\Delta
\end{equation*}
are satisfied. This gives the function algebra $C(G^\Gamma)$ the structure of a Hopf algebra.


We will then introduce a quantum deformation of the above structures from the data of the Chern-Simons action. As is well-known, these data consist of a Lie algebra cocycle $\psi$ and the associated classical $r$-matrix $r\in\g\otimes\g$ on the Lie algebra $\g=\operatorname{Lie}G$ \cite{WITTEN1990285,Dupuis:2020ndx,Meusburger:2021cxe}.\footnote{The skew-symmetric part of $r$ arises from the $A^3$ interaction term, and the symmetric part arises from the canonical symplectic form $\omega(A,A) = \int_\Sigma\langle\delta A,\delta\partial_tA\rangle$ on the moduli space of flat $G$-connections.}  The skew-symmetric part of $r$ equips $G$ with a Poisson bracket \cite{Semenov1992,Meusburger:2021cxe}, which specifies a quantum deformation of the product on $C(G)$ along $q\sim e^{i\hbar}$ where $\hbar =\frac{2\pi}{k}$. This gives $C_q(G)$ a quasitriangular Hopf algebra structure, called the \textit{quantum coordinate ring} \cite{Woronowicz1988,Majid:1996kd,Grabowski1995}. 

Here, we wish to introduce this quantum deformation to the configuration space $C(G^\Gamma)$, hence we need to extend the semiclassical Poisson bracket onto the graph $\Gamma$. We will formalize this directly from the geometry and intersections of edges in $\Sigma$.

\subsection{Quantum deformation on the lattice}
Written as above, the coproduct is the cocommutative one in $C(G)$ up to the orientation of the glued edges $e,e'$. By combining this coproduct with the Poisson bracket extracted from the Chern-Simons action, we arrive at the \textit{combinatorial Poisson bracket} on $C_q(G^\Gamma)$.

Explicitly on local graph operators $\psi_e$, this Poisson bracket takes the same form as that given in \cite{Alekseev:1994pa},
    \begin{equation}
        \{\psi_e,\psi_{e'}\}_\text{dis} =\frac{2\pi}{k}(\delta_{t(e),s(e')} r\psi_e \otimes\psi_{e'} - \delta_{s(e),t(e')}\psi_e \otimes \psi_{e'} r^T) \equiv \mu\big([r,\Delta\psi_{e\cup e'}]_c\big),\label{heis-bracket}
    \end{equation}
where $e\cup e'$ is a composite edge. Here, $\mu$ is the product on $C(G^\Gamma)$ and $[-,-]_c$ is the commutator. 

The full quantum $R$-matrix $R\in C_q(G^\Gamma)\hat\otimes C_q(G^\Gamma)$, for which $R\sim 1 + i\hbar r+\dots$ admits an expansion as a power series in $\hbar$, then gives rise to the $q$-deformed product $\star$, whose $\star$-commutator can be expressed in the form
\begin{equation}
    \psi_e\star\psi_{e'}-\psi_e\star\psi_{e'} = \mu\big([R,\Delta\psi_{e\cup e'}]_c\big).\label{heis-product}
\end{equation}
In the context of compact quantum groups, these expressions \eqref{heis-bracket}, \eqref{heis-product} for the Poisson bracket and the quantum product has also appeared previously in the literature \cite{Woronowicz1988,Majid:1996kd,Dupuis:2020ndx,Meusburger:2021cxe}.

\medskip

The coproduct compatible with $\star$, which we shall also denote by $\Delta$, then satisfies the following intertwining relation \cite{Grabowski1995}
\begin{equation*}
    R \Delta\psi  = (\sigma\circ\Delta)\psi R,\qquad \psi\in C_q(G^\Gamma),
\end{equation*}
where $\sigma: C_q(G^\Gamma)\,\bar\otimes\, C_q(G^\Gamma)\rightarrow C_q(G^\Gamma)\,\bar\otimes \,C_q(G^\Gamma)$ is a swap of (topological) tensor factors. For a quantum double (cf. \cite{Delcamp:2016yix}), this leads to the definition of Kitaev ribbon operators, in which the graph $\Gamma$ is "thickened" in order to keep track of the actions of $R,R^T$.

\begin{rmk}\label{fockrosly}
    If one takes a colimit over refinements of the graphs, then the discrete holonomies $G^\Gamma$ modulo gauge transformations $G^{\Gamma^0}$ "approaches" the character variety $\operatorname{Ch}(G)=\operatorname{Hom}(\pi_1\Sigma,G)//G$, and the Poisson bracket $\{-,-\}_\text{dis}$ approaches the canonical Fock-Rosly one \cite{Fock:1998nu} on functions of $\operatorname{Ch}(G)$ arising from Chern-Simons theory. These can be made more precise, but we are not concerned with this issue at the time.
\end{rmk}


When the function algebra $C(G)$ inherits a canonical Poisson structure from the symplectic manifold $T^*G\cong \g^*\rtimes G$, then our above prescription underlies the combinatorial quantization of 3d BF theory and the spin-networks construction \cite{Baez:1994hx}. The case for Chern-Simons theory, on the other hand, directly makes use of the Hopf $C^*$-algebra $H^\Gamma$ generated by the graph holonomy operators, where the quasitriangularity structure arises from the quantum $R$-matrix. 

This $C^*$-algebra $(H^\Gamma,\ast,\Delta)$, together with its gauge transformations $G^{\Gamma^0}$, is the main player in \cite{Alekseev:1994pa}; with both taken together, it is called the "lattice algebra $\mathcal{B}^\Gamma$ of Chern-Simons theory". Hence their categorical analogues will be the star of this paper. The space of Wilson line observables extracted out of $\mathcal{B}^\Gamma$ is the main ingredient in the computation of 3-simplex scattering amplitudes in lattice Chern-Simons theory \cite{Alekseev:1994au}; this will play a more prominent role in a future work.

\subsection{A coherent formulation of graph operators}\label{coherence}
In order to lift the above formulation to 2-groups and the categorical setting, we require a "coherent" version of the story. Toward this, we will treat the (1-truncated) graph complex $\Gamma = \Gamma^1\rightrightarrows\Gamma^0$ as a groupoid, equipped with structure maps $s,t: (01)\mapsto \{0,1\}$ sending an edge to its endpoints. A decorated graph $G^\Gamma$ is then equivalent to a functor $F: \Gamma \rightarrow BG$ between groupoids, where $BG = G\rightrightarrows \ast$ is the pointed Lie groupoid with 1-morphisms labelled by $G$. 

Indeed, a functor $F$ specifies an assignment of the trivial point $\ast$ to a point $0\in\Gamma^0$, and a group element $h_e\in G$ to an edge $e\in\Gamma^1$. By thinking of the graph complex $\Gamma$ as the 1-truncation of a cell complex on $\Sigma$, we then see that all 2-graphs are assigned the identity. This enforces the flatness condition $h_{(01)}h_{(12)} = h_{(02)}$ for any ordered 2-simplex $(012)$ (or, more generally, around any closed face).

Thence, a natural transformation $\eta:F\Rightarrow F'$ assigns a group element $\eta_0 = a_0\in G$ to a vertex of the graph, such that the naturality condition implies $$h_{(01)}' = F'(01) = a_0^{-1}F(01)a_1=a_0^{-1}h_{(01)}a_1,$$ which is precisely a gauge transformation. In other words, the functor category $\operatorname{Fun}_\mathsf{Grpd}(\Gamma,BG)$ has objects decorated graphs $G^\Gamma$ and 1-morphisms the gauge transformations. This functor category itself forms a groupoid, since all gauge transformations are invertible. We shall without loss of generality denote by this functor category $\operatorname{Fun}_\mathsf{Grpd}(\Gamma,BG)$ by $G^\Gamma$. 

graph operators are therefore given by another functorial construction, $G^\Gamma\rightarrow \bbC$, where we consider $\bbC$ as a trivial category with only identity endomorphisms and no nonzero non-endomorphisms (ie. the discrete category on $\bbC$). The functor category $\operatorname{Fun}(G^\Gamma,\bbC)$ is 0-truncated; we think of the collection $\cA$ of the objects in $\operatorname{Fun}(G^\Gamma,\bbC)$ as the $q$-deformed $C^*$-algebra $C_q(G^\Gamma)$ described above, and its morphisms as the quantum gauge transformations. In this way, we see that $\cA=C_q(G^\Gamma)$ admits an action by the group
\begin{equation*}
    G^{\Gamma^0}= \coprod_{F,F'}\operatorname{Hom}_{G^\Gamma}(F,F'),
\end{equation*}
or more precisely the Hopf algebra generated by $G^{\Gamma^0}$, formed by the hom-sets of $G^\Gamma$ via pre-composition. Invariant states/observables can therefore be defined as the \textit{equivariantization} $\operatorname{Fun}(G^\Gamma,\bbC)^{G^{\Gamma^0}}$ --- namely taking homotopy fixed points then truncating. The induced essential surjection $\cA\rightarrow \cA^{\Gamma^0}$ is given precisely by the Haar integration/group averaging over $G^{\Gamma^0}$.

\subsection{Definition of a 2-graph}\label{2graphdef}
Motivated by the above groupoid description of a graph, we now introduce a way in which a graph complex --- together with the data of polygonal faces --- can be seen as a 2-groupoid. Given a graph complex $\Gamma$, corresponding to an underlying graph dual to the triangulation of a 3-manifold $\Sigma$, for instance, we perform a 2-truncation instead of a 1-truncation. 

The resulting discrete 2-gorupoid, denoted by $\Gamma$, has the following data.
\begin{itemize}
    \item sets of vertices, edges and polygonal faces $\Gamma^0,\Gamma^1,\Gamma^2$, with each face $f\in\Gamma^2$ equipped with a choice of a boundary \textit{root edge} $e_\ast$,
    \item maps $s_1,t_1: \Gamma^1\to \Gamma^0$ which sends an edge to its endpoints, and $s_2,t_2: \Gamma^2\to \Gamma^1$ sending a face to its boundary, comprised of its root edge $e_\ast=s_2(f)$ as the source and the rest as "target", such that
    \begin{equation}
        s_1\circ s_2 = s_1\circ t_2,\qquad t_1\circ s_2 = t_1\circ t_2\,.\label{2st}
    \end{equation}
    \item $\Gamma^1\rightrightarrows \Gamma^0$ forms the usual graph groupoid under the composition of edges $\cup: \Gamma^1 \times_{\Gamma^0}\Gamma^1\to \Gamma^1\,$ with the "constant edge" at $v\in \Gamma^0$ serving as the identity $1_v: v\to v$,
    \item $\Gamma^2$ has equipped a groupoid structure given by the composition $\cup_v: \Gamma^2 \times_{\Gamma^1}\Gamma^2\to \Gamma^2\,,$ which attaches two faces $f,f'\in\Gamma^2$ along the identified edge $t_2(f)=s_2(f')$ "vertically", with the "constant PL homotopy" $\id_f: e\Rightarrow e$ serving as the 2-gorupoid identity,
    \item the inverses corresponding to the composition laws $\cup_v,\cup_h$ are given by the orientation reversals of edges and faces, $\bar e= e^{-1},~ \bar f= f^{-1}.$
\end{itemize}
Crucially, the set $\Gamma^2$ of polygonal faces has \textit{another} composition law: one which is induced by the attaching along the edges. This is defined as a map
\begin{equation*}
    \cup_h:\Gamma_2\, \times_{\Gamma^0} \Gamma_2\to\Gamma_2
\end{equation*}
where the pullback over $\Gamma^0$ is given in terms of the maps $\Gamma^2\to \Gamma^0$  given in \eqref{2st}. 

Aside from the obvious associativity and unity laws satisfied by the various composition and identities, we have the following crucial consistency condition: the \textbf{interchange law}. Namely, for any collection of four faces $f_1,f_2,f_3,f_4\in\Gamma^2$, we must have the equation
\begin{equation}\label{graphinterchange}
    (f_1\cup_h f_2)\cup_v(f_3\cup_hf_4) = (f_1\cup_vf_3)\cup_h(f_2\cup_vf_4)
\end{equation}
whenever the compositions on both sides make sense. These are typical structures and coherence conditions satisfied by a 2-groupoid. 


We call $\Gamma = \Gamma^2\rightrightarrows\Gamma^1\rightrightarrows \Gamma^0$ a \textbf{2-graph complex}. For the rest of this paper, by a "lattice" we will simply mean the 2-graph complex $\Gamma$ embedded in $\Sigma$, as well as its underlying 2-groupoid structure, defined in \S \ref{2graphdef}.

\section{2-Chern-Simons theory on the lattice}\label{2cs}
Let $\mathbb{G}=\mathsf{H}\xrightarrow{t}G$ denote a strict Lie 2-group. We will assume the associated Lie 2-algebra $\frak G =\operatorname{Lie} \mathbb{G}=\h\xrightarrow{t}\g$ is \textit{balanced} (terminology from \cite{Zucchini:2021bnn}): namely it has equipped a non-degenerate invariant pairing $\langle-,-\rangle: \frak{G}^{\otimes 2}\rightarrow\bbC[1]$ of \textit{degree-1}. In other words, $\langle-,-\rangle$ is only supported on $\g\otimes\h \oplus \h\otimes \g$.

\begin{rmk}\label{smooth2grp}
    Here, by "strict" we mean that the associator and unitor morphisms (as one typically see in monoidal categories \cite{maclane:71}) are trivial. Several equivalent \cite{Baez:2004,Crane:2003gk,Porst2008Strict2A,Pfeiffer2007,Ludewig:2023} descriptions of Lie 2-groups that we shall make use of here are (i) a category internal to the category of Lie groups $\mathsf{LieGrp}$, (ii) a Lie group crossed-module $\mathbb{G} = \mathsf{H}\xrightarrow{t}G$, and (iii) a 2-group object in the category of Lie groupoids $\mathsf{LieGrpd}$. These all have "strictness" built-in, and make it clear that Lie 2-groups $\mathbb{G}$ come with a smooth topology. In order to describe a weak variant of Lie 2-groups, on the other hand, one considers 2-group objects in the bicategory of bibundles $\mathsf{Bibun}$, instead of $\mathsf{LieGrpd}$: this is a {\bf smooth 2-group} \cite{Schommer_Pries_2011}, in which the associator and unitor morphisms can be weakened in the smooth setting. 
\end{rmk}

Given a 4-manifold $M^4$, the partition function is given formally by
\begin{equation*}
    Z(M^4) = \int D[A,B] e^{i2\pi k \operatorname{2CS}(A,B)} = \int D[A,B]e^{i2\pi k\int_{M^4}\langle B,F(A)-\frac{1}{2}tB\rangle},
\end{equation*}
where $(A,B)\in\Omega^1\otimes \g\oplus \Omega^2\otimes\h$ is a $\mathbb{G}$-connection on $M^4$. The classical equations of motion are given by fake- and 2-flatness
\begin{equation*}
    F(A) -tB = 0,\qquad d_AB = 0,
\end{equation*}
and the gauge symmetries are parameterized by a polyform $(h,\Gamma)\in C^\infty\otimes G\oplus\Omega^1\otimes\h$ such that
\begin{equation*}
    A\mapsto A^{(h,\Gamma)}= \operatorname{Ad}_g^{-1}A + g^{-1}dg + t\Gamma,\qquad B\mapsto B^{(h,\Gamma)}=g^{-1}\rhd B + d_{A^g}\Gamma-\frac{1}{2}[\Gamma,\Gamma].
\end{equation*}
This in particular reproduces the 4d BF theory when $t=0$ and the 4d BF-BB theory when $t=1$. In the latter case, the shift symmetry $A\mapsto A +t\Gamma $ has been {\it gauged} by this derived formalism.

In the following, we will demonstrate the raison d'{\^ e}tre behind the coherent formulation in \S \ref{coherence}: one can put a "2-" in front of every appropriate noun, and obtain a description of 2-Chern-Simons theory.

\subsection{Discrete 2-gauge theory}\label{disc2hol}
Now the story of trying to discretize this theory tentatively goes in the same way as in the ordinary Chern-Simons case: given a triangulation of the 3d Cauchy surface $\Sigma$ of $M^4$, we (i) decorate the dual 2-graph $\Gamma$ (cf. \S \ref{2graphdef}) with the data of $\mathbb{G}$, (ii) define functions on them, and then (iii) mod out the (2-)gauge transformations. We shall once again assume our graphs (which now contains edges and faces) are non-self-intersecting.

To make this precise, the coherent formulation of graph operators becomes very useful: 
\begin{definition}
    An (admissible) \textbf{decorated 2-graph} is a 2-functor $F: \Gamma \rightarrow \textbf{B}\mathbb{G}$ between the 2-graph complex $\Gamma$, treated as a 2-groupoid, to the algebraic delooping $\textbf{B}\mathbb{G} = \mathbb{G}\rightrightarrows\ast$. For an oriented 2-simplex $(012)$, for instance, this is the data of 
    \begin{equation*}
        F(i) = \ast, \qquad F(ij) = h_{(ij)}\in G,\qquad F_{(012)} = b_{(012)} \in\mathsf{H}\,.
    \end{equation*}
    For each closed polygonal face $f = (e_*,e_1,\dots,e_p)\in\Gamma^2$ with $p$-number of edges, we have the \textit{fake-flatness} condition \cite{Bullivant:2017sjz}
    \begin{equation*}
        \prod_{i=1}^p h_{e_i} = h_{e_*}t(b_f),
    \end{equation*}
    where $e_*\in \partial f$ is the distinguished {root edge} of $f$. In other words, the edges are glued together according to the 2-groupoid structure of $\Gamma$ given in \S \ref{2graphdef}.
\end{definition}
\noindent Since $\Gamma$ is by construction 2-truncated, $F$ assigns the trivial value $1$ to a contractible 3-cell:
\begin{equation*}
    \prod_{f\in \partial V} b_{f} =1,\qquad V \text{ contractible 3-cell}.
\end{equation*}
This gives the \textit{2-flatness} kinematical condition \cite{Martins:2006hx,Bullivant:2016clk} on the decorated 2-graph.

Notice if $\Gamma$ is a "fundamental 2-graph", ie. $\Gamma^2$ is a PL 2-disc consisting of a single face with a single 1-graph boundary, then $\mathbb{G}^{\Gamma}$ is a single copy of $\mathbb{G}$. As an abuse of notation, we will denote by $\mathbb{G}^\Gamma$ the 2-groupoid of discrete $\mathbb{G}$-holonomies in the following.

\begin{rmk}
    It is important to emphasize here that the fake- and 2-flatness conditions are imposed kinematically as Gauss constraints on the states, while the dynamical constraint defined by the delta operators in the scattering amplitude involve the discretized versions of the 2-Bianchi identities $$d_A(F-tB)=0,\qquad d_A(d_AB)=0.$$ This is true for all values of the $t$-map; in particular, for $t=\id$ the 1-Bianchi identity $d_AF=0$ that appears on-shell $F=B$ is in fact a {\it kinematical} 2-flatness condition, and hence is {not} part of the dynamical constraint on scattering amplitudes.
\end{rmk}

A {\it pseudo}natural transformation $\eta:F\Rightarrow F'$ assigns an element of $G$ to a vertex, and an element of $\mathsf{H}\rtimes G$ to an edge, such that several diagrams commute. Working this all out gives.
\begin{definition}
    A \textbf{2-gauge transformation} between two decorated 2-graphs is a pseudonatural transformation $\eta:F\Rightarrow F'$. On an oriented 1-simplex $(01)$, for instance, this is the data of
    \begin{equation*}
        \eta_i = a_i\in G,\qquad \eta_{(01)} = \gamma_{(01)}\in \mathsf{H}\,.
    \end{equation*}
    On every oriented face $f\in\Gamma^2$ rooted at the source edge $e_\ast: v_0\to v_1$, we have
    \begin{equation}
        h_{e_\ast}' = a_{v_0}^{-1}h_{e_\ast} t(\gamma_{e_\ast})a_{v_1},\qquad b_{f}' = a_{v_0}^{-1}\rhd((h_{e_\ast}\rhd \gamma_{s_s})^{-1}(a_{v_1}\rhd b_{f})\gamma_{e_\ast})\,.\label{2gt}
    \end{equation}
\end{definition}
\noindent In more compact notation, following the blob model of bicategories in \cite{Baez:2004}, these are actions by conjugation
\begin{equation*}
    h'_{e}\xrightarrow{b_{f}'} = (a_{v_0}\xrightarrow{\gamma_{e}})^{-1}\cdot (h_{e}\xrightarrow{b_{f}})\cdot (a_{v_1}\xrightarrow{\gamma_{e}}) \equiv \operatorname{hAd}_{(a,\gamma)_e}^{-1}(h_{e},b_{f}),
\end{equation*}
under the \textit{horizontal} composition $\cdot$ in $\mathbb{G}^\Gamma$. Note the target of $a\xrightarrow{\gamma}$ is determined by $at(\gamma)$, so we did not write them down.

We are not done yet. For pseudonatural transformations $\eta,\eta':F\Rightarrow F'$ between 2-functors, there is the notion of {\it modifications} $m: \eta\Rrightarrow\eta'$. This defines the notion of {\bf secondary gauge transformations} or \textbf{ghosts-of-ghosts}, ie. redundancy between 2-gauge transformations. This is the data of an element of $\mathsf{H}\rtimes G$ on each vertex $v$ such that 
\begin{equation*}
    a_{v}' = a_vt(m_v),\qquad \gamma_{e}' = m_{v_0}^{-1}\gamma_{e}m_{v_1}
\end{equation*}
for each edge $e: v_0\to v_1 $ in $\Gamma^1$. In other words, this is the conjugation action
\begin{equation*}
    a'_{v_0}\xrightarrow{\gamma'_{e}} \,= (a_{v_0}\xrightarrow{m_{v_0}} )^{-1}\circ (a_{v_0}\xrightarrow{\gamma_{(01)}}) \circ (a_{v_1}\xrightarrow{m_{v_1}} ) \equiv \operatorname{vAd}_{m_0}^{-1}(a_{v_0},\gamma_{e})
\end{equation*}
under \textit{vertical} composition $\circ$ in $\mathbb{G}^\Gamma$. This describes fully the 2-categorical structure determined by the 2-functors $\operatorname{Fun}(\Gamma,\textbf{B}\mathbb{G})$.

As an abuse of notation, we will denote by $\mathbb{G}^{\Gamma^1}$ the monoidal categories formed by the pseudonatural natural transformations and the secondary gauge transformations --- namely, they are the 2- and 1-morphisms of $\operatorname{Fun}(\Gamma,\textbf{B}\mathbb{G})$.

\medskip

    Note the strictness of the underlying Lie 2-group $\mathbb{G}$ here means that the hom-categories in $\mathbb{G}^\Gamma$ are strict monoidal, and hence we can truncate them and treat $\mathbb{G}^\Gamma$ as a 1-groupoid. Once we have done this, the 1-morphisms in it are then labelled by secondary-gauge equivalence classes of 2-gauge transformations. We shall see that this dramatically simplifies much of our discussions in the next sections.

\subsubsection{Weak 2-gauge theory based on weakly-associative smooth 2-groups}\label{weak2gauthy} We pause here to make several comments about the weakly associative setting. It is well-known that \textit{finite} 2-groups are classified up to equivalence by its {\it Ho{\` a}ng data} $(G=\operatorname{coker}t,A=\operatorname{ker}t,\tau)$ \cite{Nguyen2014CROSSEDMA,Ang2018,Baez2023HoangXS}, where $\tau\in H^3(N,A)$ is a group 3-cohomology class called the \textbf{Postnikov class}.  In this context, the 3-cocycle $\tau$ can be thought of as an associator isomorphism $\tau(g_1,g_2,g_3)\in A$ over $g_1g_2g_3\in G$, which only has a component proportional to the identity \cite{Wen:2019,Cui_2017,Chen:2023tjf}. 

There had been numerous works in the literature which studied higher-group gauge theories built from such Ho{\` a}ng data, and they led to the so-called {\it 2-group Dijkgraaf-Witten TQFTs} \cite{Zhu:2019,Kapustin:2013uxa,Kapustin2017,Bullivant:2019tbp,Bullivant:2017sjz,Chen2z:2023}. These can be understood as 4d Douglas-Reutter TQFTs \cite{Douglas:2018} built out of the symmetric 2-category $\operatorname{2Rep}(G,A,\tau)$ of the 2-representations \cite{Bartsch:2022mpm} of the 2-group.

\medskip

To describe associators $\tau(g_1,g_2,g_3): (g_1g_2)g_3\rightarrow g_1(g_2g_3)$ with non-identity components, we must work in the context of weakly associative smooth 2-groups. Here, the objects $G$ of $\mathbb{G}$ no longer form groups, as its monoidal structure is no longer associative $ (g_1g_2)g_3\neq g_1(g_2g_3)$. Although the associativity of the morphisms is retained, its composition law is modified: namely for $(g_1,\alpha_1)\circ(g_2,\alpha_2)\circ (g_3,\alpha_3)$ to be composable, we require
\begin{equation*}
     \tau(g,t(\alpha_1),t(\alpha_2)): (g_1t(\alpha_1))t(\alpha_2)\rightarrow g_1(t(\alpha_1\alpha_2)) = g_1(t(\alpha_1)t(\alpha_2)).
\end{equation*}
Despite the abundance of literature on finite 2-group Dijkgraaf-Witten theories, and despite the fact that we do have the proper setting of \textit{smooth 2-groups} \cite{Schommer_Pries_2011} mentioned previously to talk about smooth associator morphisms, the "weak 2-gauge theory" built out of such smooth 2-groups are much less well-understood. Though, the form of the action is known \cite{Jurco:2018sby,Chen:2024axr,Schenkel:2024dcd},
\begin{equation*}
        S_{w2CS}[A,B] = \int_X \langle B,F(A)-\frac{1}{2}tB\rangle + \frac{1}{4!}\langle \kappa(A),A\rangle,
\end{equation*}
as well as the local kinematical data \cite{Kim:2019owc,Zucchini:2021bnn,Soncini:2014} (the so-called "weak 2-connections" and their 2-gauge transformations). These fields are described by the structure of a weak Lie 2-algebra $\G=\h\xrightarrow{t}\g$ with a {\it Jacobiator} $\kappa: \g^{\times 3}\rightarrow \h$ \cite{Chen:2012gz,Chen:2013,BaezRogers}, which sources the covariant 2-curvature \cite{Soncini:2014,Baez:2004in,Cordova:2018cvg,Benini_2019},
\begin{equation*}
    d_AB - \frac{1}{3!}\kappa(A,A,A) = 0.
\end{equation*}

\subsubsection{Closure of the weak 2-gauge algebra}\label{descendant}
An issue encountered in weak 2-Chern-Simons theory is that the 2-gauge algebra, in the weakened context, does not close off-shell of the fake-flatness condition \cite{Kim:2019owc,Soncini:2014}. In our combinatorial setting, on the other hand, we can directly take the effects of the associator $\tau$ into account.

Recall that the 1- and 2-morphisms in the 2-functor 2-category $\operatorname{Fun}(\Gamma,\textbf{B}\mathbb{G})$ are given by the pseudonatural transformations and the modifications which govern 2-gauge transformations on the lattice. By direct computation, one can see that the monoidality of the composition of the 2-gauge transformation \eqref{2gt} is witnessed by an invertible modification,
  \begin{equation*}
      m_v(a,h):\operatorname{hAd}^{-1}_{(a_1,\gamma_1)}\circ\operatorname{hAd}^{-1}_{(a_2,\gamma_2)} \Rightarrow \operatorname{hAd}^{-1}_{(a_1a_2,\gamma_1(a_1\rhd\gamma_2))},
  \end{equation*}
given by a vertical conjugation $\operatorname{vAd}_{\tau(a_1,a_2,h)}^{-1}$, where $(h,b)$ denotes the source of the 2-gauge transformation $(a_2,\gamma_2)$. In terms of field theory, this is known as the \textit{first descendant} of $\tau$ \cite{Kapustin:2013uxa,Chen:2022hct,Kim:2019owc}, which in the continuum defines a 2-cocycle which depends on the holonomy $h_e$, as well as the gauge parameter $a_v$.

If such modifications have non-identity components, which is indeed the case when $\mathbb{G}$ has equipped a \textit{weak} associator $\tau$, then one cannot 2-truncate $\mathbb{G}^\Gamma$ to obtain an algebra in the usual sense. This explains why the weak 2-gauge symmetries in general does not close \textit{as algebras} --- {the higher-gauge symmetries form monoidal categories in general}! In the continuum, this fact was also noted in \cite{Soncini:2014,Zucchini:2021bnn} from the  BRST perspective.



\subsection{Configuration space of lattice 2-Chern-Simons theory}\label{configs}
We now work to construct the "2-algebra of 2-graph operators" by making use of the coherent formulation. One is then tempted to study the functions on the $\mathbb{G}$-holonomies, denoted $C(\mathbb{G}^{\Gamma})$. However, as we will explain in \S \ref{coprod}, this formulation suffers from various mathematical and physical drawbacks.

As such, we instead consider a higher categorical notion of $\mathbb{C}$ --- namely the category $\mathsf{Hilb}$ of complex vector (Hilbert) spaces. The goal is to study maps which assigns a \textit{vector space} to a $\mathbb{G}$-holonomy, and a linear isomorphism to a (secondary gauge equivalence class) of the 2-gauge transformation on them. 


\begin{definition}\label{2graphopsguide}
    A \textbf{2-graph operator with covariance data} is the tuple $\Phi=(\phi,\varphi)$ consisting of the following data:
    \begin{enumerate}
        \item a map $\phi:\mathbb{G}^\Gamma\to\mathsf{Hilb}$ that assigns to each decorated 2-graph $\{(h_e,b_f)\}_{(e,f)}\in\mathbb{G}^{\Gamma}$ a Hilbert space $\phi_{\{(h_e,b_f)\}_{(e,f)}}\in\mathsf{Hilb}$, and
        \item a map $\varphi$ which assigns to each (secondary gauge equivalence class of) 2-gauge transformation $\eta:\{(h_e,b_f)\}_{(e,f)}\rightarrow \{(h_e',b_f')\}_{(e,f)}$ a linear isomorphism $\varphi_\eta=\Lambda:\phi_{\{(h_e,b_f)\}_{(e,f)}}\xrightarrow{\sim}\phi_{\{(h_e,b_f)\}_{(e,f)}}'=\phi_{\{(h'_e,b'_f)\}_{(e,f)}}$  defined by \eqref{2gt}.
    \end{enumerate}
    A \textbf{morphism} between 2-graph operators with covariance data is an assignment of linear maps $\psi:\phi_{\{(h_e,b_f)\}_{(e,f)}}\to \tilde{\phi}_{\{(h_e,b_f)\}_{(e,f)}}$ to each decorated 2-graph, which intertwines the 2-gauge transformations; formally, we have $\psi\circ    \varphi = \tilde{\varphi} \circ\psi$.
\end{definition}
\noindent



\medskip

The above definition pins down, at least algebraically, the properties that 2-graph operators should have. But since we are dealing with \textit{Lie} 2-groups and \textit{infinite-dimensional} Hilbert spaces, certain functional analytic conditions must also be given. As such, the above definition should be treated as more of a "guide" towards the actual definition of 2-graph operators.

In the following section, we shall recall the framework which properly treats analytic properties.

\subsubsection{Square-integrable functors: measureable fields}\label{measureablefields}
Similar to the case of the discretized Chern-Simons theory, the 2-graph operators should form an "infinite-dimensional 2-algebra". For this, it is useful to consider \textbf{measureable categories} of Crane-Yetter \cite{Crane:2003ep,Yetter2003MeasurableC,Baez:2012}.
\begin{definition}
    A \textbf{measurable category} $\mathcal{H}^X$ is a $C^*$-category with the following.
    \begin{itemize}
        \item The objects are measurable fields $H^X$, which is the data of a measure space $(X,\mu)$ together with an assignment $x\mapsto H_x$ of (infinite-dimensional) Hilbert spaces $(H_x,\langle-,-\rangle_x)$ for each $x\in X$ such that one has a subspace $\cM_H\subset \coprod_x H_x$ defined by:
    \begin{enumerate}
        \item the norm function $x\mapsto |\xi_x|_{H_x}=\sqrt{\langle \xi_x,\xi_x\rangle_{H_x}}$ is $\mu$-measurable,
        \item if $\eta\in\coprod_xH_x$ is such that $x\mapsto \langle \eta_x,\xi_x\rangle_{H_x}$ is $\mu$-measurable for all $\xi\in\cM_H$, then $\eta\in\cM_H$,
        \item there exists a sequence $\{\xi_i\}\subset\cM_H$ that $\{(\xi_i)_x\}_i\subset H_x$ is dense for all $x$.
    \end{enumerate}
    \item A morphism $f: H^X\rightarrow H'^X$ is a $X$-family $f_x: H_x\rightarrow H'_x$ of bounded linear operators such that $f(\cM_H)\subset\cM_{H'}$.
    \end{itemize}
\end{definition}
Just as in the case of {\it finite}-dimensional 2-Hilbert spaces \cite{Baez1996HigherDimensionalAI}, a measureable field $H^X$ has hom's given by a $C^*$-algebra of bounded linear operators. But here, these $C^*$-algebras are indexed by a measure space $X$, instead of just a finite set of basis elements. The collection of all measurable categories form a 2-category $\mathsf{Meas}$. We shall in the following cast $\mathfrak{C}(\mathbb{G}^{\Gamma})\in\mathsf{Meas}$ as a measurable category. 


There is an analogue of the integration operation for measureable fields \cite{Yetter2003MeasurableC,Baez:2012}.
\begin{proposition}
    The \textbf{direct integral} $\cH=\int_X^\oplus d\mu_x H_x$ is a functor $\int^\oplus_X d\mu_X(-): \cH^X\rightarrow\mathsf{Hilb}$.
\end{proposition}
\noindent The Hilbert space $\cH=\int_X^\oplus d\mu_x H_x$ associated to the measurable field $H^X\in\cH^X$ is defined as the space of $\mu$-a.e. equivalent classes of $L^2$-integrable sections $\psi\in\cM_H$ equipped with the inner product $$\langle\psi,\psi'\rangle = \int_X d\mu_x \langle\psi_x,\psi'_x\rangle_{H_x}<\infty.$$


\subsubsection{Haar measures on locally compact Lie 2-groups} 
Now in order to apply the notion of measureable fields to \textbf{Definition \ref{2graphopsguide}}, we must first make $\mathbb{G}$ into a measure space. We will do this by following \cite{Williams2015HaarSO}.

\begin{definition}
    Consider $\mathbb{G}$ as a locally compact Hausdorff groupoid. A \textbf{Haar system} on $\mathbb{G}$ is a $G$-family $\{\nu^a\mid a\in G\}$ of positive Radon measures $\nu^a$ supported on $\mathsf{H}$ (considered as the source fibre of $a\in G$) such that for all compactly supported $f\in C_c(\mathsf{H})$,
    \begin{enumerate}
        \item the assignment $a\mapsto \int_{\mathsf{H}}d\nu^a(\gamma) f(\gamma)$ is continuous, and
        \item for all $\gamma\in\mathbb{G}$, with source $a$ and target $t(\gamma)a$, then $\int_\mathsf{H}d\nu^{a}(\gamma')f(\gamma') = \int_\mathsf{H}d\nu^{t(\gamma)a}(\gamma')f(\gamma')$.
    \end{enumerate}
    Note the second condition implies the left-invariance of $\nu^a$ by groupoid (vertical) multiplication.
\end{definition}
\noindent Now let $\mathbb{G}$ be a locally compact Hausdorff Lie 2-group (ie. both $\mathsf{H},G$ are locally compact Hausdorff, and the structure maps $s,t$ are smooth). We require the Haar system $\{\nu^a\}_a$ on $\mathbb{G}$ to be compatible with the group structure. 

Explicitly, this means that for each $a\in G$ and each measureable subset $A\subset\mathsf{H}$, the map $\nu^b(A)\mapsto \nu^{ab}(A)$ is measureable. Take a usual Haar measure $\sigma$ on $G$, the following measure
\begin{equation*}
    d\mu_{(a,\gamma)} = d\sigma(a)d\nu^a(\gamma)
\end{equation*}
is then left-invariant under horizontal \textit{whiskering}: for each compactly-supported $\bbC$-valued $f\in C_c(\mathbb{G})$ and $a\in G$, we have
\begin{align*}
    \int_\mathbb{G}d\mu_{(ab,a\rhd \gamma)}f(b,\gamma) &= \int_G d\sigma(ab) \int_\mathsf{H} d\nu^{ab}(a\rhd\gamma)f(b,\gamma) \\
    &= \int_G d\sigma(b) \int_\mathsf{H} d\nu^{t(\gamma)b}(\gamma)f(a^{-1}b,a^{-1}\rhd \gamma) = \int_\mathbb{G}d\mu_{(b,\gamma)}f(a^{-1}b,a^{-1}\rhd\gamma),
\end{align*}
where we have made a change of variable $(ab,a\rhd \gamma)\mapsto (b,\gamma)$. This provides an invariant Haar measure on the 2-group $\mathbb{G}$. We now condense this notion into a proper definition.

First, notice the above condition implies that the Haar system $\Omega_\mathsf{H}= \{\nu^a\mid a\in G\}$, understood as a subspace of all measureable functions $\cM_\mathsf{H}$ on $\mathsf{H}$, is a measureable $G$-representation. We call such Haar systems {\it $G$-equivariant}.
\begin{definition}
    A \textbf{2-group Haar measure} $\mu$ is a Radon measure equipped with a \textit{disintegration} \cite{Pachl_1978,Baez:2012} $\{\nu^a\}_{a\in G}$ along the source map $s:\mathbb{G}\rightarrow G$ into a $G$-equivariant Haar system,
\begin{equation*}
    \int_\mathbb{G} d\mu(\zeta) f(\zeta) = \int_G d\sigma(a) \int_{s^{-1}G} d\nu_a(\gamma) f(a,\gamma),\qquad \forall~ f\in C(\mathbb{G}),~\zeta=(a,\gamma)\in\mathbb{G},
\end{equation*}
such that $\sigma = \mu\circ s^{-1}$ is itself a Haar-Radon measure on $G$.
\end{definition}
\noindent Recall the family of measures $\{\mu^a\}_a$ exist (as ordinary Radon measures) due to the disintegration theorem \cite{Pachl_1978}, while the $G$-equivariance is an extra algebraic condition.

In the following, we will always assume that the Haar measure $\mu$ is also Borel: namely that all $\mu$-measureable subsets are open in the smooth topology of $\mathbb{G}$

\subsubsection*{Volumes of Lie 2-groups.}
Let us now try to compute the volume $\mu(\mathbb{G})$ of the compact Lie 2-group $\mathbb{G}$. Formally, this can be understood as the sum of all the volumes $\operatorname{vol}_a(\mathsf{H})=\int_\mathsf{H}d\nu^a(\gamma)$ of $\mathsf{H}$ determined by the Haar system. To ensure that this definition is well-defined, we require the map $a\mapsto \operatorname{vol}_a(\mathsf{H})$ itself to be $\sigma$-measureable. This is implied precisely by the $G$-equivariance of the Haar system: the map $\nu^b(A)\mapsto \nu^{ab}(A)$ is measureable in $a\in G$. The map
\begin{equation*}
    \nu^1(\mathsf{H}) \mapsto \nu^a(\mathsf{H}) = \operatorname{vol}_a(\mathsf{H})
\end{equation*}
is thus measureable in $a\in G$, thanks to the compactness of $\mathsf{H}$. The continuity of multiplication $b\mapsto ab$ in $G$ then implies that the function $a\mapsto \operatorname{vol}_a(\mathsf{H})$ is $\sigma$-measureable,
and hence the {\it Haar volume}
\begin{equation*}
    \mu(\mathbb{G}) = \int_Gd\sigma(a)\operatorname{vol}_a(\mathsf{H})<\infty
\end{equation*}
is well-defined. 

This allows us to normalize the measure $\mu$ such that $\mu(\mathbb{G})=1$; or in other words, all integration operations of the form $\int_\mathbb{G}d\mu,\int_\mathbb{G}^\oplus d\mu$ comes with an implicit factor of $1/\mu(\mathbb{G})$. In the following, we will assume that all 2-group Haar measures are normalized in this way.

\subsection{Categorical 2-group functions as measurable fields}\label{2grpfunc}
With the Haar measure $\mu$ on $\mathbb{G}$ in hand, we can now construct a model for the "Hilbert space-valued 2-group functions" $\mathfrak{C}(\mathbb{G})$. In accordance with the datum $\phi$ in \textbf{Definition \ref{2graphopsguide}},  elements in it should be maps that assign entire 2-group elements to a Hilbert space. 

\begin{rmk}\label{not-2group}
    Note we are aiming for simply an assignment of a Hilbert space to each $\mathbb{G}$-decorated 2-graphs. This is in contrast to the formulation in terms of the functor category $\operatorname{Fun}(\mathbb{G},\mathsf{Hilb})$, as studied in \cite{huang2023tannaka}, as we do not a priori consider our assignments as functorial. The framework here is instead more closely related to the "2-groupoid algebra" described in \cite{Bullivant:2019tbp}. 
\end{rmk}

We begin by modelling elements $\phi\in\mathfrak{C}(\mathbb{G})$ as a measurable field $H^X$ over $X= (\mathbb{G},\mu)$; recall $\mu$ is assumed to be Borel. For each $(g,a)\in\mathbb{G}$, we assign a (finite-dimensional) Hilbert space 
\begin{equation*}
    H_{(g,a)}(\phi) = \phi(g,a),\qquad \langle -,-\rangle_{H_{(g,a)}} = \langle-,-\rangle_{\phi(g,a)},
\end{equation*}
called the \textit{stalk} at $(g,a)$, equipped with a fibrewise inner product. Next, we take the space $\cM_{H}\subset \coprod_{(g,a)} H_{(g,a)}$ of measurable sections to consist of vectors $\xi_{(g,a)}$ such that the norm map $\mathbb{G} \rightarrow \R:(g,a)\mapsto |\xi_{(g,a)}|_{H_{(g,a)}}$ is continuous (with respect to the smooth topology of $\mathbb{G}$).

Throughout the rest of this paper, we shall assume the existence of $\mu$-measureable covering $U\rightarrow \mathbb{G}$ by Borel open sets. We will restrict to a smaller, better behaved Hilbert fields in this paper (see also \cite{TRENTINAGLIA2010750}).
\begin{definition}\label{catfunctions}
    Let $X=(\mathbb{G},\mu)$ denote a compact Lie 2-group equipped with a Haar measure $\mu$, and consider the measureable category $\cH^X\in\mathsf{Meas}$ on it. The \textbf{categorical function algebra} $\C(\mathbb{G})\subset \cH^X$ on $\mathbb{G}$ is the full additive measureable subcategory consisting of measureable fields $H^X$ whose direct integral $$\Gamma_c(H^X): U\mapsto \left(\int_U^\oplus d\mu(g,a)H_{(g,a)}\right)$$ defines a \textit{quasicoherent sheaf}, ie. locally presentable projective $C^\infty(X)$-module, of Hilbert spaces.
\end{definition}
\noindent In other words, we consider $\C(\mathbb{G})$ to be the minimal Abelian completion of the Crane-Yetter measureable category $\cH^\mathbb{G}$ on $(\mathbb{G},\mu)$ within the bicategory $\mathsf{Meas}$.

Now assuming $\mathsf{Meas}$ has all pullbacks/pushouts of measureable categories, we can define the following.
\begin{definition}\label{coproducts}
    Let $C(\mathbb{G}^{\times 2}) \cong C(\mathbb{G})~\bar \otimes~ C(\mathbb{G})$ denote the completed topological tensor product. 
    \begin{enumerate}
        \item We define $\C(\mathbb{G})\times\C(\mathbb{G})$ as the category of measureable Hilbert $C(\mathbb{G})\bar \otimes C(\mathbb{G})$-modules. Together with the equivalence $\mathfrak{C}(\mathbb{G}^{\times 2})\simeq \C(\mathbb{G})\times\C(\mathbb{G})$ \cite{Yetter2003MeasurableC}, the map $ \Delta_h$ defines the additive \textbf{horizontal coproduct/coproduct functor}
    \begin{equation*}
    \Delta_h:\C(\mathbb{G})\rightarrow \C(\mathbb{G})\times\C(\mathbb{G}).
\end{equation*} 
        \item For $\mathbb{G}_1=(\mathsf{H}\rtimes G)$, define by $(\mathbb{G}_1\times_G\mathbb{G}_1,\mu\times_\sigma\mu)$ the pullback  measure space along the surjective submersive source/target maps $s,t: \mathbb{G}_1\rightarrow \mathbb{G}_0=G$. Define by $\C(\mathbb{G}_1)\times^{\C(G)}\C(\mathbb{G}_1)$ the measureable category on $(\mathbb{G}_1\times_G\mathbb{G}_1,\mu\times_\sigma\mu)$, such that it fits into the pushout square
\[\begin{tikzcd}
	{\C(G)} & {\C(\mathbb{G}_1)} \\
	{\C(\mathbb{G}_1)} & {\C(\mathbb{G}_1)\times^{\C(G)}\C(\mathbb{G}_1)}
	\arrow["{s^*}", from=1-1, to=1-2]
	\arrow["{t^*}"', from=1-1, to=2-1]
	\arrow["\lrcorner"{anchor=center, pos=0.125}, draw=none, from=1-1, to=2-2]
	\arrow[from=1-2, to=2-2]
	\arrow[from=2-1, to=2-2]
\end{tikzcd}\]
along the pullbacks $s^*,t^*: \C(G)\to \C(\mathbb{G}_1)$. The map $ \Delta_v$ defines the \textbf{vertical coproduct/cocomposition}
\begin{equation*}
    \Delta_v: \C(\mathbb{G}_1) \rightarrow \C(\mathbb{G}_1)\times^{\C(G)}\C(\mathbb{G}_1).
\end{equation*}
    \end{enumerate}
\end{definition}
\noindent Note that in this categorical language, the analogue of the coproduct density conditions for compact quantum groups \cite{Woronowicz1988} is hidden in the measureable equivalence $\mathfrak{C}(\mathbb{G}^{\times 2})\simeq \C(\mathbb{G})\times\C(\mathbb{G})$.


\begin{proposition}
    The coproducts satisfy the following \textbf{cointerchange relations}
    \begin{equation}
        (\Delta_v\times\Delta_v)\circ \Delta_h \cong (1\times \sigma\times 1)\circ (\Delta_h\times\Delta_h)\circ \Delta_v\label{cointerchange}
    \end{equation}
    where $\sigma$ is a swap of Cartesian product factors.
\end{proposition}
\begin{proof}
    This follows immediately from the interchange relation on $\mathbb{G}$.
\end{proof}

\subsection{Classical 2-graph operators and 2-gauge transformations}\label{covarrep}
By the same procedure as in \S \ref{2grpfunc}, we shall define 2-graph operators in $\C(\mathbb{G}^{\Gamma})$ as measurable fields $H^{{\Gamma^2}}$. Given a Haar measure on $\mathbb{G}$, we can then define a corresponding measure $\mu_{\Gamma}$ on $\mathbb{G}^{\Gamma}$ given by
\begin{equation*}
    d\mu(h,b)_{\Gamma^2} = \prod_{e\in\Gamma^1}d\sigma(h_e) \prod_{f\in\Gamma^2}d\nu^{h_e}(b_f),
\end{equation*}
where the product in the second factor is over faces $f\in\Gamma^2$ for which $e$ is its source edge. {As an abuse of notation, we shall use $(h,b)_{(e,f)}$ to simply denote a configuration $\{(h_e,b_f)\}_{e\in\Gamma^1,f\in\Gamma^2}$ of 2-graph decorations here.} We will assume $\Gamma$ is finite, such that $H^X$ still admits an interpretation as a measureable (quasicoherent) sheaf on $X=\mathbb{G}^{\Gamma}$. 


\medskip

Let us now define a measureable model for the covariance datum $\varphi$, as guided by \textbf{Definition \ref{2graphopsguide}}. From \S \ref{configs}, there is an induced action of the monoidal groupoid $\mathbb{G}^{\Gamma^1}$ of 2-gauge transformations on the 2-graph operators $\mathfrak{C}(\mathbb{G}^{\Gamma})$. As such, we define $\varphi=\Lambda$ to be the {\it infinite-dimensional} $\mathbb{G}^{\Gamma^1}$-module  structure  \cite{Baez:2012} of $\C(\mathbb{G}^\Gamma)$,
\begin{equation}
    \Lambda: \mathbb{G}^{\Gamma^1}\times \mathfrak{C}(\mathbb{G}^{\Gamma}) \rightarrow  \mathfrak{C}(\mathbb{G}^{\Gamma})\,,\label{2gauactino}
\end{equation}
where $X=\mathbb{G}^{\Gamma}$. To characterize $\Lambda_{(a,\gamma)}$ in terms of the horizontal conjugation action $\operatorname{hAd}^{-1}_{(a,\gamma)}$ given by a 2-gauge transformation \eqref{2gt}, we shall leverage Prop. 46 of \cite{Baez:2012}, which states the following. 



\begin{proposition}\label{pullbackmeas}
    All measureable automorphisms on a measureable category $\cH^X$ over $(X,\mu)$ are measureably naturally isomorphic to one induced by pulling back a measureable map $f:X\rightarrow X$.
\end{proposition}
\noindent We introduce the following notion. 
\begin{definition}\label{concretify}
    We say that the 2-gauge transformation action $\Lambda$ \eqref{2gauactino} is \textbf{realized concretely} on $\mathfrak{C}(\mathbb{G}^{\Gamma})$ iff
\begin{enumerate}
    \item $\Lambda_\zeta \cong (\operatorname{hAd}_\zeta)^*$ are measureable naturally isomorphic by \textbf{Proposition \ref{pullbackmeas}} for {each} decorated 1-graph $\zeta\in\mathbb{G}^{\Gamma^1}$,
    \item for each 2-graph operator $\phi=\Gamma_c(H^X)$, the pullback $(\operatorname{hAd}_\zeta)^*(\phi)=\Gamma_c(H_\zeta^X)$ induces a(n essentially) bounded linear operators $U_\zeta=U_\zeta^\phi: \Gamma_c(H^X)\rightarrow \Gamma_c(H_{\zeta}^X)$ on the measureable sections, and
    \item  the operator norm map $\zeta\mapsto |U_\zeta|$ is $\mu_{\Gamma^1}$-measureable for each $\phi\in\C(\mathbb{G}^{\Gamma})$.
\end{enumerate}
\end{definition}
\noindent In the following, we will always assume that $\mathbb{G}^{\Gamma^1}$ is realized concretely on $\mathfrak{C}(\mathbb{G}^{\Gamma})$. 


We require $H_{(1,{\bf 1}_1)}^X=H^X$ and $U_{(1,{\bf 1}_1)}=\id$ on the unit decorated 1-graph. The strict associativity of $\mathbb{G}$ means that there are canonical ($\mu_{\Gamma}$-a.e.) identifications of bundles
\begin{gather*}
    \alpha^\Lambda_{(a,\gamma),(a',\gamma')}:\Lambda_{(a,\gamma)}(\Lambda_{(a',\gamma')}H^X) \cong \Lambda_{(a,\gamma)\cdot (a',\gamma')}H^X,\\ \alpha^\Lambda_{(a,\gamma),(a',\gamma')}\circ (U_{(a,\gamma)}U_{(a',\gamma')})=U_{(a,\gamma)\cdot(a',\gamma')}\circ \alpha^\Lambda_{(a,\gamma),(a',\gamma')}
\end{gather*}
over $X$, which implements the monoidality of the $\mathbb{G}^{\Gamma^1}$-representation $\Lambda$, where $\cdot$ is the composition of 2-gauge transformations defined in \S \ref{disc2hol}. 

\begin{rmk}\label{projrep}
    The above sheaf isomorphism $\alpha^\Lambda_{(a,\gamma),(a',\gamma')}$ is known as the \textit{module associator} for $\Lambda$. They are required to satisfy the \textit{module pentagon} relations \cite{etingof2016tensor,Bartlett:2009PhD,Bartsch:2022mpm,Bartsch:2023wvv,Delcamp:2023kew} against the tensor product associator on $\mathfrak{C}(\mathbb{G}^{\Gamma})=\C(\mathbb{G}^{\Gamma})$. In the context of lattice 2-gauge theory, it is directly induced by the modifications $m:\operatorname{hAd}^{-1}\circ\operatorname{hAd}^{-1}\Rrightarrow \operatorname{hAd}^{-1}$ described in \S \ref{weak2gauthy}, and hence can in general have non-identity components in weak 2-Chern-Simons theory. In the strict theory, they can be normalized to a $U(1)$-phase.
\end{rmk}

There is an analogous construction for the {\it right} 2-gauge transformation $\mathrm{P}: \mathbb{G}^{\Gamma^1}\times  \mathfrak{C}(\mathbb{G}^{\Gamma}) \rightarrow\mathfrak{C}(\mathbb{G}^{\Gamma})$, which is nothing but the left action $\Lambda$ of the opposite $(\mathbb{G}^{\Gamma^1})^\text{op}$ with the sources and targets swapped. We represent this action by the {\bf contragredient 2-representation} $\mathrm{P}$ of $\Lambda$. All constructions in the following have an adjoint counterpart, hence we shall mainly focus on $\Lambda$.

We shall return to the 2-gauge transformations in \S \ref{quantum2gautransfo}, where we will make $\mathbb{G}^{\Gamma^1}$ into an additive (Hopf) measureable category, and endow $\Lambda$ with the structure of a monoidal action functor.


\section{Categorical quantum 2-graph operators}\label{quantumcategory}
Let us now consider what sort of structure the 2-graph operators comes equipped with. We start by studying the geometry of planar graphs on $\Sigma$ with no self-intersections, where we keep track of faces. The goal of this section is to extend the \textbf{Definition \ref{coproducts}} of the "local" coproducts to the geometric setting on the 2-graph $\Gamma$.

\subsection{Geometry of 2-graphs}\label{2graphcoprod}
In analogy with the 3d case, the (partial) coproduct(s) on 2-graph operators should be defined from  pulling back the attaching operations on the 2-graph $\Gamma$. Recall that these attaching operations are modelled by a 2-groupoid structure $\cup_v,\cup_h$ in \S \ref{2graphdef}.

To describe them, we first describe the product structures on the 2-holonomies $\mathbb{G}^{\Gamma}$ (see \cite{Martins:2006hx,Bullivant:2019tbp,Bochniak:2020vil}),
\begin{equation*}
    \cdot_h:\mathbb{G}^{\Gamma}\times \mathbb{G}^{\Gamma} \rightarrow \mathbb{G}^{\Gamma},\qquad \cdot_v:\mathbb{G}_1^{\Gamma^2} \times_{G^{\Gamma^1}} \mathbb{G}_1^{\Gamma^2} \rightarrow \mathbb{G}_1^{\Gamma^2}\,,
\end{equation*}
where $\mathbb{G}_1=(\mathsf{H}\rtimes G)$ and $\mathbb{G}_1^{\Gamma^2} \times_{G^{\Gamma^1}} \mathbb{G}_1^{\Gamma^2}$ denotes the pullback along the source/target maps $s,t: (\mathsf{H}\rtimes G)^{\Gamma^2}\to G^{\Gamma^1}$ on the decorated 2-graphs. 

More precisely, if the face $(e,f)\in\Gamma^2$ is given by the {\it horizontal} gluing $\cup_h$ of the half-faces $(e_1,f_1),(e_2,f_2)$, then we put 
$$ ((h,b)\cdot_h(h',b'))_{(e,f)} =(h_{e_1},b_{f_1})\cdot_h(h_{e_2}',b_{f_2}') = (h_{e_1}h'_{e_2},b_{f_1}(h'_{e_1}\rhd b'_{f_2})).$$ Similarly, if $(e,f)$ is given by the \textit{vertical} gluing $\cup_v$ of the half-faces such that the target edge $e'_1$ of $f_1$ coincides with the source edge $e_1$ of $f_1$, then we have $$((h,b)\cdot_v(h',b'))_{(e,f)} = (h_{e_1},b_{f_1})\cdot_v(h'_{e_2},b'_{f_2}) = (h_{e_1},b_{f_1}b'_{f_2}).$$ The pullback $\mathbb{G}_1^{\Gamma^2} \times_{G^{\Gamma^1}} \mathbb{G}_1^{\Gamma^2}$ means that $\cdot_v$ is only well-defined if $h_{e_1'}=h_{e_2}$.  Now if $(e_i,f_i)$ for $i=1,\dots,4$ are 2-graphs which are mutually horizontally/vertically composable, then \eqref{graphinterchange} in conjunction with the strict interchange law in $\textbf{B}\mathbb{G}$, we have the interchange relation
\begin{equation*}
    ((h_{e_1},b_{f_1})\cdot_h(h_{e_2},b_{f_2}))\cdot_v((h_{e_3},b_{f_3})\cdot_h(h_{e_4},b_{f_4})) = ((h_{e_1},b_{f_1})\cdot_v(h_{e_3},b_{f_3}))\cdot_h((h_{e_2},b_{f_2})\cdot_v(h_{e_4},b_{f_4}))
\end{equation*}
for the decorated 2-graphs. Equality is achieved on-shell of the 2-flatness condition.  

\medskip

Pulling back the products $\cdot_h,\cdot_v$ as in \textbf{Definition \ref{coproducts}}, we can then define the desired geometric,  coproducts:
\begin{equation*}
    \Delta_h: {\mathfrak{C}}(\mathbb{G}^{\Gamma})\rightarrow\mathfrak{C}(\mathbb{G}^{\Gamma})\times\mathfrak{C}(\mathbb{G}^{\Gamma}),\qquad \Delta_v: {\mathfrak{C}}(\mathbb{G}^{\Gamma})\rightarrow\mathfrak{C}(\mathbb{G}^{\Gamma})\times^{\C(G^{\Gamma^1})}\mathfrak{C}(\mathbb{G}^{\Gamma})\,,
\end{equation*}
where $\C(\mathbb{G}^{\Gamma})\times^{\C(G^{\Gamma^1})}\mathfrak{C}(\mathbb{G}^{\Gamma})$ denotes the pushout. In analogy with \eqref{partialcoprod}, the 2-groupoid identity $\id_e$, which can be viewed as a constant PL homotopy of the edge $e$, allows us to view all faces $f\in\Gamma^2$ as "composite" and make use of the formulas given in \textbf{Definition \ref{2graphfusion}} below.

\medskip
 
To describe them more explicitly, we make use of the \textbf{localized 2-graph operators} $\phi_{(e,f)}$ of a given measureable sheaf $\phi$, whose global sections are multiplied by the characteristic function (ie. a Kronecker delta) at the single 2-graph face $(e,f)\in\Gamma^2$.

\begin{definition}\label{2graphfusion}
    Let $\phi_{(e,f)}$ denote the 2-graph operator localized on $(e,f)\in\Gamma^2$. The \textbf{2-graph coproducts} $\Delta_h,\Delta_v$ on $\mathfrak{C}(\mathbb{G}^{\Gamma})$ read, in Sweedler notation,
    \begin{align*}
    &\Delta_h(\phi_{(e,f)}) = \bigoplus_h \delta_{t(e_1),s(e_2)}(\phi_{(1)})_{(e_1,f_1)}\times(\phi_{(2)})_{(e_2,f_2)},\qquad (e,f) = (e_1,f_1)\cup_h(e_2,f_2),\\ &\Delta_v(\phi_{(e,f)}) = \bigoplus_v\delta_{e_2,e_1\ast \partial f_1} (\phi_{(1)})_{(e_1,f_1)} \times(\phi_{(2)})_{(e_2,f_2)},\qquad (e,f)=(e_1,f_1)\cup_v(e_2,f_2),
\end{align*}
where "$\bigoplus_h$" denotes a direct sum over all 2-graph operators such that we have a sheaf isomorphism
\begin{equation*}
    \bigoplus_h \phi_{(1)}(\{(h_{e_1},b_{f_1})\}_{(e_1,f_1)})\otimes \phi_{(2)}(\{(h_{e_2},b_{f_2})\}_{(e_2,f_2)}) \cong \phi(\{(h_{e_1},b_{f_1})\cdot_h(h_{e_2},b_{f_2})\}_{(e,f)})
\end{equation*}
over $X$. Similarly for "$\bigoplus_v$". 
\end{definition}
\noindent In the following, it would be convenient to use the shorthand
\begin{equation}
    \bigoplus_h\phi_{(1)}\times\phi_{(2)} = \phi_{(1)}\times_h\phi_{(2)},\qquad \bigoplus_v\phi_{(1)}\times\phi_{(2)} = \phi_{(1)}\times_v\phi_{(2)}.\label{shorthand}
\end{equation}
Clearly, we recover \textbf{Definition \ref{coproducts}} when $\Gamma=D^2$ is a PL 2-disc --- ie. it consists of a single face $(e,f)$ and a single vertex $v$. 

\begin{proposition}\label{cointerchangeprop}
    These coproducts are strict coassociative, and satisfy the cointerchange \eqref{cointerchange}.
\end{proposition}
\begin{proof}
    These follow directly from properties of the decorated 2-graphs.
\end{proof}


\begin{rmk}
    The geometric interpretation of the cointerchange relation \eqref{cointerchange} is the consistency of decomposing a 2-graph operator in two different ways, when it is localized on a 2-graph obtained by gluing four composable faces. The vertex of this gluing is precisely a \textit{triple point} in the handlebody decomposition of 3-manifolds \cite{Sakata2022-il}. By the 2-flatness condition, an invertible homotopy $\beta$ will be assigned to witness \eqref{cointerchange}, which we have chosen to be the identity (this is doable in the strict case; for weak 2-Chern-Simons theory, this homotopy $\beta$ will have non-identity components). This is how 2-group gauge theory can detect the triple linking of surfaces; see also \S 7 of \cite{PUTROV2017254}.
\end{rmk}

Following {\bf Definition \ref{2graphfusion}}, we shall see in the following how we can introduce a 2-group version of the Fock-Rosly Poisson bracket, which will lead to a quantum deformation of these coproducts.

\subsubsection{Combinatorial 2-group Fock-Rosly Poisson bracket}
We now study a higher but discrete version of the Fock-Rosly Poisson structure, and see how it helps in defining the quantum deformation of the measureable sections "$\C_q(\mathbb{G}^{\Gamma})$" modelling the 2-graph operators. We shall see how the data of the 2-Chern-Simons action deforms these coproduct structures, and in particular equips $\C_q(\mathbb{G}^{\Gamma})$ with the structures of a \textbf{Hopf measureable cocategory}. The use of Hopf (co)categories to construct 4d TQFTs is not a new concept \cite{Crane:1994ty,Pfeiffer2007,Baez:1995xq}, but we provide here an explicit construction from the underlying 4d topological 2-Chern-Simons action.



As in the usual 3d Chern-Simons case, the 2-Chern-Simons action determines a {\it classical 2-$r$-matrix} \cite{Bai_2013} of odd degree 1,
\begin{equation*}
    r=r' - D_tr_0 \in (\G\otimes \G)_1,\qquad D_{t}r=(t\otimes 1 - 1\otimes t)r=0,
\end{equation*}
where $r_1\in \h\otimes\h$ with
\begin{equation*}
    D_tr_1 = (t\otimes 1 + 1\otimes t)r_1.
\end{equation*}
The symmetric part comes from the symplectic form
\begin{equation*}
    \omega(A,B) =\frac{k}{2\pi} \int_\Sigma \langle \delta B,\delta A\rangle
\end{equation*}
of 2-Chern-Simons theory, while the skew-symmetric part can be read off from the interaction terms $\langle B,[A,A]-tB\rangle$ \cite{chen:2022}. 

This identifies a Lie 2-algebra cobracket $\delta = [-,r]$ \cite{Bai_2013}, from which one can construct a bivector field $\Pi\in \frak{X}^2$ on $\mathbb{G}$. This bivector field $\Pi$ can be shown to be multiplicative, with respect to both the group and groupoid structures, precisely when the cocycle condition for $\delta$ is satisfied \cite{Chen:2023integrable}. This makes $(\mathbb{G},\Pi)$ into a \textit{Poisson-Lie 2-group} \cite{Chen:2012gz}.

\begin{rmk}\label{deg0cybe}
    The solutions $r\in\G^{\otimes 2}_1$ to the 2-graded classical Yang-Baxter equations $\llbracket r,r\rrbracket=0$ on a strict Lie 2-algebra $\G = \h\xrightarrow{t}\g$, where $\llbracket-,-\rrbracket$ is the \textit{graded} Schouten bracket, was analyzed in \cite{Bai_2013}. It was found that the image $(1\otimes t)r = (t\otimes 1)r$ of $r$ under the Lie 2-algebra structure map $t:\h\rightarrow\g$ is a solution of the \textit{ordinary} classical Yang-Baxter equations on $\g$. Based on this observation, it was shown in \cite{Bai_2013,chen:2022} that: (i) there is a one-to-one correspondence between ordinary classical $r$-matrices for a simple Lie algebra $\g$ and classical 2-graded $r$-matrices for its inner automorphism Lie 2-algebra $\G = \operatorname{inn}\g = \g\xrightarrow{\id}\g$; (ii) there is a one-to-one correspondence between ordinary classical $r$-matrices for the \textit{semidirect product} $V\rtimes \g$, where $V$ is an Abelian $\g$-module, and classical 2-graded $r$-matrices for $\G = V\xrightarrow{0}\g$.
\end{rmk}

It is also worth mentioning that the quadratic 2-Casimirs were studied in \cite{chen:2022}, while \cite{Cirio:2012be} examined solutions of the classical 2-Yang-Baxter equations in the context of \textit{weak} Lie 2-algebras.

\medskip

In analogy with the usual Drinfel'd-Jimbo deformation quantization, the data of the classical 2-$r$-matrix is expected to give rise to quantum deformations of \textit{both} the product and the coproduct on $\mathfrak{C}(\mathbb{G})$. This deformation is controlled to first order by the classical 2-$r$-matrix  $r\in\mathfrak{G}_1^{\otimes 2}$ and its graded transpose $r^T$,
\begin{equation*}
    r = \sum r_1\otimes r_2,\qquad r^T_{12} = \sum r_2\otimes r_1.
\end{equation*}
By interpreting elements of the Lie 2-algebra $\frak{G}$ as derivations on functions $C(\mathbb{G})$ of $\mathbb{G}$ (see \textbf{Proposition \ref{derivationaction}} and \cite{Chen:2012gz} for more details), we can extend it to act on measureable global sections. 

This allows us to introduce the following combinatorial Poisson brackets of Fock-Rosly type. 
\begin{definition}
    The \textbf{combinatorial 2-group Fock-Rosly Poisson bracket} on a measureable sheaf $\Gamma_c(H^X)$ over $X=(\mathbb{G}^{\Gamma},\mu_{\Gamma})$ is
    \begin{align}
    & \{\xi_{(e,f)},\xi_{(e',f')}\} = \frac{2\pi}{k}(-\cdot-)\left(\delta_{t(e),s(e')} r (\xi_{(e,f)} \xi_{(e',f')}) - \delta_{s(e),t(e')}(\xi_{(e,f)} \xi_{(e',f')}) r^T\right),\label{2fockrosly}
\end{align}
where $\xi_{(e,f)}$ denotes a global section of a localized 2-graph operator $\phi_{(e,f)}$; equivalently, it is the localization of a global section $\xi\in \Gamma_c(H^X)=\phi$ to the face $(e,f)\in\Gamma^2$. More precisely, this is $\chi_{(e,f)}^{[2]}\xi$ where $\chi^{[2]}_{(e,f)}$ is the characteristic function on $(e,f)$.
\end{definition}
\noindent The (strict) Jacobi identity of $\{-,-\}$ follows from the 2-graded classical Yang-Baxter equations satisfied by $r$ \cite{Bai_2013}.

Suppose $\Gamma^2$ is the 2-graph complex underlying a 2-disc, such that it consist of a single face and a single edge starting and ending on a single vertex. Then the brackets \eqref{2fockrosly} recover the corresponding "fundamental 2-graded Poisson structure" $\{-,-\}$ on the function algebra of $\mathbb{G}$ \cite{Chen:2012gz,Chen:2023integrable}.


Let $(e,f)\cup_{h,v}(e',f')\in \Gamma$ denote horizontal//vertical composite polygonal faces from $(e,f),(e',f')$. Taking inspiration from \eqref{heis-bracket}, it will be useful to rewrite the bracket in the following way,
\begin{equation}
   \{\xi_{(e,f)},\xi_{(e',f')}\} \equiv \frac{2\pi}{k}(-\cdot -)[r,\Delta_h(\xi_{(e,f)\cup_{h}(e',f')})]_c,\label{2fockrosly-heis}
\end{equation}
where $\Delta_h$ is the horizontal 2-graph coproduct introduced in \textbf{Definition \ref{2graphcoprod}} and $[-,-]_c$ is the commutator with respect to the (commutative) product on the space of sheaves on $\Gamma_c(H^X)$. 

Of course, if the two faces $(e,f),(e',f')$ are too "far apart"/delocalized --- then we interpret $ \{\xi_{(e,f)},\xi_{(e',f')}\} =0$ as the zero section over $X=\mathbb{G}^{\Gamma}$. We shall use this Poisson bracket to introduce a deformation quantization of $\mathfrak{C}(\mathbb{G}^{\Gamma})$ in \S \ref{hopfopalg}.

\subsubsection{Antipodes and 2-$\dagger$ unitarity}\label{2dagger}
With the introduction of the coproducts above, we now leverage the geometry of 2-graphs once more to define the \textit{antipode} functor $S_\mathrm{v,h}$ on $\mathfrak{C}(\mathbb{G}^{\Gamma})$. Specifically, $S$ is induced from \textit{orientation reversal}, as inspired from \cite{Alekseev:1994pa,Delcamp:2016yix}.

Following {\bf Example 5.5} of \cite{ferrer2024daggerncategories}, we take the 2-graph $\Gamma^2$ as a framed piecewise-linear (PL) 2-manifold. The PL-group $\operatorname{PL}(2) = O(2) = SO(2)\rtimes \bbZ_2$ tells us directly what the 2-dagger structure on $\Gamma$ is --- $\dagger_2$ is given by the orientation reversal $\bbZ_2$ subgroup and $\dagger_1$ is a $2\pi$-rotation in the framing $SO(2)$-factor. 

Crucially, these daggers are involutive $\dagger_2^2 = \id,~\dagger^2_1 \cong \id$ and they {\it strongly commute} 
\begin{equation}
    \dagger_2\circ \dagger_1 = \dagger_1^\text{op}\circ\dagger_2\label{commutedagger}.
\end{equation}
For edges in $\Gamma^1$, on the other hand, $\dagger_2$ implements an orientation reversal $e^{\dagger_2} = \bar e$ while $\dagger_1$ rotates its framing: if $\nu$ is a trivialization of the normal bundle along the embedding $e\hookrightarrow \Sigma$, then $(e,\nu)^{\dagger_1} =(e,-\nu)$. Let us denote this frame rotation by the shorthand $e^T = (e,-\nu)$.

We denote the induced maps on the measureable Lie 2-groups by
$X=\mathbb{G}^{\Gamma}\xrightarrow{\sim} \overline{X}^{\mathrm{h,v}} = \mathbb{G}^{(\Gamma^2)^{\dagger_2,\dagger_1}}$. Recall the action of the 2-gauge transformations $\Lambda$ given by bounded linear operators $U$ form \S \ref{covarrep}.
    \begin{definition}\label{unitary2hol}
Define the \textbf{antipode functors}
\begin{equation}
  S_v:\C(\mathbb{G}^{\Gamma})\rightarrow  \C(\mathbb{G}^{\Gamma})^{{\text{op}}},\qquad S_h: \C(\mathbb{G}^{\Gamma})\rightarrow \C(\mathbb{G}^{\Gamma})^{\text{m-op,c-op}},\label{antilinear}
\end{equation} 
      where "$-^{{\text{op}}}$" denotes taking the opposite cocategory, and "$-^{\text{m-op,c-op}}$" denotes taking the reverse monodal/comonoidal structure. The \textbf{2-$\dagger$ unitarity of the 2-holonomies} is the property that:
        \begin{itemize}
            \item For each 2-graph operator in $\mathfrak{C}(\mathbb{G}^{\Gamma})$, we have stalk-wise for each $\mathrm{z}=\{(h_e,b_f)\}_{(e,f)}\in\mathbb{G}^{\Gamma}$,
            \begin{align*}
        (S_h\phi)_{\mathrm{z}} &= \bar\phi_{\mathrm{z}^{\dagger_1}},\qquad \mathrm{z}^{\dagger_1}=\{(h_{e^{\dagger_1}},b_{f^{\dagger_1}})\}_{(e,f)}\\ (S_v\phi)_{\mathrm{z}} &= \phi^T_{\mathrm{z}^{\dagger_2}},\qquad \mathrm{z}^{\dagger_2} = \{(h_{e^{\dagger_2}},b_{f^{\dagger_2}})\}_{(e,f)}
    \end{align*}
        where $\bar\phi$ is the measureable field $(H^*)^X$ complex linear dual to $\phi$, and $\phi^T$ is the same sheaf underlying $\phi\in\C(\mathbb{G}^{\Gamma})$ but equipped with the adjoint sheaf morphisms. 
    \item For the 2-gauge transformation operators $U_\zeta$ introduced in \S \ref{covarrep}, we have pointwise for each $\zeta=\{(a_v,\gamma_{e}\}_{(e,v)}\in\mathbb{G}^{\Gamma^1}$ (recall $e^T = (e,-\nu)$ denotes a frame rotation of an edge),
    \begin{align*}
            U_{\tilde S_h\zeta} &= \bar U_{\zeta^{\dagger_1}},\qquad \zeta^{\dagger_1}=\{ (a_{v'}\xrightarrow{\gamma_{\bar e}}a_v)\}_{(a,v)},\\ U_{\tilde S_v\zeta}&= U^\dagger_{\zeta^{\dagger_2}},\qquad \zeta^{\dagger_1}=\{(a_{v}\xrightarrow{\gamma_{e^T}}a_{v'})\}_{(a,v)}
    \end{align*}
    where $\bar U_\zeta$ is the complex conjugate operator and $U_\zeta^\dagger$ is the (Hermitian) adjoint.
         \end{itemize}
    \end{definition}
\noindent These 2-$\dagger$-unitarity properties will come into play once again in \S \ref{*op}.

In the following section, we shall introduce a quantum deformation of the 2-graph operators from $\C(\mathbb{G}^{\Gamma})$ to $\C_q(\mathbb{G}^{\Gamma})$, which promotes the antipode $S_h$ in \eqref{antilinear} to be generally non-involutive.

\subsection{Quantum deformation on the lattice}\label{categoricalquantumdeformations}
Recall in the case of the coordinate ring $C(G)$ for an ordinary Lie group $G$, a quantum deformation $\star$ of its commutative product can be introduced from the data of a classical $r$-matrix on $\g=\operatorname{Lie}G$, such that the $\star$-commutator $ [-,-]_\star$ is controlled to first order in $\hbar$ by \eqref{2fockrosly-heis} \cite{Woronowicz1988,Majid:1996kd,Semenov1992}. 

We are now tasked with two goals:
\begin{enumerate}
    \item  \textit{quantize} the classical 2-$r$-matrix $r$ on $\G$ to a quantum 2-$R$-matrix $R$, which act as quantum deformations of the structure \eqref{2fockrosly},
\begin{equation*}
    R\sim 1+ i\hbar r +o((i\hbar)^2)
\end{equation*}
(see also \textit{Remark \ref{2rmatrixboostrap}} later).
    \item \textit{categorify} the deformed $\star$-product to the entire measureable category $\mathfrak{C}(\mathbb{G}^{\Gamma})$.
\end{enumerate}
The first point, for $C(\mathbb{G})$, was studied in \cite{Chen:2023tjf}, which led to the motivating example for the notion of a 2-$R$-matrix there. Let us briefly recall this result.

\subsubsection{Hopf 2-algebras and Baez-Crans 2-vector spaces}\label{coprod}
A \textbf{Baez-Crans 2-vector space} $V\in \mathsf{2Vect}^{BC}=\operatorname{Cat}_\mathsf{Vect}$ is a category internal to $\mathsf{Vect}$ \cite{Baez:2003fs,Kemp_2025}. The central motivation for considering this setting is that Baez-Crans 2-vector spaces are the backdrop for the sort of $L_\infty$-(bi)algebras \cite{Chen2z:2023,Bai_2013,Chen:2013} and Hopf $A_\infty$-algebras \cite{Wagemann+2021,Chen:2023integrable} that arise in principal higher-bundles \cite{Baez:2004in,schreiber2013connectionsnonabeliangerbesholonomy,Nikolaus2011FOUREV,Wockel2008Principal2A,Kim:2019owc} and higher-gauge theory \cite{Baez:2002jn,Soncini:2014,Song:2021,chen:2022}. 

Indeed, a Lie 2-algebra is nothing but a Lie algebra object in $\mathsf{2Vect}^{BC}$ \cite{Baez:2003fs}, and this inspires the following notion introduced in \cite{Chen:2023tjf}.
\begin{definition}
    A \textbf{Hopf 2-algebra} $A$ is a Hopf algebra object in $\mathsf{2Vect}^{BC}$.
\end{definition}
\noindent A crucial characterization of Baez-Crans 2-vector spaces was given in \cite{Baez:2003fs}.
\begin{proposition}\label{2chain}
    There is an equivalence $\mathsf{2Vect}^{BC}\simeq \mathsf{2Ch}(\mathsf{Vect})$ with the 2-truncated 2-category of 2-term chain complexes.
\end{proposition}
\noindent One side of the equivalence, which associates a Baez-Crans 2-vector space $V=V_1\overset{s}{\underset{t}{\rightrightarrows}} V_0$ with a 2-term chain complex $V_{-1}\xrightarrow{\mu_1}V_0$, is given by
$$V_{-1} = \operatorname{ker}s\subset V_1,\qquad \mu_1(y) = x'-x,$$
where $y:x\rightarrow x'\in V_1$ is a morphism. 

Due to \textbf{Proposition \ref{2chain}}, the coproduct $\Delta= (\Delta_0,\Delta_1): A[1]\rightarrow A^{\otimes 2}$ on a Hopf 2-algebra $A=A_{-1}\xrightarrow{\mathsf{t}}A_0$ is a \textit{differential graded} map, satisfying
    \begin{equation}
    (\mathsf{t}\otimes 1 + 1\otimes \mathsf{t})\circ \Delta_{-1} = \Delta_0\circ \mathsf{t},\qquad (\mathsf{t}\otimes 1-1\otimes  \mathsf{t})\circ\Delta_0 =0.\label{coprodequiv}
\end{equation}
The structure of the classical 2-$r$-matrix on $\G$ \cite{Bai_2013,Chen:2023integrable} then suggests the following quasitriangularity structure.\footnote{In the ordinary quantum groups case, there is strictly speaking an obstruction to deformation quantization \cite{Drinfeld:1986in,Semenov1992}. These are known to live as a certain degree-3 cohomology class of the underlying *-algebra, but such obstructions vanish for solution of the classical Yang-Baxter equation \cite{LAZAREV1996141}. Similarly here: since the Lie 2-algebra cocycle $\delta$ under consideration is determined by a solution to the 2-graded classical Yang-Baxter equations  \cite{Bai_2013}, we expect such obstructions to vanish.}
\begin{definition}\label{2Rmat}
    Let  $A=A_{-1}\xrightarrow{\mathsf{t}}A_0$ denote a Hopf 2-algebra. A \textbf{2-$R$-matrix} $R$ for $A$ is an element $R\in A_0\otimes A_{-1}\oplus A_{-1}\otimes A_0$ which satisfies the following graded conditions:
    \begin{itemize}
        \item the \textbf{intertwining relation}
        \begin{equation*}
            \Delta^\text{op}(\mathrm{z})R = R\Delta(\mathrm{z}),\qquad\forall~ \mathrm{z}\in A,
        \end{equation*}
        \item the \textbf{equivariance condition}
        \begin{equation}
            (\mathsf{t}\otimes 1-1\otimes \mathsf{t})R =0,\label{Rmatequiv}
        \end{equation}
        and
        \item \textbf{quasitriangularity conditions}
            \begin{equation}
        (\Delta\otimes 1)R = R_{13} R_{12},\qquad  (1\otimes\Delta)R = R_{13}R_{23},\label{quasit}
    \end{equation}
    \item the \textbf{antipode conditions}
    \begin{equation*}
        ((1\otimes S)R)_{12}\cdot R_{23} = R_{12}\cdot ((S\otimes 1)R)_{23} = \eta_1\eta_3,
    \end{equation*}
    where the legs labelled by "2" are contracted, and
    \item the obvious counit conditions $(1\otimes\epsilon)R=(\epsilon\otimes 1)R=\eta$.
    \end{itemize}
\end{definition}
\noindent It is not hard to show that $\bar R=(\mathsf{t}\otimes1)R=(1\otimes\mathsf{t})R$ satisfies
\begin{equation*}
    \bar\Delta_0^\text{op}(x)\bar R = \bar R\bar\Delta_0(x),\qquad \forall x\in A_0,
\end{equation*}
where $\bar \Delta_0 = (\mathsf{t}\otimes 1)\circ\Delta_0=(1\otimes \mathsf{t})\circ\Delta_0$. In fact, at degree-0, $(A_0,S_0,\bar\Delta_0,\bar R)$ is an ordinary Hopf algebra equipped with a $R$-matrix. 

Clearly, by thinking of $A$ as a category internal to $\mathsf{Vect}$, then a(n invertible) 2-$R$-matrix induces, through the conjugation action, a bimonoidal automorphism
\begin{equation*}
    \mathcal{R} = \operatorname{ad}_R: A\otimes A\rightarrow A\otimes A,\qquad \mathcal{R}\circ\Delta = \Delta^\text{op}.
\end{equation*}
This is the crucial insight that we shall leverage below.

\begin{rmk}\label{2rmatrixboostrap}
 The condition \eqref{Rmatequiv} above implies that a 2-graded $R$-matrix for $\mathfrak{C}(\mathbb{G})$ can be "boostrapped" from a $R$-matrix for $C(G)$ at degree-0. Further, if $\mathbb{G}=\operatorname{Inn}G$ were the inner automorphism 2-group of a compact simple Lie group $G$ with $t=\id$, then 2-$R$-matrices for $\operatorname{Inn}G$ has a direct bijective correspondence with quantum $R$-matrices for $C(G)$ through \eqref{Rmatequiv}. As is well-known, solutions of the quantum Yang-Baxter equations have been extensively studied since the 80's \cite{Jimbo:1985zk,Drinfeld:1986in,,Majid:1996kd,Woronowicz1988,reshetikhin2010lecturesintegrability6vertexmodel,Zhang:1991,Larsson:2002pk,sym15091623}. This observation is useful for constructing explicit examples of 2-Chern-Simons TQFTs.
\end{rmk}

The main result in \cite{Chen:2023tjf} verifies that a quasitriangular 2-$R$-matrix on $A$ endows the 2-category of 2-representations of $A$ a braided monoidal structure. {\it Example 2.12.1} there states that the $\bbC$-valued function algebra $A=C(\mathbb{G}) = C(G)\xrightarrow{t^*}C(\mathsf{H})$ is a commutative {Hopf 2-algebra}. 

Moreover, Appendix B of \cite{Chen:2023tjf} proves the following.
\begin{proposition}\label{hopf2-algdeformation}
    There exists a non-commutative quasitriangular Hopf 2-algebra equipped with a 2-$R$-matrix $R$, denoted $(C_q(\mathbb{G});R)$, which at first order in $\hbar$ reduces to the Poisson-Lie 2-group $(C(\mathbb{G}),\{-,-\})$ \cite{Chen:2012gz}. Further,  the Poisson bracket $\{-,-\}$ is canonically induced by a classical 2-$r$-matrix $r$ on the associated Lie 2-algebra $\G$ \cite{Bai_2013}.
\end{proposition}
\noindent The associativity of $\star$ thus follows from the strict Jacobi identity satisfied by the combinatorial 2-Fock-Rosly brackets \eqref{2fockrosly}.

\begin{rmk}\label{quasihopf}
        In the weakly-associative case, we must deal semiclassically with (at least) a quasi-Lie 2-bialgebra, namely a Lie 2-bialgebra with non-trivial cohomotopy map \cite{Chen:2013}. It is known \cite{Chen:2012gz} that such a structure integrates to a \textit{quasi-Poisson-Lie 2-group}, which has equipped a multiplicative trivector field $\eta$ witnessing the Jacobi identity for the graded Poisson brackets. From the above construction, it then stands to reason that $\eta$ gives rise to an associator for the quantum deformed monoidal structure on $\mathfrak{C}(\mathbb{G}^{\Gamma})$ (see \S \ref{quantummonoidal} later). 
\end{rmk}


\begin{tcolorbox}[breakable]
    \subsubsection*{Motivation for categorification.} The reader may wonder why we have chosen to work in the higher-categorical context $\mathfrak{C}(\mathbb{G})$, instead of just working with the simpler function Hopf 2-algebra $C(\mathbb{G})$. Aside from the categorical ladder philosophy  \cite{Baez:1995xq} (fig. \ref{fig:1}), there are multi-fold reasons:
    \begin{enumerate}
        \item Mathematically, categorification bypasses many issues suffered by the Baez-Crans 2-vector spaces $\mathsf{2Vect}^{BC}$ and the Morita context it defines (see \S \ref{hopf2algproblems}),
        \item Physically, Hopf 2-algebras based on Baez-Crans 2-vector spaces is \textit{not} sufficient to describe correlation functions of Wilson surfaces and 2-holonomies, even in the strict case (see \S \ref{2holstates}),
        \item Practically, categorification is necessary in order to detect $k$-invariants and higher-codimensional defects arising from higher-gauge symmetry; for instance, the module associators $\alpha^\Lambda$ (see \textit{Remark \ref{projrep}}) would not be present otherwise.
    \end{enumerate}
    
\end{tcolorbox}



\subsubsection{Lifting the quantum product to measureable sheaves}\label{quantummonoidal}
To categorify the above structures, we need to invoke the main result of \cite{Bursztyn2000DeformationQO}:
\begin{theorem}\label{deformedsection}
    Let $X$ denote a Riemannian manifold. A fixed $\star$-product on $C(X)$ determines uniquely (up to isometry on $X$) a $\star$-product on the smooth sections $\Gamma(X,E)$ of a vector bundle $E\rightarrow X$. The resulting sheaf of $\star$-deformed global sections, denoted $\Gamma(X,E)[[\hbar]]$, is a $C(X)\otimes_\bbC \bbC[[\hbar]]$-module $C^*$-algebras.
\end{theorem}
\noindent Since the \textit{classical} 2-graph operators $\phi\in\mathfrak{C}(\mathbb{G}^{\Gamma})$ are modelled as modules over the structure sheaf $\mathcal{O}_X= C(\mathbb{G}^{\Gamma})$ on $X=(\mathbb{G}^{\Gamma},\mu_{\Gamma})$, then once we know the $\star$-deformation $C(\mathbb{G}^{\Gamma})\rightsquigarrow C_q(\mathbb{G}^{\Gamma})$ through \eqref{2fockrosly} and \textbf{Proposition \ref{hopf2-algdeformation}}, we can leverage \textbf{Theorem \ref{deformedsection}} to categorify the $\star$-product and the 2-$R$-matrix.\footnote{Technically, we would like to use a generalization of \textbf{Theorem \ref{deformedsection}} to quasicoherent sheaves. It was noted in \cite{Bursztyn2000DeformationQO} that their results hold for sheaves given by generic projective modules.}

\medskip

Denote by $\Gamma_c(H^X)[[\hbar]]$ the sheaf of \textit{formal power series} in the global sections over $X=(\mathbb{G}^{\Gamma},\mu_{\Gamma})$. We now construct a categorical \textit{deformed tensor product} from the underlying $\star$-product. For this, it would be useful to recall the following general fact about sheaves \cite{Forster1967}.
\begin{theorem}
    There is a canonical isomorphism $\Gamma(X,\cF\otimes_{\cO_X}\cF') \cong \Gamma(X,\cF)\otimes_{\cO_X}\Gamma(X,\cF')$ for any sheaves $\cF,\cF'$ over $X$.
\end{theorem}
\noindent When applied to sheaves of global sections in $\mathfrak{C}(\mathbb{G}^{\Gamma})$, this means that there are canonical isomorphisms
\begin{equation}
    \phi\otimes\phi'=\Gamma_c(H^X)\otimes\Gamma_c(H'^X)\cong \Gamma_c((H\otimes H')^X)\label{canoniso}
\end{equation}
of $C(X)$-modules for each $\phi,\phi'\in\mathfrak{C}(\mathbb{G}^{\Gamma})$, where $X=(\mathbb{G}^{\Gamma},\mu_{\Gamma})$.


\medskip

Our goal is to promote the canonical isomorphism \eqref{canoniso} to the quantum deformed case, with $q=e^{i\hbar}$, by using \textbf{Theorem \ref{deformedsection}}. This will require the assumption that there exists a "decategorification" $\lambda: \mathfrak{C}(\mathbb{G}^{\Gamma})\mapsto C_q(X)= C(X)\otimes\bbC[[\hbar]]$ sending measureable $C_q(X)$-module $\star$-algebras to $C_q(X)$. See \textbf{Definition \ref{hypH}} later for more details on this assumption.
\begin{definition}\label{ostarprod}
    The \textbf{deformed tensor product} is a monoidal structure $\ostar: \mathfrak{C}(\mathbb{G}^{\Gamma})\times\mathfrak{C}(\mathbb{G}^{\Gamma})\rightarrow \mathfrak{C}(\mathbb{G}^{\Gamma})$ on the 2-graph operators $\mathfrak{C}(\mathbb{G}^{\Gamma})$ such that 
    \begin{enumerate}
        \item we have natural sheaf isomorphisms
        \begin{equation}
            \Gamma_c(H^X)[[\hbar]]\ostar \Gamma_c(H'^X)[[\hbar]] \cong \Gamma_c((H\otimes H')^X)[[\hbar]],\label{deformedtensor}
        \end{equation}
        linear over $\bbC[[\hbar]]$, for all $\Gamma_c(H^X),\Gamma_c(H'^X)\in\mathfrak{C}(\mathbb{G}^{\Gamma})$ and
        \item $\lambda$ is "(strictly) monoidal", in the sense that on the essential image of $\ostar$, we have an isomrphism \begin{equation}
            \lambda(-\ostar-)\cong (\lambda -\star\lambda-)\label{ostardiagram}
        \end{equation}
        of algebra in $\mathsf{2Vect}^{BC}=\operatorname{Cat}_\mathsf{Vect}$.
    \end{enumerate}
    At the same time, we will assume that $\lambda$ fits in a commutative diagram for the \textit{coproducts}, analogous to the above. We call $\ostar$ the \textbf{lift} of $\star$ along the decategorification $\lambda$.
\end{definition}
In the undeformed case, we of course recover the usual tensor product $\ostar=\otimes$ and the commutative ring $C(X)$; see \textit{Remark \ref{E1modA}}. This will be important in \S \ref{semiclassical} later.

The first part of this definition allows us to define the $\star$-deformed product between global sections of \textit{any} two measureable sheaves, and the second part states that $\ostar$ is determined up to isomorphism by this $\star$-product. 
\begin{proposition}\label{sclimit}
    The natural sheaf isomorphism \eqref{deformedtensor} gives rise to a commutative square 
\begin{equation}
    \begin{tikzcd}
	{\mathfrak{C}(\mathbb{G}^{\Gamma})\times \mathfrak{C}(\mathbb{G}^{\Gamma})} & {\mathfrak{C}(\mathbb{G}^{\Gamma})} \\
	{\mathfrak{C}(\mathbb{G}^{\Gamma})\times \mathfrak{C}(\mathbb{G}^{\Gamma})} & {\mathfrak{C}(\mathbb{G}^{\Gamma})}
	\arrow["\ostar", from=1-1, to=1-2]
	\arrow[from=1-1, to=2-1]
	\arrow["\cong", Rightarrow, from=1-1, to=2-2]
	\arrow[from=1-2, to=2-2]
	\arrow["\otimes"', from=2-1, to=2-2]
\end{tikzcd},\nonumber
\end{equation}
where the vertical arrows are given by "evaluating" at $\hbar=0$; $(-)_0: \Gamma_c(H^X)[[\hbar]] \mapsto \Gamma_c(H^X)$.
\end{proposition}
\begin{proof}
    Let $\phi=\Gamma_c(H^X)[[\hbar]]$ and $\phi'=\Gamma_c(H'^X)[[\hbar]]$ be objects in $\mathfrak{C}(\mathbb{G}^{\Gamma})$. We have the natural isomorphism
    \begin{equation}
        (\phi\ostar\phi')_0 \stackrel{\eqref{deformedtensor}}{\cong} \Gamma_c((H\otimes H')^X) \cong \Gamma_c(H^X)\otimes \Gamma_c(H'^X)=(\phi)_0\otimes (\phi')_0
    \end{equation}
    provided by the universal property \eqref{canoniso} of the tensor product.
\end{proof}

\begin{rmk}\label{E1modA}
    Let us describe a simpler incarnation of $\lambda$ in a more concrete way. Consider the Kapranov-Veovodsky model of "2-vector spaces" $\mathsf{2Vect}=\mathsf{2Vect}^{KV}$ \cite{Kapranov:1994} given by fully-fualizable \textit{finite} semisimple linear categories $C\in\mathsf{2Vect}.$ Denote by $\mathsf{sepMor}$ the Morita bicategory of separable $\bbC$-algebras, it is well-known that there is an equivalence $\mathsf{sepMor}\simeq\mathsf{2Vect}$ given by sending $A\mapsto \operatorname{Mod}(A)$ to its category of modules \cite{etingof2016tensor}. Define the following {2-functor} (the subscript "gl" stands for "global")
    \begin{equation*}
        \lambda_\text{gl}: \mathsf{2Vect}\rightarrow \mathsf{sepAlg}_\bbC,\qquad C\simeq\operatorname{Mod}(A) \mapsto A
    \end{equation*}
    which is left-adjoint to the equivalence  $\mathsf{sepMor}\simeq\mathsf{2Vect}$. The analogue of the diagram \eqref{ostardiagram} in this context is saying that $\lambda_\text{gl}$ preserves the monoidal structure of a \textit{particular} algebra object $C\in\mathsf{2Vect}$, such that
    \begin{equation*}
        \lambda_\text{gl} (-\ostar -) = \lambda_\text{gl}(-)\star \lambda_\text{gl}(-),\qquad\begin{cases}
            \ostar: C\boxtimes C\rightarrow C \\ 
           \star: A\otimes A\rightarrow A 
        \end{cases}.
    \end{equation*}
    In other words, under the the Deligne tensor product, the algebra structure $\star: A\otimes A\to A$ lifts to its modules
    \begin{equation*}
        \ostar: \operatorname{Mod}(A)\boxtimes \operatorname{Mod}(A) \simeq \operatorname{Mod}(A\otimes A)\to \operatorname{Mod}(A)\,.
    \end{equation*}
    For us, we require a "sheafy", local-systems version of $\lambda_\text{gl}$, which sends $\mathcal{O}_X$-module algebras over $X$ to the structure sheaf $\mathcal{O}_X$ itself. The roles of $\mathfrak{C}(\mathbb{G}^{\Gamma})$ and its structure sheaf $\mathcal{O}_X= C_q(X)=C(X)\otimes\bbC[[\hbar]]$, for $q=e^{i\hbar}$, correspond respectively to the particular algebra objects $C\in\mathsf{2Vect}$ and $A\in\mathsf{Alg}_\bbC$.
\end{rmk}

\begin{definition}\label{hypH}
    \textbf{Hypothesis (H)} is the assumption that there exists a limit-preserving 2-functor $\mathsf{Meas}\rightarrow \mathsf{Vect}$ (possibly infinite-dimensional vector spaces) such that
    \begin{enumerate}
        \item it induces the decategorification $\lambda:\operatorname{Cat}_\mathsf{Meas}\rightarrow \operatorname{Cat}_\mathsf{Vect}=\mathsf{2Vect}^{BC}$,
        \item for any 2-graph complex $\Gamma$, the induced map $\lambda: \C_q(\mathbb{G}^{\Gamma})\mapsto C_q(\mathbb{G}^{\Gamma})$ (i) satisfies the conditions in \textbf{Definition \ref{ostarprod}}, and (ii) preserves the coproduct $\Delta_h$ and the antipode $S_h$.
    \end{enumerate}
\end{definition}
\noindent Hypothesis (H) is required in order for us to endow a {monoidal grading} (cf. \cite{SOZER2023109155}) to a bimonoidal (co)category internal to $\mathsf{Meas}$ given by a \textit{Hopf 2-algebra} $A\in\mathsf{2Vect}^{BC}$, while allowing us to keep additivity at all levels (morphisms, objects); see \S \ref{hopf2algproblems} for more details. 


\subsubsection{Categorical $R$-matrix}\label{catRmat}
Given the above setup, we shall leverage \eqref{ostardiagram} to define a \textit{categorical $R$-matrix}, also denoted by $R\in \C_q(\mathbb{G}^{\Gamma})\times \C_q(\mathbb{G}^{\Gamma})$, such that by \textbf{Definition \ref{hypH}} $\lambda$ sends it to the 2-$R$-matrix on $C_q(\mathbb{G}^{\Gamma})$ arising from the 2-Chern-Simons action. We can then write, as a slight abuse of notation, the following
\begin{equation}
   [\phi_{(e,f)},\phi_{(e',f')}]_{\ostar} = (-\otimes -)\Big(R \ostar \Delta_h(\phi_{(e,f)\cup_{h}(e',f')}) - \Delta^\text{op}_\mathrm{h}(\phi_{(e,f)\cup_{h}(e',f')})\ostar R^T\Big),\label{2quantumproduct-heis}
\end{equation}
to denote the \textit{Hilbert spaces} of sections obtained from localized 2-graph operators $\phi=\Gamma_c(H^X)[[\hbar]]$.

\medskip

We now consider some coherence conditions satisfied by the $R$-matrix $R$.  First, we will impose the natural compatibility against the pullbacks $\hat s^*,\hat t^*$ of the source and target maps $\hat s,\hat t: (\mathsf{H}\rtimes G)^{\Gamma^2}\rightrightarrows G^{\Gamma^1}$ on the 2-holonomies $\mathbb{G}^{\Gamma}$. By taking the following components
\begin{equation*}
    R = R|_{\C_q((\mathsf{H}\rtimes G)^{\Gamma^2})^{\times 2}},\qquad R_0=R|_{\C_q(G^{\Gamma^1})^{\times 2}}
\end{equation*}
of the higher $R$-matrix, this naturality condition is expressed as 
\begin{align}
    & (R\ostar-)\circ (\hat s^*\times \hat s^*)  =  (\hat s^*\times \hat s^*)\circ (R_0\ostar -),\nonumber\\
    & (R\ostar -)\circ (\hat t^*\times \hat t^*)=(\hat t^*\times \hat t^*) \circ (R_0\ostar-), \label{catRnaturality}
\end{align}
where we have abused notation and denote by the canonical monoidal structure on $\C_q(\mathbb{G}^{\Gamma})\times \C_q(\mathbb{G}^{\Gamma})$ also by $\ostar$. 



Recall the strictly coassociative and cointerchanging 2-graph coproducts on $\mathfrak{C}(\mathbb{G}^{\Gamma})$ introduced in \textbf{Definition \ref{2graphfusion}}. Their quantum versions, denoted also by $\Delta_h,\Delta_v$, must then satisfy the following \textit{intertwining relations} against the higher $R$-matrix: there exist natural sheaf isomorphisms
\begin{gather}
    (\sigma\circ\Delta_h)(\phi)\ostar R \cong R\ostar\Delta_h(\phi),\label{quantumR}\\
    \Delta_{\mathrm{v}}\circ (R\ostar -) \cong ((R\ostar -)\times (R\ostar -)) \circ\Delta_v\label{naturality}
\end{gather}
where $\sigma:\mathfrak{C}(\mathbb{G}^{\Gamma})\times\mathfrak{C}(\mathbb{G}^{\Gamma})\xrightarrow{\sim} \mathfrak{C}(\mathbb{G}^{\Gamma})\times\mathfrak{C}(\mathbb{G}^{\Gamma})$ is a swap of tensor factors.


The compatibility of these $R$-matrices with the cointerchange \eqref{cointerchange} is captured by the commutative diagrams
\begin{equation}
\begin{tikzcd}
	{(\phi_{(1)(1)}\times_h\phi_{(1)(2)})\times_v(\phi_{(2)(2)}\times_h\phi_{(2)(1)})} & {(\phi_{(1)(1)}\times_v\phi_{(2)(2)})\times_h(\phi_{(1)(2)}\times_v\phi_{(2)(1)})} \\
	{(\phi_{(1)(1)}\times_h\phi_{(1)(2)})\times_v(\phi_{(2)(1)}\times_h\phi_{(2)(2)})} & \\
	{(\phi_{(1)(2)}\times_h\phi_{(1)(1)})\times_v(\phi_{(2)(1)}\times_h\phi_{(2)(2)})} &{(\phi_{(1)(2)}\times_v\phi_{(2)(1)})\times_h(\phi_{(1)(1)}\times_v\phi_{(2)(2)})} 
	\arrow["{\beta_{12;43}}", from=1-1, to=1-2]
	\arrow["{1\times R^{3;4}}", from=2-1, to=1-1]
	\arrow["{R^{1;2}\times1}"', from=2-1, to=3-1]
	\arrow["{R^{23;41}}"', from=3-2, to=1-2]
	\arrow["{\beta_{21;34}}"', from=3-1, to=3-2]
\end{tikzcd},\label{hexagon1}
\end{equation}
\begin{equation}
\begin{tikzcd}
	{(\phi_{(1)(1)}\times_v\phi_{(2)(2)})\times_h(\phi_{(1)(2)}\times_v\phi_{(2)(1)})} & {(\phi_{(2)(2)}\times_v\phi_{(1)(1)})\times_h(\phi_{(1)(2)}\times_v\phi_{(2)(1)})} \\
	{(\phi_{(1)(2)}\times_v\phi_{(2)(1)})\times_h(\phi_{(1)(1)}\times_v\phi_{(2)(2)})} & 
    {(\phi_{(1)(2)}\times_v\phi_{(2)(1)})\times_h(\phi_{(2)(2)}\times_v\phi_{(1)(1)})}
	\arrow["{(-^\text{op})^{1;4}}", from=1-1, to=1-2]
	\arrow["{R^{23;41}}"', from=2-1, to=1-1]
	\arrow["{ (-^\text{op})^{1;4}}"', from=2-1, to=2-2]
    \arrow["{R^{23;14}}",from=2-2,to=1-2]
\end{tikzcd},\label{hexagon}
\end{equation}
Here, $-^\text{op}$ denotes the \textit{opposite cocateory} --- namely a swap of the tensor factors in the summands of the vertical coproduct, and we have used the shorthand \eqref{shorthand} to write $$(\Delta_h\otimes\Delta_h)\Delta_{\mathrm{v}}(\phi)=\Delta_h(\phi_{(1)})\times_v \Delta_h(\phi_{(2)}) = (\phi_{(1)(1)}\times_h\phi_{(1)(2)})\otimes_{\mathrm{v}}(\phi_{(2)(1)}\times_h\phi_{(2)(2)}),\qquad\text{etc.}$$ for the coproducts. These diagrams come with dual diagrams with the $\mathrm{h,v}$ swapped.

The arrows labelled by "$R^{i;j}$" implements a conjugation by the $R$-matrix \eqref{quantumR} on the $i,j$-th factor. The $\beta$'s denote the witness for the cointerchange law \eqref{cointerchange} on the 2-graph operators, which we recall can be trivialized by going on-shell of the 2-flatness condition. We will prove in \textbf{Lemma \ref{bimon}} that these quantum deformed coproducts \eqref{quantumR} are compatible in a Hopf categorical sense with a deformed monoidal structure $\otimes_q=\ostar$ on $\mathfrak{C}(\mathbb{G}^{\Gamma})$.

Note the property of having {\it two} types of (co)products and one product is shared by {\it cotrialgebras} \cite{Pfeiffer2007} (mentioned also in fig. \ref{fig:1}). However, here we have much more structure: this is the subject of the following section \S \ref{hopfopalg}.

\subsection{Higher Hopf structures on the 2-graph operators}\label{hopfopalg}
Given the above quantum deformed corpdoucts and $R$-matrices, we now investigate the structure of the 2-graph operators $\mathfrak{C}(\mathbb{G}^{\Gamma})$. Since these were induced through dualization directly from the 2-groupoid structure of the 2-group $\mathbb{G}$ or the geometry of the 2-graphs $\Gamma$. 

\subsubsection{As a Hopf internal category}
We first fix the definitions, then we get to work. 

\paragraph{Internal categories.} 
\begin{definition}\label{catinternal}
    A \textbf{category $C$ internal to $\cV$} is a strict category object in a bicategory $\cV$ with pushouts and pullbacks (such as $\cC=\mathsf{Meas}$). It consists of the data:
\begin{itemize}
    \item a pair of objects $C_1,C_0\in\cC$,
    \item a pair of \textit{fibrant} 1-morphisms $s,t: C_1\rightarrow C_0$ in $\cV$ called the \textit{source/target}, and their pullback $C_1\,_t\times_sC_1$,
    \item a 1-morphism $\circ: C_1\,_t\times_sC_1\rightarrow C_1$ in $\cV$, called the \textit{composition law}, and
    \item a 1-morphism $\eta:C_0\rightarrow C_1$, called the \textit{unit}, such that
    \begin{enumerate}
        \item the composition law $\circ$ is strictly associative: the 2-morphism
\begin{tikzcd}
	{C_1\times_{C_0}C_1\times_{C_0}C_1} & {C_1\times_{C_0}C_1} \\
	{C_1\times_{C_0}C_1} & {C_1}
	\arrow["{\id\times \circ}", from=1-1, to=1-2]
	\arrow["{\circ\times \id}"', from=1-1, to=2-1]
	\arrow["\cong", shorten <=8pt, shorten >=8pt, Rightarrow, from=1-1, to=2-2]
	\arrow["\circ", from=1-2, to=2-2]
	\arrow["\circ"', from=2-1, to=2-2]
\end{tikzcd}
is invertible, 
        \item $\circ,1$ satisfy strict unity: for each $f\in C_1$ with $s(f) = x$ and $t(f)=y$, we have invertible 2-morphisms $1_y\circ f\cong f \cong f\circ 1_x$,
        \item the invertible compositional unitors and associators satisfy 
        \begin{enumerate}
            \item the exchange equation (which we call the \textit{interchange law}), 
            \item the left- and right-pentagon equations, and
            \item the left-, middle- and right-triangle equations,
        \end{enumerate}
        on the pullbacks $C_1^{[n]}=C_1\times_{C_0}C_1\times_{C_0}\dots\times_{C_0}C_1$. 
\end{enumerate}
\end{itemize}
A \textbf{cocategory $D$ internal to $\cV$} is a strict category object in $\cV^\text{op}$. It is equipped with \textit{cofibrant} functors $u,v: D_0\rightarrow D_1$, a strict counit $\epsilon:D_1\rightarrow D_0$ and a strictly coassociative cocomposition law $\Delta_v: D_1\rightarrow D_1~_v\times_u D_1$ along the pushout.
\end{definition}
Keep in mind that internal categories do not have cocompositions, and cocategories do not have compositions.

A (strict) functor $F:C\rightarrow D$ of categories internal to $\cV$ is of course a pair of 1-morphisms $F_{i}:C_i\rightarrow D_i$ for $i=0,1$, equipped with invertible 2-morphisms
\[\begin{tikzcd}
	{C_1} && {D_1} \\
	& \cong \\
	{C_0} && {D_0}
	\arrow["{F_1}", from=1-1, to=1-3]
	\arrow["{s_C}"', shift right, curve={height=6pt}, from=1-1, to=3-1]
	\arrow["{t_C}", shift left, curve={height=-6pt}, from=1-1, to=3-1]
	\arrow["{s_D}"', shift right, curve={height=6pt}, from=1-3, to=3-3]
	\arrow["{t_D}", shift left, curve={height=-6pt}, from=1-3, to=3-3]
	\arrow["{F_0}", from=3-1, to=3-3]
\end{tikzcd},\qquad F(\circ) \cong \circ (F\times F),\qquad F_1 \eta_C \cong \eta_D F_0\]
which ensures that $F$ commutes with the fibrant source/target maps and the composition.

Let $\operatorname{Cat}_{\cV},\operatorname{Cocat}_{\cV}$ denote the categories/cocategories {internal} to $\cV$, respectively, then the canonical equivalence $\cV\simeq \cV^\text{op}$ induces $$\operatorname{Cocat}_\cV \simeq \operatorname{Cat}_{\cV^\text{op}},$$ which will play a big role in this paper.

\paragraph{Hopf internal categories.} Now suppose $\cV$ is symmetric monoidal, with a monoidal unit object $I\in\cV$. As an abuse of notation, we will also denote by $I$ its discrete category $I\rightrightarrows I$ internal to $\cV$.
\begin{definition}\label{internalhopf}
    Let $(\cV,\times,I)$ be a ($\bbC$-linear) symmetric monoidal 2-category. A \textbf{(strict) Hopf monoidal category $\cH$ in $\cV$} is a Hopf algebra object in $\operatorname{Cat}_{\cV}$. Namely, it is equipped with the following internal functors:
    \begin{enumerate}
        \item the \textit{product} $\otimes: \cH\times\cH\rightarrow\cH$ (with a unit $\iota\in \cH$),
        \item the strictly monoidal \textit{coproduct} $\Delta:\cH\rightarrow\cH\times\cH$ (with a counit $\epsilon: \cH\rightarrow I$), and
        \item the strictly op-comonoidal op-monoidal \textit{antipode} $S:\cH\rightarrow \cH^{\text{m-op},\text{c-op}}$,
    \end{enumerate}
    as well as internal natural transformations
    \begin{enumerate}
        \item the \textit{associators} $a^\otimes: \otimes\circ (\otimes\times1_\cH) \Rightarrow \otimes\circ(1_\cH\times\otimes)$ and \textit{unitors} $r^\otimes: (-\otimes \iota)\Rightarrow 1_\cH,~\ell^\otimes: (\iota\otimes -)\rightarrow 1_\cH$ satisfying the strict pentagon and triangle axioms, 
        \item the \textit{coassociators} $a^\Delta: (\Delta\times 1_\cH)\circ\Delta\Rightarrow (1_\cH\times\Delta)\Rightarrow\Delta$ and \textit{counitors} $r^\Delta:( \epsilon\times 1_\cH)\circ\Delta \Rightarrow1_\cH ,~ \ell^\Delta: (1_\cH\times \epsilon)\circ\Delta\Rightarrow 1_\cH$ satisfying the strict copentagon and cotriangle) axioms,
        \item the invertible \textit{bimonoidal natural transformations}
        \begin{equation}
            \Delta \circ \otimes \cong (1_\cH\times \sigma\times 1_\cH) \circ (\otimes\times\otimes)\circ \Delta\label{bimonoid}
        \end{equation}
        \item the \textit{antipode relations}
        \begin{equation}
            \otimes\circ (S\times 1_\cH)\circ\Delta \cong \iota\otimes\epsilon\cong \otimes\circ(1_\cH\times S)\circ\Delta,\label{antpode}
        \end{equation}
    \end{enumerate}
    such that these internal natural transformations are mutually coherent.
    
    A \textbf{(strict) Hopf comonoidal cocategory in $\cV$} is a (strict) Hopf monoidal category in $\cV^\text{op}$.
\end{definition}

The main example we shall consider in this paper is the symmetric monoidal 2-category $\cV=\mathsf{Meas}$ of Crane-Yetter measureable categories \cite{Crane:2003gk,Yetter2003MeasurableC,Baez:2012}.


\medskip

Recall that the bicategory $\mathsf{Meas}$ of Crane-Yetter measureable categories \cite{Crane:2003gk,Yetter2003MeasurableC} can be understood as $W^*$-modules over $L^\infty(X,\mu)$, which is a commutative von Neumann algebra.
\begin{definition}\label{measqdef}
     Denote by $\mathsf{Meas}_q$ the bicategory of $W^*$-module categories over the non-commutative von Neumann algebra $\mathbb{C}[[\hbar]]\otimes_\bbC L^\infty(X,\mu)$ --- namely, the objects are categories of measuareable sheaves, whose global sections are also power series in $\hbar$. 
\end{definition}
\noindent Similar framework has been used in \cite{Kristel:2023gus}, but for us the extent of non-commutativity of our von Neumann algebras is strictly controlled by the power series in $\hbar$. It can be thought of as the infinite-dimensional version of $\mathsf{2Vect}$ (see \textit{Remark \ref{E1modA}}), or as a sheafy measureable version of the bicommutant categories of Henriques-Pennys \cite{Henriques2017-gm}.

\medskip

We now work to prove the main theorem by breaking it up into a few lemmas. 
\begin{lemma}\label{cocat}
   $\C_q(\mathbb{G}^{\Gamma})$ is an additive comonoidal cocategory internal to $\mathsf{Meas}_q$, with a coproduct and cocomposition.
\end{lemma}
\begin{proof}
 Additivity follows from additivity of the target $\mathsf{Hilb}$. Let $\hat s,\hat t$ denote the source and target maps in the 2-groupoid $\mathbb{G}^{\Gamma}=\operatorname{Fun}(\Gamma^2,\mathbf{B}\mathbb{G})$ as the confluence of those in $\mathbb{G}$ and $\Gamma^2$. This confluence is well-defined on-shell of the fake-flatness condition 
    \begin{equation}
        h_{\partial f} = \mathsf{t}(b_f),\qquad \forall~ (h_e,b_f)\in\mathbb{G}^{\Gamma}.\label{fakfeflat}
    \end{equation}
 
 Since $\hat s,\hat t$ are by hypothesis surjective submersive \cite{MACKENZIE200046,Chen:2012gz}, their pullbacks induce measureable functors \cite{Baez:2012} which are cofibrant
    \begin{equation*}
        \C_q((\mathsf{H}\rtimes G)^{\Gamma^2})\xleftarrow{\hat s^*}\C_q(G^{\Gamma^1}) \xrightarrow{\hat t^*} \C_q((\mathsf{H}\rtimes G)^{\Gamma^2})
    \end{equation*}
    with a coidentity $\epsilon: \C_q((\mathsf{H}\rtimes G)^{\Gamma^2})\to \C_q(G^{\Gamma^1})$ given by pulling back the unit section $\id: g\mapsto (g,1)$ on $\mathbb{G}$. The coidentity represents the trivial face-localized state.       
    
    The cocomposition is given by the vertical coproduct $\Delta_v:\tilde{\C}_q((\mathsf{H}\rtimes G)^{\Gamma^2}) \rightarrow  \C_q((\mathsf{H}\rtimes G)^{\Gamma^2}) \times \C_q((\mathsf{H}\rtimes G)^{\Gamma^2})$, which inserts a face-localized state in between two edge-localized states. The comonoidal coproduct is an additive measureable functor $\Delta_h: \tilde{\C}_q(\mathbb{G}^{\Gamma})\rightarrow\C_q(\mathbb{G}^{\Gamma})\times \C_q(\mathbb{G}^{\Gamma})$ given by the horizontal coproduct. 
 



    By \textbf{Proposition \ref{cointerchangeprop}},  $\mathfrak{C}(\mathbb{G}^{\Gamma})$ is an additive comonoidal cocategory internal to $\mathsf{Meas}_q$.
\end{proof}


The next step is to prove that the quantum deformed product $\ostar$ on $\mathfrak{C}(\mathbb{G}^{\Gamma})$ satisfies \eqref{bimonoid}. We will check this condition is first satisfied classically, then check that the combinatorial 2-Fock-Rosly Poisson brackets \eqref{2fockrosly} are compatible.

\begin{lemma}\label{bimon}
    Provided the map $\lambda$ preserves the coproducts, then \eqref{bimonoid} holds for $\mathfrak{C}(\mathbb{G}^{\Gamma})$.
\end{lemma}
\begin{proof}
    We first establish \eqref{bimonoid} classically. On localized sections, properties of the tensor product $\otimes$ and the geometry of the 2-holonomies then provide sheaf isomorphisms
        \begin{align*}
        (\Delta(\xi \otimes\xi'))(h_e,b_f) &= \sum (\xi \otimes\xi')_{(1)}(h_{e_1},b_{f_1}) \times (\xi \otimes\xi')_{(2)}(h_{e_2},b_{f_2}) \\
        &\mapsto \sum ((\xi_{(h_{e_1},b_{f_1})})_{(1)} \otimes (\xi'_{(h_{e_1},b_{f_1})})_{(1)}) \times ((\xi_{(h_{e_2},b_{f_2})})_{(2)}\otimes (\xi'_{(h_{e_2},b_{f_2})})_{(2)}) \\
        &\mapsto \sum ((\xi_{(h_{e_1},b_{f_1})})_{(1)}\times (\xi_{(h_{e_2},b_{f_2})})_{(2)})\otimes\sum ((\xi'_{(h_{e_1},b_{f_1})})_{(1)} \times (\xi'_{(h_{e_2},b_{f_2})})_{(2)}) \\
        &= (\Delta\xi)(h_e,b_f) \otimes (\Delta\xi')(h_e,b_f) = (\Delta\xi\otimes\Delta\xi')(h_e,b_f)
    \end{align*}    
    for any sections $\xi,\xi'$ of sheaves $\phi,\phi\in\mathfrak{C}(\mathbb{G}^{\Gamma})$ localized on $(e,f)\in\Gamma^2$. Here $\Delta$ can mean either the horizontal or vertical 2-graph coproduct.
    

     
    Now semiclassically, the classical 2-Yang-Baxter equation (leading to the Lie 2-algebra cocycle condition for $\delta = dr\in Z^2(\G,\G\wedge\G)$ \cite{Bai_2013}) directly implies the multiplicativity \cite{Chen:2012gz}
    \begin{equation*}
            \Delta_h(\{\xi,\xi'\}) = \{\Delta_h(\xi),\Delta_h(\xi')\},\qquad \Delta_v(\{\xi,\xi'\}) = \{\Delta_v(\xi),\Delta_v(\xi')\}
        \end{equation*}
    of the Poisson bracket \eqref{2fockrosly} against the group/groupoid multiplication in $\mathbb{G}^{\Gamma}$, where we have implicitly applied the sheaf isomorphism \eqref{deformedtensor}. Quantizing lifts to the multiplicativity of the coproducts $\Delta_\mathrm{h,v}$ against the $\star$-product on measureable sections,
    \begin{equation*}
        \Delta_h(\xi\star\xi') = \Delta_h(\xi)\star \Delta_h(\xi'),\qquad \Delta_v(\xi\star \xi') = \Delta_v(\xi)\star \Delta_v(\xi').
    \end{equation*}
    If Hypothesis (H) holds, namely that $\lambda$ preserves the coproducts, then this multiplicativity lifts by \textbf{Definition \ref{ostarprod}} to $\ostar$ and \eqref{bimonoid} follows.
        
        

\end{proof}

Given the above result, $\mathfrak{C}_q(\mathbb{G}^{\Gamma})$ then gives rise to a \textit{non-symmetric} cocategory.
\begin{definition}\label{cobraidingdef}
    Let $\cH$ denote a Hopf cocategory with the comonoidal coproduct functor $\Delta_h: \cH\rightarrow \cH\times\cH$. A(n invertible) \textbf{cobraiding} $(\mathcal{R},\mathsf{T})$ on $\cH$ is an additive bimonoidal automorphism $\mathcal{R}:\cH\times\cH\rightarrow \cH\times\cH$ equipped with a natural transformation
    \begin{equation*}
        \mathsf{T}: \mathcal{R}\circ\Delta_h\Rightarrow \Delta_h^\text{op},
    \end{equation*}
    satisfying various Hopf coherence conditions against the higher morphisms in \textbf{Definition \ref{internalhopf}}.
\end{definition}
\noindent We finally can state the following.
\begin{theorem}\label{hopfcat}
    $\mathfrak{C}_q(\mathbb{G}^{\Gamma})$ is a \textbf{cobraided} Hopf cocategory internal to $\mathsf{Meas}_q$.
\end{theorem}
\begin{proof}
\begin{itemize}
    \item \textbf{The antipode}: Given 2-$\dagger$-unitarity \textbf{Definition \ref{2dagger}}, the antipode functors $S_{h,v}$ promoted to the quantum 2-graph operators $\mathfrak{C}_q(\mathbb{G}^{\Gamma})$ in \S \ref{2dagger}. The antipode axioms \eqref{antpode} then follow directly from the underlying geometry of $\Gamma^2$. 

    Together with \eqref{quantumR}, one can also verify that there exist (internal) natural transformations witnessing the following identity:
    \begin{equation}
        (S_h\times 1)R ~\hat\ostar~ R \cong  (1\times S_h)R \cong  1_{\bbC}\otimes 1_\bbC,\label{leftadj}
    \end{equation}
    where $1_\bbC\in\C_q(\mathbb{G}^{\Gamma})$ denotes the monoidal unit represented by the trivial measuereable field over $(\mathbb{G}^{\Gamma},\mu_{\Gamma})$.
    
    \item \textbf{The cobraiding}: We define the cobraiding $(\mathcal{R},\mathsf{T})$ on $\mathfrak{C}(\mathbb{G}^{\Gamma})=\mathfrak{C}_q(\mathbb{G}^{\Gamma})$ to be the following. Suppressing the strict associators, the automorphism $\mathcal{R} = \operatorname{ad}_{R}$ is given by a \textit{conjugation}
    \begin{equation*}
        \mathcal{R}(\phi\times\phi') = R\ostar (\phi\times\phi')\ostar R^{-1},\qquad \phi,\phi'\in\mathfrak{C}(\mathbb{G}^{\Gamma})
    \end{equation*}
    by the invertible (with respect to $\ostar^{\times2}$ on $\mathfrak{C}(\mathbb{G}^{\Gamma})\times\mathfrak{C}(\mathbb{G}^{\Gamma})$) higher $R$-matrix $R$. Each component $\mathsf{T}_\phi$ of the natural transformation are then given by \eqref{quantumR}.
    

    By definition, $\mathcal{R}$ is monoidal. We thus now need to prove the following:
    \begin{enumerate}
        \item $\mathcal{R}=\operatorname{ad}_{R}: \mathfrak{C}(\mathbb{G}^{\Gamma})\times\mathfrak{C}(\mathbb{G}^{\Gamma})\rightarrow \mathfrak{C}(\mathbb{G}^{\Gamma})\times\mathfrak{C}(\mathbb{G}^{\Gamma})$ defines a functor between cocategories, and
        \item multiplication by $R$ is comonoidal, ie. compatible with $\Delta_h$.
    \end{enumerate}
    For the first point, we need $\mathcal{R}$ to (i) commute with the cosource/cotarget maps on $\mathfrak{C}_q(\mathbb{G}^{\Gamma})$ and (ii) is natural with respect to the cocomposition $\Delta_v$. (i) is guaranteed by \eqref{catRnaturality}, while (ii) is guaranteed by \eqref{naturality} and the diagram \eqref{hexagon}. 
    

    Now for the second point, we require higher quasitriangularity relations for $R$; namely that there exist isomorphisms in $\C_q(\mathbb{G}^{\Gamma})^{\times 3}$,
    \begin{equation}
        Q_\ell: (1\times\Delta_h)R \cong R_{13}\ostar R_{23},\qquad Q_r: (\Delta_h\times 1)R\cong R_{13}\ostar R_{12},\label{higherquasi}
    \end{equation}
    which are consistent with the $\Delta_h$-coassociator. These follow from \eqref{hexagon1} (see also \textbf{Lemma \ref{Requivlemma}} later).
    
\end{itemize}

\end{proof}
\noindent By construction, the decategorification $\lambda$ of \textbf{Definition \ref{ostarprod}} simply restricts to the structure sheaf $C_q(\mathbb{G})$ and trivializes the natural transformation $\mathsf{T}$.

\begin{definition}
    The \textbf{categorical quantum coordinate ring} $\mathfrak{C}_q(\mathbb{G})$ is the 2-graph operators $\mathfrak{C}_q(\mathbb{G}^{\Gamma})$ on the 2-graph $\Gamma$ consisting of only a single face $f:e\rightarrow e$, a single edge loop $v\xrightarrow{e}v\in\Gamma^1$ and a single vertex $v\in \Gamma^0$.
\end{definition}
\noindent See also \textit{Remark \ref{not-2group}}. 

\begin{rmk}\label{quantumcobraiding}
    Let us elaborate a bit more on the cobraiding. Suppose $\mathbb{G}$ is a strict Lie 2-group such that we can suppress the co/associators. For each $\phi\in\C_q(\mathbb{G})$, a cobraiding $(\mathcal{R},\mathsf{T})$ as above is given by $\mathcal{R}(\phi\times\phi') = R\ostar (\phi\times\phi')\ostar R^{-1}$, together with measureable natural isomorphisms
    \begin{equation*}
        \mathsf{T}_\phi: R\ostar \Delta(\phi)\ostar R^{-1} \xrightarrow{\sim} \Delta^\text{op}(\phi).
    \end{equation*}
    The components $\mathsf{T}_{R_{(1)}},\mathsf{T}_{R_{(2)}}$, together with the isomorphisms $Q_\ell,Q_r$ \eqref{higherquasi}, give rise to invertible natural transformations 
    \begin{equation*}
            R_{23}\ostar (R_{13} \ostar R_{12}) \cong (R_{12}\ostar R_{13}) \ostar R_{23}
    \end{equation*}
    which witness the Yang-Baxter equations (cf. \eqref{2yb}). Such Yang-Baxter intertwiners are known \cite{Kuniba_2023} to be very closely related to solutions of the {\bf Zamolodchikov tetrahedron equations} \cite{Zamolodchikov:1980,Kapranov:1994}.
\end{rmk}

We note that the generalization to the weakly-associative case can be carried out directly, by keeping track of the appearance of $\tau$ and (its first descendant mentioned in \S \ref{descendant}) through natural ismorphisms. We expect to obtain a Hopf cocategory as well in this case, but with coassociative and cointerchange witnesses which need not only have identity components.

\subsubsection{The semiclassical limit}\label{semiclassical}
The goal of this subsection is to show that the cobraiding on the 2-graph operators comes from a \textit{quantum 2-$R$-matrix} on  $C_q(X)$, as a \textit{(quasi)triangular} Hopf 2-algebra as in \S \ref{coprod} and \cite{Chen:2023tjf}. We will use this to determine the semiclassical limit.


We now work to recover $R$ as a 2-$R$-matrix in the sense given in \S \ref{coprod}.
\begin{lemma}\label{Requivlemma}
    Provided $\lambda$ preserves the bimonoidal structure, the cobraiding on $\C_q(\mathbb{G}^{\Gamma})$ reduces to a 2-$R$-matrix on $C_q(\mathbb{G}^{\Gamma})$. 
\end{lemma}
\begin{proof}
    Without loss of generality, we work directly with $\Gamma^2$ given by the fundamental 2-graph. We prove that the component of $R$ restricted to the structure sheaf $\mathcal{O}_{\mathbb{G}} = C_q(\mathbb{G})$ is a quantum 2-$R$-matrix $R$ for the Hopf 2-algebra
    \begin{equation*}
        A = A_1= C_q(\mathsf{H})\xrightarrow{\mathsf{t}^*}A_0=C_q(G),
    \end{equation*}
where $\mathsf{t}^*$ is the pullback of the Lie group crossed-module map $\mathsf{t}: \mathsf{H}\rightarrow G$.

From the quantum $R$-matrix $\lambda(R)$, we define $\Delta'=\lambda(\Delta)$ and
    \begin{equation}
        R' = \lambda(R)|_{A_1\otimes A_0}\oplus \lambda(R)|_{A_0\otimes A_1} \equiv R^l \oplus R^r, \label{twoRs}
    \end{equation}
    such that $R' \in (C_q(\mathbb{G}^{\Gamma}))_1^{\otimes 2}$ is an element of total degree-1.

    Clearly \eqref{catRmat} implies the intertwining property of $R$. Moreover, given the characterization \textbf{Proposition \ref{2chain}}, \eqref{catRnaturality} reduces to \eqref{Rmatequiv}. It now remains to recover the quasitriangularity conditions. To do so, we apply $\lambda$ to the diagrams \eqref{hexagon}, \eqref{hexagon1}. 

    Consider the arrow $R^{23;14}_\mathrm{h}$ in \eqref{hexagon1}. It involves the quantity $(\Delta_h\otimes\Delta_h)R$, while the expressions we want are
    \begin{equation*}
        (1\otimes\Delta_h)R,\qquad (\Delta_h\otimes 1)R.
    \end{equation*}
    These can be computed from various contractions in \eqref{hexagon1}. Indeed, by contracting the tensor legs labelled by $2,3$ with $\ostar$ and putting $\phi_{(2)}\times \phi_3 =\phi$, we obtain the commutative diagram,
\[\begin{tikzcd}
	{\phi_{(1)}\times_h(\phi\times_h\phi_4)} \\
	{(\phi_{(1)}\times_h\phi_4)\times_h\phi} && {\phi\times_h(\phi_{(1)}\times_h\phi_4)}
	\arrow["{R^{13}}", from=1-1, to=2-1]
	\arrow["{R^{12}}"', from=2-3, to=1-1]
	\arrow["{(\Delta_h\otimes1)R}", from=2-3, to=2-1]
\end{tikzcd}\]
which gives, under $\lambda$ (here the superscripts on $R$ indicate which tensor factor the $R$-matrices act on, not the subscripts of the $\phi$'s),
\begin{equation*}
    (\Delta'\otimes 1)R' = R'^{13}\otimes R'^{12}.\label{quasi1}
\end{equation*}
Similarly we have
\begin{equation}
    (1\otimes\Delta' )R' = R'^{13}\otimes R'^{23}.\label{quasi2}
\end{equation}
These are nothing but the quasitriangularity conditions.
\end{proof}

The conditions \eqref{quasi1}, \eqref{quasi2} and \eqref{quantumR} were shown in {Proposition 3.13} of \cite{Chen:2023tjf} to be equivalent to the "2-Yang-Baxter equations"
\begin{equation}
    R'^{23}(R'^{13} R'^{12}) = (R'^{12} R'^{13}) R'^{23}\label{2yb}
\end{equation}
on $C_q(\mathbb{G}^{\Gamma})$.

We can now prove that $\C_q(\mathbb{G})$ admits the Lie 2-bialgebra $(\G;\delta)$ underlying the 2-Chern-Simons action \cite{Chen:2022hct} as its semiclassical limit, with $\G=\operatorname{Lie}\mathbb{G}$. Here, by a "semiclassical limit", we mean taking the "quantum deformation parameter(s)" $q\sim 1+\hbar=1+\frac{2\pi}{k}$ to first order and applying the decategorification $\lambda$.

\begin{theorem}
    Suppose $\lambda$ satisfies hypothesis (H), then the semiclassical limit of $\mathfrak{C}_{q}(\mathbb{G})$ is dual to the Lie 2-bialgebra $(\G=\operatorname{Lie}\mathbb{G};\delta)$.
\end{theorem}
\begin{proof}
    By \textbf{Lemma \ref{Requivlemma}}, $\lambda$ sends the categorical $R$-matrix on $\C_q(\mathbb{G})$ to a 2-$R$-matrix on $C_q(\mathbb{G})$. The statement then follows upon taking the limit $q\rightarrow 1$ by the results of Appendix B in \cite{Chen:2023tjf}.

\end{proof}
\noindent In other words, through hypothesis (H), the 2-graph operators described above does indeed give a quantization of 2-Chern-Simons theory on the lattice.

\begin{rmk}\label{hypHconstruction}
    We emphasize here that the notion of a Hopf (co)category, possibly equipped with a cobraiding, that we have defined in \S \ref{hopfopalg} is completely independent of Hypothesis (H). \textbf{Definition \ref{hypH}} serves to relate the quantum 2-holonomy operators described by the Hopf category to the higher-gauge fields described by the Hopf 2-algebra. If one is content with constructing TQFTs directly from the 2-holonomies/Hopf category, then one would not need to invoke Hypothesis (H) at all.
\end{rmk}



\subsubsection{The 4d categorical ladder}

The above set of results is a realization of the categorical ladder proposal \cite{Pfeiffer2007,Baez:1995xq,Crane:1994ty}, which states that 4d TQFTs are determined by a Hopf monoidal category; see fig. \ref{fig:1}. In this context, the Hopf category $\mathfrak{C}(\mathbb{G}^{\Gamma})$ plays a role analogous to the Hopf algebra of Chern-Simons observables defined on the lattice in \cite{Alekseev:1994au}. 

\begin{figure}
    \centering
    \includegraphics[width=0.8\linewidth]{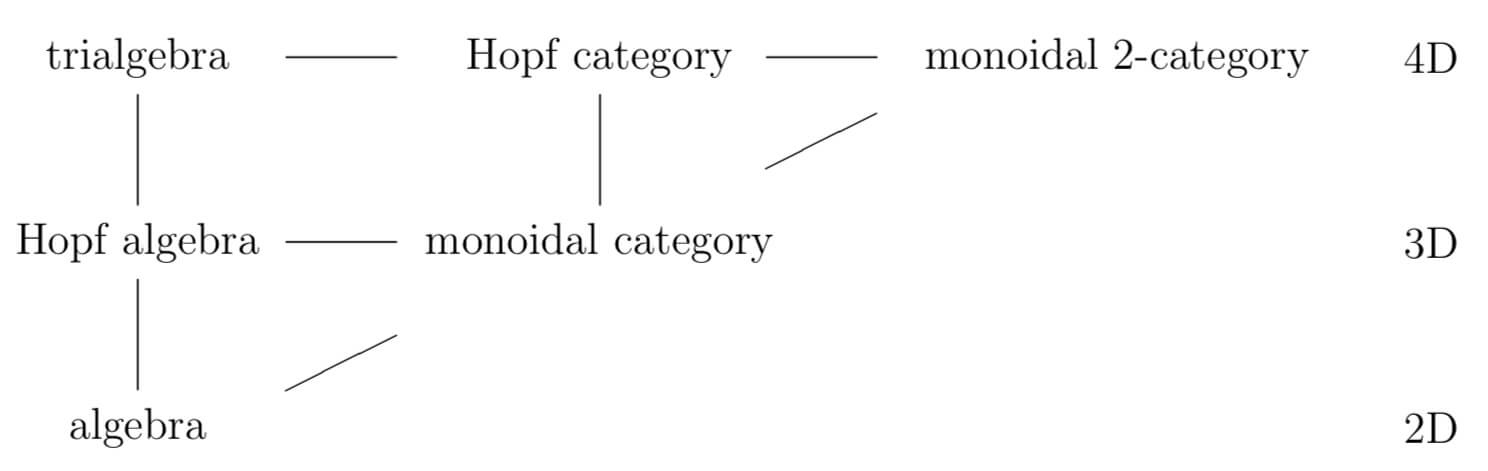}
    \caption{The categorical ladder as proposed in \cite{Pfeiffer2007,Baez:1995xq,Crane:1994ty}, which gives a prescription for how the observables in a higher-dimensional TQFT should behave. Here, the vertical axis is the dimension and the horizontal axis denotes the operation of taking modules.}
    \label{fig:1}
\end{figure}

As such, in analogy with the seminal work of Witten \cite{WITTEN1990285}, observables of 2-Chern-Simons theory (modulo 2-gauge transformations) should then be described by an assignment of elements of the categorified quantum coordinate ring $\mathfrak{C}_q(\mathbb{G})$ to the so-called "2-ribbons" \cite{BAEZ2003705,CARTER19971,kapovich1996} --- namely surfaces $\Sigma$ bounding incoming and outgoing links/tangles --- this will be made precise in the next paper in the series. 

\medskip

The key question now is to characterize the \textit{quantum 2-gauge transformations}, denoted by $\tilde{\cC}=\mathbb{G}^{\Gamma^1}$, on the 2-graph operators; we will show in \S \ref{regrep} that, under appropriate coherence conditions, $\mathfrak{C}_q(\mathbb{G}^{\Gamma})$ is a \textit{monoidal} measureable module over $\tilde{\cC}=\mathbb{G}^{\Gamma^1}$. The companion paper \cite{Chen:2025?} then extracts a {\it canonical higher ribbon} structure out of the corresponding braided tensor 2-representation 2-category $\operatorname{2Rep}(\tilde{\cC};\tilde R)$.

The 2-tangle hypothesis \cite{BAEZ2003705} then implies that such an assignment completes --- {\` a} la, for instance, a 4-dimensional version of handlebody surgery \cite{Freedman:1982,Casson:1986,Bennett_2016} --- into a \textbf{2-Chern-Simons TQFT}
\begin{equation}
    Z_{2CS}^\mathbb{G}: \operatorname{Bord}^O_{\langle 4,3\rangle+\epsilon}\rightarrow\mathsf{Vect}.\label{d2CS}
\end{equation}
This is a direct categorification of the Reshetikhin-Turaev construction \cite{Reshetikhin:1990pr,Reshetikhin:1991tc,Turaev:1992hq,Turaev:1992} to 4-dimensions. The ultimate goal of this project is to construct the functor \eqref{d2CS}.





\section{Categorical quantum 2-gauge transformations}\label{quantum2gautransfo}
Equipped with the knowledge that 2-graph operators form various interrelated Hopf structures, we are going to introduce a Hopf structure on 2-gauge transformations such that $\mathfrak{C}(\mathbb{G}^{\Gamma})$ consist of {\it covariant} elements under the 2-gauge transformation $\mathbb{G}^{\Gamma^1}$-representation. Further, the Hopf structures on the two sides shall be compatible, in the sense that $\Lambda$ defines a Hopf module structure.

The bounded linear operators $U$ making the 2-gauge transformations $\Lambda$ concrete strictly speaking now act on spaces of formal power series of sections over $X$. In this way, the groupoid of 2-gauge parameters $\mathbb{G}^{\Gamma^1}$ themselves acquire a dependence on the formal parameter $\hbar$, as hence are themselves operator-valued formal power series. However, as most of what we will prove in the following is algebraic, this will not play a major role, hence we shall keep the dependence on $\hbar$ and $q$ implicit.

\begin{rmk}
    In the following, we shall consider an \textit{additive} extension of the (geometric/classical) 2-gauge transformations $\tilde{\cC}$. This is defined as an {additive} $\mathbb{G}^{\Gamma^1}$-graded Hopf category internal to $\mathsf{Meas}$; see \cite{SOZER2023109155} for a definition of a monoidal category graded by a 2-group. The homogeneous elements $\zeta\in\tilde{\cC}$ are determined by an element $\{(a_v,\gamma_e)\}_{(v,e)}\in\mathbb{G}^{\Gamma^1}$, hence we will often make arguements directly with $\mathbb{G}^{\Gamma^1}$.  \textbf{Lemma \ref{2-groupgraded}} later will treat this in more detail.
\end{rmk}

\subsection{Coproducts on the 2-gauge transformations}\label{covariance}
Recall the $\mathbb{G}^{\Gamma^1}$-module structure of the 2-graph operators is defined as a map $\Lambda$ from the decorated 1-graphs $\mathbb{G}^{\Gamma^1}$ into automorphisms of the 2-graph operators $\mathfrak{C}(\mathbb{G}^{\Gamma})$, which is realized concretely on each 2-graph operator as sheaves of bounded linear operators $\Gamma_c(H^X)\rightarrow \Gamma_c((\Lambda H)^X)$. We have previously noted that there are sheaf isomorphisms witnessing the compositions of 2-gauge transformations 
\begin{equation*}
    \Lambda_{(a_v,\gamma_{e})} \cdot \Lambda_{(a_{v}',\gamma'_{e})} \cong \Lambda_{(a_va_v',\gamma_{e}(a_v\rhd \gamma_{e}'))}
\end{equation*}
horizontally, and also vertically
\begin{equation*}
    \Lambda_{(a_v,\gamma_{e_1})}\circ \Lambda_{(a_{v'},\gamma_{e_2})} = \Lambda_{(a_v,\gamma_{e_1}\gamma_{e_2})},\qquad a_{v'}=a_vt(\gamma_e)
\end{equation*}
on adjacent 1-graphs $v\xrightarrow{e_1}v'\xrightarrow{e_2}v''$.

Now the point is that $\Lambda$ should endow $\mathfrak{C}(\mathbb{G}^{\Gamma})$ with the structure of a Hopf module over the 2-gauge transformations $\mathbb{G}^{\Gamma^1}$ (or its additive completion $\tilde{\cC}$). In order for this to be the case, a \textit{coproduct}
$$\tilde \Delta: \tilde{\cC}\rightarrow \tilde{\cC}\times\tilde{\cC}$$
 akin to \S \ref{2graphcoprod}, must be specified.
\begin{definition}\label{derivationpropertydefinition}
    We say the action functor $\Lambda:\mathbb{G}^{\Gamma^1}\times\mathfrak{C}(\mathbb{G}^{\Gamma})\rightarrow\mathfrak{C}(\mathbb{G}^{\Gamma})$ has the \textbf{categorical quantum derivation property} iff there exist sheaf identifications such that
\begin{equation}
    \Lambda_-\circ (-\ostar-) \cong (-\ostar-)\circ (\Lambda\otimes\Lambda)_{\tilde\Delta}\,.\label{covar}
\end{equation}
\end{definition}
\noindent The reason why this condition is named such is given in \textit{Remark \ref{derivation}}. We shall prove in {\bf Proposition \ref{moduleassoc}} later that \eqref{covar} is indeed necessary for the monoidal module structure.



\medskip

In the following, we shall instead focus on a more geometric interpretation of the coproduct $\tilde \Delta$ on $\tilde{\cC}$.

\begin{rmk}\label{derivation}
    Let us return for the moment to the case of the ordinary compact semisimple Lie group $G$. Recall the coproduct on $U\g$ is primitive $\tilde\Delta(X) = X\otimes 1 + 1\otimes X$. The condition analogous to \eqref{covar} in the decategorified classical case then reproduces the \textit{Leibniz rule}
    \begin{equation*}
        \Lambda_X(fg) = (\Lambda_Xf)g + f(\Lambda_Xg),\qquad X\in U\mathfrak{g},\qquad f,g\in C(G),
    \end{equation*}
    which identifies $\Lambda: U\g\otimes C(G)\rightarrow C(G)$ as the canonical action of $U\g$ by derivations on the functions $C(G)$ \cite{Semenov1992}. This explains why \eqref{covar} was called the "{quantum derivation}" property: the action of $\mathbb{U}_q\G$ on $\mathfrak{C}_q(\mathbb{G})$ behaves like a "derivation", and give further support for the interpretation that $\mathbb{U}_q\G$ is the categorical version of the quantum enveloping algebra. See also \S \ref{lattice2alg} later.
\end{rmk}

\subsection{Geometry of the coproduct $\tilde\Delta$}\label{2gtgeometry}
As we have mentioned, the condition \eqref{covar} is necessary for the Hopf category $\mathfrak{C}(\mathbb{G}^{\Gamma})$ to be a Hopf module (as a measureable category) under 2-gauge transformations $\Lambda:\mathbb{G}^{\Gamma^1}\times\mathfrak{C}(\mathbb{G}^{\Gamma})\rightarrow\mathfrak{C}(\mathbb{G}^{\Gamma})$. Geometrically, this condition also has an interpretation in terms of the composition of faces and edges in the graph complex $\Gamma$, similar to what was described in \S \ref{2graphcoprod}.

Let us for now focus on the classical case to make the geometry more explicit. If we write, in Sweedler notation,
\begin{equation*}
    \tilde\Delta_\zeta= \sum \zeta_{(1)}\times\zeta_{(2)},\qquad \Lambda_{\tilde\Delta_\zeta} = \sum \Lambda_{\zeta_{(1)}}\times \Lambda_{\zeta_{(2)}},
\end{equation*}
then the point of \eqref{covar} (as well as the introduction of $\tilde\Delta$) is to ensure that the geometric picture in fig. \ref{fig:2gaugesplitting} is consistent.
\begin{figure}
    \centering
    \includegraphics[width=0.7\linewidth]{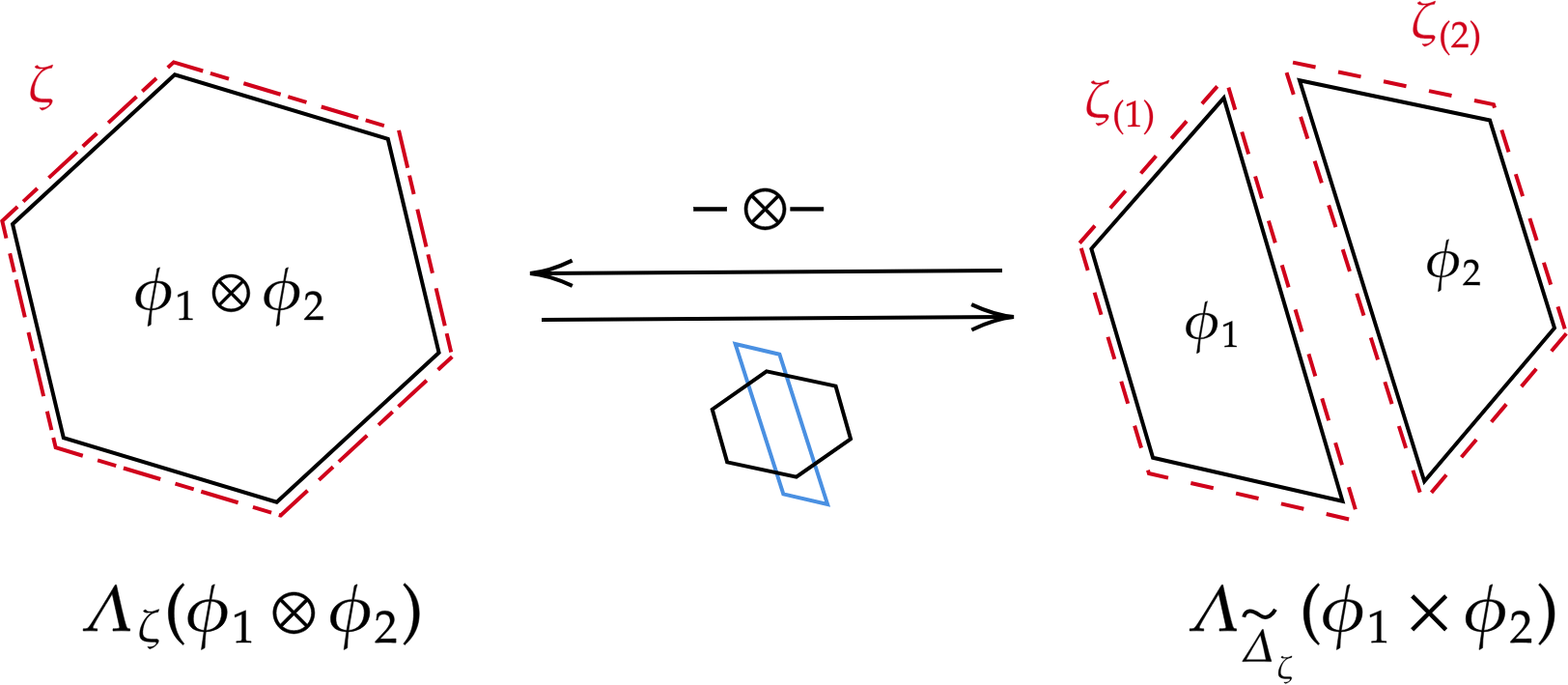}
    \caption{A coproduct $\tilde\Delta$ on the 2-gauge transformation parameters $\zeta\in\tilde{\cC}$ manifests naturally from the composition of the 2-graphs.}
    \label{fig:2gaugesplitting}
\end{figure}
Let us make this more explicit in the following.

\subsubsection{Graph and 2-gauge transformations}\label{2gausplitting}
Consider a face $(e,f)\in\Gamma^2$ which is obtained from the horizontal composition of two half-faces $(e_1,f_1),(e_2,f_2)$. The 2-graph operators $\phi_{(e_1,f_1)},\phi'_{(e_2,f_2)}\in\C(\mathbb{G}^{\Gamma})$  localized on such a face must then agree on the attaching edge $(e_1,f_1)\cap (e_2,f_2)$. 

We wish to examine the 2-gauge transformation properties of such a configuration of classical, geometric 2-graph operators. For this, it would be useful in the following to introduce the following notation which detects the proximity of a 2-gauge transformation from the decorated 2-graph localized at the face $(e,f)$. 

For each face $(e,f)\in\Gamma^2$ with root edge $e$, we define a 2-gauge action $\Lambda^{(e,f)}$ as follows,
\begin{equation*}
    \Lambda^{(e',f')}_{(a_v,\gamma_{e})} = \begin{cases}
        \Lambda_{(a_v,\gamma_{e'})} &; e'=e \\
        \Lambda_{(a_{v},\gamma_{e'})}^{-1} &; e' = e\ast \partial f \\
        \Lambda_{(a_v,1)} &; v = t(e') = s(e) \\ 
        \Lambda_{(a_v,{\bf 1}_{a_v})}^{-1} &; v=s(e') = t(e) \\
        \Lambda_{(1_v,({\bf 1}_{1_v})_e)} &; \text{otherwise}
    \end{cases},\qquad (a_v,\gamma_e)\in\mathbb{G}^{\Gamma^1}
\end{equation*}
where $\bar e$ denotes the orientation reversal of the edge $e$, whose source $s(\bar e)$ is the target $t(e)$ of $e$. This notation. Thus, if $(e,f)$ is composite, then we are able to deduce its action by 2-gauge transformations locally by looking at how its 1-graph boundary $\Gamma^1$ is composite.

Denote by 
\begin{equation*}
    \phi\mapsto \Lambda_{(a_{v_1},\gamma_{e_1})}^{(e,f)}\phi
\end{equation*}
the \textit{localized} 2-gauge transformation action on a 2-graph operator at $(e,f).$ Given the composite $(v,e) = (v_1,e_1)\cup (v_2,e_2)$ of the boundary edge, we have
\begin{equation*}
    \Lambda_{(a_{v_1},\gamma_{e_1})}^{(e_1,f_1)}\phi,\qquad \Lambda^{(e_2,f_2)}_{(a_{v_2},\gamma_{e_2})}\phi,
\end{equation*}
where $v_{1,2}$ denotes the source vertex of the edge $e_{1,2}$. Recall that the left-action $\Lambda$ is defined by precomposing with a horizontal conjugation action \eqref{2gt}, hence we see at the level of the decorated 2-graphs that
\begin{equation*}
    \operatorname{hAd}^{-1}_{(a_{v_1},\gamma_{e_1})}(h_{e_1},b_{f_1})\cdot_h\operatorname{hAd}^{-1}_{(a_{v_2},\gamma_{e_2})}(h_{e_2},b_{f_2}) = \operatorname{hAd}^{-1}_{(a_{v_1},\gamma_{e_1}\gamma_{e_2})}(h_e,b_f),    
\end{equation*}
where we have noted that $v_2$ must also be the target vertex of $e_1$ from the geometry, and by definition $(h_{e_1},b_{f_1})\cdot_h(h_{e_1},b_{f_1}) = (h_e,b_f)$. This dictates how the 2-gauge transformations act on tensor products of 2-graph operators: evaluating this identity yields the \textit{covariance} of the graph-cutting coproduct:
\begin{align*}
    \Lambda_{(a_{v_1},\gamma_{e})}^{(e,f)}\phi&\cong \left(\bigoplus_{\substack{v_2 = s(e_2) \\ \gamma_e=\gamma_{e_1}\gamma_{e_2}}}\Lambda_{(a_{v_1},\gamma_{e_1})}^{(e_1,f_1)}(\phi_{(1)}) \otimes \Lambda^{(e_2,f_2)}_{(a_{v_2},\gamma_{e_2})}(\phi_{(2)})\right),
\end{align*}
from which we can deduce the coproduct
\begin{equation}
    \tilde\Delta_{\zeta_{(v,e)}} = 
        \bigoplus_{\zeta_{(v_1,e_1)}\circ \zeta_{(v_2,e_2)}=\zeta_{(v,e)}}\zeta_{(e_1,f_1)}\times \zeta_{(e_2,f_2)}  \label{compat1}
\end{equation}
for 2-gauge parameters localized on composite 1-graphs. 

In this classical undeformed setting, the $R$-matrix is trivial, as such the expressions \eqref{compat1} give precisely the {classical} version of \eqref{covar}. The Sweedler notation for the coproduct is also "grouplike". We shall see in the following section that, in order to have a  quantum version of \eqref{covar}, we must have a non-trivial $R$-matrix on the 2-gauge parameters $\tilde{\cC}$ as well.

\subsubsection{The induced $R$-matrix on $\tilde{\cC}$}
To promote the compatibility conditions \eqref{compat1} to the quantum theory, we are going to assume that there exist elements $\tilde R\in\tilde{\cC}\times\tilde{\cC}$ such that there are sheaf identifications
\begin{align}
    & R \otimes \big(\Lambda_{\tilde\Delta_\zeta}(\phi \times \phi' )\big) \cong \Lambda_{\tilde R\cdot \tilde \Delta_\zeta}(\phi \times\phi') \nonumber\\ 
    & \big(\Lambda_{\tilde\Delta_\zeta}(\phi \times\phi' )\big)\otimes R \cong \Lambda_{\tilde\Delta_\zeta\cdot\tilde R}(\phi\times\phi') \label{inducrmat}
\end{align}
for any $\phi\times\phi'\in (\mathfrak{C}_q(\mathbb{G}^{\Gamma}))^{\times 2}$ and homogeneous elements $\zeta\in\tilde{\cC}$ determined by $\mathbb{G}^{\Gamma^1}$.\footnote{The "$\cdot$" on the right-hand sides denote the product of 2-gauge transformations in $\mathbb{G}^{\Gamma^1}$.} 

The purpose of this condition is that, if $\tilde R$ induces a cobraiding for the quantum 2-gauge parameters $\tilde{\cC}$ --- that is, there are natural transformations 
\begin{equation}
    \tilde{\mathsf{T}}_{(a_v,\gamma_e)}: (\tilde\Delta^\text{op})(a_v,\gamma_e)\cdot \tilde R \cong \tilde R\cdot\tilde\Delta_h(a_v,\gamma_e),\label{2gtR}
\end{equation}
for each $\zeta=(a_v,\gamma_e)\in\mathbb{G}^{\Gamma^1}$, then \eqref{covar} follows. We will now prove this.

\begin{theorem}
    Suppose there exists $\tilde R \in\tilde{\cC}\times\tilde{\cC}$ and the natural transformation \eqref{2gtR} which identifies a \textrm{cobraiding} on $\tilde{\cC}$. If the identifications \eqref{inducrmat} exist, then the 2-gauge transformation functor $\Lambda$ has the categorical quantum derivation property as in \textbf{Definition \ref{derivationpropertydefinition}}.
\end{theorem}
\begin{proof}
    We are going to leverage the expression \eqref{2quantumproduct-heis} for the quantum tensor product $\ostar$. If $(e,f),(e',f')$ are delocalized faces (namely their intersection is empty), then from \eqref{2fockrosly-heis} the Poisson bracket is trivial, hence the product reduces to the classical one $\ostar=\otimes$. Thus suppose from now on that $(e,f),(e',f')$ are not delocalized.
    
    Let $\phi_{(e'',f'')}$ be a localized 2-graph operator for which $(e'',f'') = (e,f)\cup (e',f')$ is composite face. Then beginning from the right-hand side of \eqref{covar}, we have
    \begin{align*}
        (-\ostar-)\big(\Lambda_{\tilde\Delta_\zeta}(\phi_{(e,f)}\otimes\phi_{(e',f')})\big) &= (-\otimes-)[R,\Lambda_{\tilde\Delta_\zeta}\Delta_h\phi_{(e'',f'')}]_c \\
        &\cong (-\otimes-)\big(R \otimes \Lambda_{\tilde\Delta_\zeta}\Delta_h\phi_{(e'',f'')} - (\Lambda_{\tilde\Delta_\zeta}\Delta_h \phi_{(e'',f'')})^\text{op}\otimes \tilde R^T\big),
    \end{align*}
    where the superscript "op" means that the two tensor factors are swapped (note $\Delta^\text{op}$ is the same as $\sigma\circ\Delta$). Consider the first term: using the condition \eqref{inducrmat}, then \eqref{2gtR} and \eqref{inducrmat} again gives us
    \begin{align*}
        R \otimes \Lambda_{\tilde\Delta_\zeta}\Delta_h\phi_{(e'',f'')} &\cong \Lambda_{\tilde R\cdot \tilde\Delta_\zeta}\Delta_h\phi_{(e'',f'')} \cong\Lambda_{\tilde\Delta_\zeta^\text{op}\cdot\tilde R}\Delta_h\phi_{(e'',f'')}\\
        &\cong \big(\Delta^\text{op}_\mathrm{h}\phi_{(e'',f'')}\Lambda^T_{\tilde\Delta_\zeta^\text{op}\cdot\tilde R}\big)^\text{op}=\big((\Delta^\text{op}_\mathrm{h}\phi_{(e'',f'')}\otimes \tilde R^T) \Lambda^T_{\tilde\Delta_\zeta^\text{op}}\big)^\text{op}\\
        &\cong \Lambda_{\tilde\Delta_\zeta}\big((\Delta_h^\text{op}\phi_{(e'',f'')}\otimes R^T) \big)^\text{op},
    \end{align*}
    where "$\otimes$" here denotes the monoidal structure on $\C_q(\mathbb{G}^{\Gamma})^{\times 2}$.

    If we perform a contraction $(-\otimes-)$, then we can make use of \eqref{compat1} to get
    \begin{align*}
        \Lambda_\zeta \big((-\otimes -)(\Delta^\text{op}_\mathrm{h}\phi_{(e'',f'')}\otimes R^T)^\text{op} \cong \Lambda_\zeta (-\otimes-)(R \otimes \Delta_h\phi_{(e'',f'')}),
    \end{align*}
    which is nothing but the first term of $\Lambda_\zeta (\phi_{(e,f)}\ostar\phi_{(e',f')})$. The same argument takes care of the second term, whence we finally achieve \eqref{covar}
    \begin{equation*}
        (-\ostar-)\big(\Lambda_{\tilde\Delta_\zeta}(\phi_{(e,f)}\otimes\phi_{(e',f')})\big) \cong \Lambda_\zeta(\phi_{(e,f)}\ostar\phi_{(e',f')}),
    \end{equation*}
    as desired.
\end{proof}


To summarize, we have essentially promote the results of \S \ref{2gausplitting} to the quantum case by inducing a $R$-matrix on $\tilde{\cC}$. We will see what \textbf{Definition \ref{derivation}} means categorically in \S \ref{lattice2alg}.

\subsection{Hopf structure on the quantum 2-gauge transformations}\label{2gthopf}
Let us now introduce an analogous set of compatibility conditions for $\tilde R$ on $\tilde{\cC}$, in analogy with \S \ref{catRmat}. Putting $\tilde{\cC}_0$ as the component additively generated by the vertex gauge transformation parameters $G^{\Gamma^0}$, we require the restriction $\tilde R|_{\tilde \cC_0\times\tilde \cC_0}=\tilde R_0\in \tilde{\cC}_0\times \tilde{\cC}_0$ to satisfy (cf. \eqref{catRnaturality})
\begin{align}
    (\bar s\times \bar s) \circ (\tilde R\cdot-) = (\tilde R_0\cdot -)\circ (\bar s\times \bar s)\qquad (\bar t\times \bar t) \circ (\tilde R\cdot-) = (\tilde R_0\cdot -)\circ (\bar t\times \bar t).\label{2gtRnaturality}
\end{align}
Next (cf. \eqref{hexagon} and \ref{hexagon1}), we require the comonoidality 
\begin{equation}
    (\tilde \Delta\times 1)\tilde R \cong \tilde R^{13}\cdot \tilde R^{12},\qquad 
    (1\times\tilde\Delta)\tilde R \cong \tilde R^{13}\cdot \tilde R^{23},\label{2gtquasitri}
\end{equation}
as usual, but also the condition that, given the vertical "antipode" involution $\tilde S_v$ on the 2-gauge transformations given in \textbf{Definition \ref{unitary2hol}}, the categorical $R$-matrix $(\tilde S_v\times 1)\tilde R=\tilde R^T=(1\times \tilde S_v)\tilde R$ satisfies the naturality condition
\begin{equation}
    (\tilde R^T\cdot-)\circ  -^\text{op} \cong -^\text{op}\circ (\tilde R\cdot -),\label{2gtnat}
\end{equation}
where $-^\text{op}:\tilde{\cC}\rightarrow \tilde{\cC}^\text{op}$ takes the opposite category.

\begin{rmk}
    This condition \eqref{2gtnat} is akin to the fact that, if $R$ is a quasitriangular $R$-matrix for an ordinary Hopf algebra $A$, then $R^{-1}$ is a quasitriangular $R$-matrix for the opposite algebra $A^\text{opp}$. However, here $-^\text{op}$ swaps the composition of the \textit{morphisms}, and does not swap the monoidal product. We shall denote $\tilde{\cC}$ with the reverse monoidal structure by $\tilde{\cC}^\text{m-op}$.
\end{rmk}

\subsubsection{Antipode on the 2-gauge transformations}
The antipodes $\tilde S_h,\tilde S_v$  are, similar to the unitarity of the 2-holonomies \textbf{Definition \ref{unitary2hol}}, geometric in nature. From the computations in \S \ref{2gtgeometry} with the 1-graph coproducts, we impose the following conditions
 \begin{align}
    & (-\cdot-) \circ(\tilde S_h \times 1) \circ\tilde\Delta_h \cong (-\cdot-)\circ (1\times\tilde S_h)\tilde\Delta_h \cong\tilde \eta \cdot \tilde{\epsilon},\nonumber\\
     & (\tilde S_v\gamma\circ \gamma) \cong \id_{\bar s\gamma},\qquad (\gamma\circ \tilde S_v\gamma)\cong \id_{\bar t\gamma},\qquad \forall~\gamma\in\tilde{\cC}_1\label{antpod}
\end{align}
where $\tilde \eta,\tilde\epsilon$ are the co/monoidal co/unit on $\tilde{\cC}$, and $\id_a$ is the compositional unit on $a\in \tilde{\cC}_0$. Here, $\tilde S_v$ acts as the identity on objects of $\tilde{\cC}$

Moreover, the horizontal antipodes $S_h,\tilde S_h$ should be compatible with respect to the $\tilde{\cC}$-module structure: there are sheaf isomorphisms
\begin{equation*}
    S_h(\Lambda_\zeta \phi)^* \cong \Lambda_{\tilde S_h\zeta} (S_h\phi^*)
\end{equation*}
which satisfy the obvious coherence conditions against the module coherence data.

Notice here that $\tilde S_v$ acts essentially like a dagger involution on $\tilde{\cC}$; it is an equivalence by construction. On the other hand, we will assume that the "true" antipode $\tilde S_h=\tilde S$ is also an equivalence,\footnote{This is natural, as in the case of an ordinary Hopf algebra, the bijectivenss of the antipode can be deduced from the relation $(S\otimes S)R=R$ as well as the quasitriangularity of the $R$-matrix \cite{Majid:1996kd}.} but \textit{not} necessarily involutive. From the 2-dagger geometry, we have that these antipodes strongly commute $$\tilde S_h^{\text{op}}\circ \tilde S_v = \tilde S_v^\text{m-op,c-op}\circ \tilde S_h.$$ This property will come back to us in \S \ref{*op}.


\subsubsection{2-gauge transformations as a Hopf category}\label{2gthopfalgbd}
Recall the notion of a Hopf internal category in \textbf{Definition \ref{internalhopf}}. We wish to prove the following.
\begin{lemma}\label{2-groupgraded}
    Consider a category $\tilde{\cC}$ graded by the 2-gauge parameters $\mathbb{G}^{\Gamma^1}$. Then $\tilde C$ is additive bimonoidal.
\end{lemma}
\begin{proof}
By a monoidal category $C$ graded by a monoid $A$, we mean a set/class of \textit{homogeneous objects} $C^\text{hom}$ and a map $|-|: C^\text{hom}\rightarrow A$, called a grading, which preserves the source/target maps; see also \cite{SOZER2023109155}.

    Define $\tilde{\cC}^\text{hom}=\mathbb{G}^{\Gamma^1}$, then we make $\tilde{\cC}$ as the additive completion of $\tilde{\cC}^\text{hom}$ by setting $\Lambda_{\zeta\oplus\zeta'} = \Lambda_\zeta\oplus \Lambda_{\zeta'}$. The source and target maps $\bar s,\bar t$ on $\tilde{\cC}$ are defined on homogeneous elements such that the conditions
    \begin{equation*}
        \hat s^*(\Lambda_\zeta\phi) = \Lambda_{\bar s\zeta}\hat s^*\phi,\qquad \hat t^*(\Lambda_\zeta\phi) = \Lambda_{\bar t\zeta}\hat t^*\phi
    \end{equation*}
    are satisfied, where we recall $\hat s^*,\hat t^*$ are the cosouce/cotarget maps on the 2-graph operators $\mathfrak{C}(\mathbb{G}^{\Gamma})$ (see \textbf{Lemma \ref{cocat}}). 

We specify the monoidal structure on $\tilde{\cC}$ on its homogeneous elements. This comes from the composition of 2-gauge transformations $\Lambda_\zeta\cdot \Lambda_{\zeta'} \cong \Lambda_{\zeta\cdot\zeta'}$ for $\zeta,\zeta'\in\tilde{\cC}=\mathbb{G}^{\Gamma^1}$, as explained in \S \ref{covarrep}. Similarly, the comonoidal functor is given by the horizontal coproduct $\tilde\Delta_h$ on $\tilde{\cC}^\text{hom}$, via $\Lambda_{\tilde\Delta_\zeta} = (\Lambda\times\Lambda)_{\tilde\Delta_\zeta}$.
    
 We now need to verify the bimonoidal axioms \eqref{bimonoid}. For each 2-graph operator $\phi\in\cC=\mathfrak{C}(\mathbb{G}^{\Gamma})$, we have from \eqref{covar} that
    \begin{align*}
        \Delta(\Lambda_\zeta (\Lambda_{\zeta'}\phi)) &\cong \Lambda_{\tilde\Delta_\zeta}(\Delta(\Lambda_{\zeta'}\phi)) \cong \Lambda_{\tilde\Delta_\zeta}\cdot \Lambda_{\tilde\Delta_{\zeta'}}(\Delta\phi) \\
        &\cong \big((\Lambda_{\zeta_{(1)}}\cdot\Lambda_{\zeta'_1})\times(\Lambda_{\zeta_{(2)}}\cdot\Lambda_{\zeta_{(2)}'})\big)\Delta(\phi),
    \end{align*}
    whereas
    \begin{align*}
        \Delta(\Lambda_{\zeta\cdot\zeta'} \phi)&\cong \Lambda_{\tilde\Delta_{\zeta\cdot\zeta'}}(\Delta\phi)\cong (\Lambda_{(\zeta\cdot\zeta')_1} \times \Lambda_{(\zeta\cdot\zeta')_2})\Delta(\phi).
    \end{align*}
    These two expressions, which describe spaces of  global measureable sections, are equivalent up to possibly a projective phase $c$ mentioned in \textit{Remark \ref{projrep}}, hence 
    \begin{equation}
        (1\otimes\sigma\otimes1)(-\cdot- \times -\cdot -)(\tilde\Delta_\zeta\otimes \tilde\Delta_{\zeta'}) \cong\zeta_{(1)}\cdot\zeta_{(1)}' \times \zeta_{(2)}\cdot \zeta_{(2)}' =\tilde\Delta_{\zeta\cdot\zeta'}\label{2gtbimon}
    \end{equation}
    as desired.
    
\end{proof}
    

Now recall the notion \textbf{Definition \ref{concretify}} that 2-gauge transformations $\Lambda_\zeta$ are by definition realized concretely by bounded linear operators $U_\zeta$. Given the Haar measure $\mu$ on $\mathbb{G}$, we can endow a Haar measure $\mu_{\Gamma^1}$ on $\mathbb{G}^{\Gamma^1}$ by
\begin{equation*}
    d\mu(a,\gamma)_{\Gamma^1} = \prod_{v\in\Gamma^0}d\sigma(a_v) \prod_{e\in\Gamma^1}d\nu^{a_v}(\gamma_e),
\end{equation*}
where in the edges $e$ in the second factor has the vertex $v$ as its source. 
\begin{theorem}\label{2gthopfalgd}
    Under the regularity assumption that the assignments $U:\zeta\mapsto U_\zeta$ define \textit{measureable fields} of ($\mu_{\Gamma}$-essentially) bounded linear operators over $(\mathbb{G}^{\Gamma^1},\mu_{\Gamma^1})$, then $\tilde{\cC}$ is a cobraided Hopf category internal to $\mathsf{Meas}_q$.
\end{theorem}
\begin{proof}
    
    By the regularity assumed in the hypothesis, \textbf{Lemma \ref{2-groupgraded}} can be proven within $\mathsf{Meas}_q$. The antipode $\tilde S$ and the cobraiding $(\tilde{\mathcal{R}},\tilde{\mathsf{T}})$ induced by the $R$-matrices $\tilde R$ are already described in the previous sections. 
    
    
\end{proof}

As suggested by its structures, this Hopf category "$\tilde{\cC}$" is supposed to model a categorification of the quantum enveloping algebra of $\G$.
\begin{definition}
    The \textbf{categorical quantum enveloping algebra} $\mathbb{U}_q\G$ is the quantum 2-gauge transformation $\tilde{\cC}$, which is an additive Hopf category internal to $\mathsf{Meas}_q$ monoidally graded (cf. \cite{SOZER2023109155})by the 2-group $\mathbb{G}^{\Gamma^1}$ on the 1-graph $\Gamma^1$  consisting of a single edge loop $v\xrightarrow{e}v$, and a marked vertex $v\in\Gamma^0$ as its 0-skeleton. 
\end{definition}
Clearly, the categorical quantum coordinate ring $\mathfrak{C}_q(\mathbb{G})$ is a monoidal module cocategory over $\mathbb{U}_q\G$ through $\Lambda$.

\section{Lattice 2-algebra of 2-Chern-Simons theory}
Given what we have found in the previous sections, we now demonstrate how the Hopf categorical structure of 2-gauge parameters can be "combined" with that of the 2-graph operators. This is done by describing the way in which $\mathfrak{C}(\mathbb{G}^{\Gamma})$ can be seen a a "regular" representation under $\tilde{\cC}$.

Then, we will define a Hopf category $\mathscr{B}^\Gamma$ encapsulating all of the degrees-of-freedom and gauge symmetries of 2-Chern-Simons theory on the lattice. This is accomplished through a categorical semidirect product construction, which will by definition embed the 2-graph operators $\mathfrak{C}(\mathbb{G}^{\Gamma})$ into $\mathscr{B}^\Gamma$ as a subcategory. We call this Hopf category $\mathscr{B}^\Gamma$ the \textbf{lattice 2-algebra} of 2-Chern-Simons theory, which can be understood as a categorified 4d analogue of the lattice algebra constructed in \cite{Alekseev:1994pa}.

\subsection{Bimodule structure of the 2-graph operators}\label{regrep}
Let $\mathbb{G}$ be a compact matrix Lie 2-group, in the sense that $G,\mathsf{H}$ are both compact matrix groups. We now put the Hopf categories $\mathfrak{C}(\mathbb{G}^{\Gamma})$ of 2-graph operators and their 2-gauge transformations $\tilde{\cC}$ on them together. To do so, we first need a notion of "regularity" for the $\tilde{\cC}$-module structure of the 2-graph operators.

To begin, we first recall the bounded linear operator $U_{(a_v,\gamma_e)}:\Gamma_c(H^X)\rightarrow \Gamma_c((\Lambda_{(a_v,\gamma_e)}H)^X)$ making $\Lambda$ concrete on the 2-graph operators $\phi$. By \textbf{Proposition \ref{pullbackmeas}}, $\Lambda$ is naturally isomorphic to the pullback of the \textit{conjugation} action of $\mathbb{G}^{\Gamma^1}$ on $\mathbb{G}^{\Gamma}$. By looking at \textit{translation} actions insread, we can induce a $\tilde{\cC}$-bimodule structure on $\mathfrak{C}(\mathbb{G}^{\Gamma})$. More precisely, for each $\zeta=(a_v,\gamma_e)\in\tilde{\cC}$, the left/right 2-group multiplication $\zeta\cdot-,~-\cdot \zeta$ on $\mathbb{G}^{\Gamma}$ induce measureable automorphisms 
\begin{equation*}
    -\bullet\zeta\cong (\zeta\cdot-)^{-1},~\zeta\bullet-\cong (-\cdot \zeta)^{-1}: \mathfrak{C}(\mathbb{G}^{\Gamma})\rightarrow\mathfrak{C}(\mathbb{G}^{\Gamma})
\end{equation*}
on the 2-graph operators, for which there exists a natural measureable transformation
\begin{equation*}
    (\zeta\bullet-) \circ (-\bullet \zeta)\cong (\zeta\bullet-)\circ (-\bullet \zeta).
\end{equation*}
Clearly, these module structures come with their natural module associators. 

We shall in the following collectively denote by this bimodule action by $\zeta\bullet\phi$ and $\phi\bullet \zeta\in\mathfrak{C}(\mathbb{G}^{\Gamma})$ for $\zeta=(a_v,\gamma_e)$. Now the point is that we should be able to recover the 2-gauge transformations $\Lambda$ from these bimodule actions. However, due to the shift gauge symmetry $h_e\mapsto \mathsf{t}(\gamma_e)h_e$ present in the edge holonomies $\{h_e\}_{e\in\Gamma^1}$, we must divide $\gamma_e$ into two decorated half-edges.

\medskip

This is done using the coproduct $\tilde\Delta$ on $\tilde{\cC}$. This motivates the following notion analogous to that given in \cite{Alekseev:1994pa}.
\begin{definition}
    A 2-graph operator $\phi\in \C_q(\mathbb{G}^{\Gamma})$ is called \textbf{left-covariant} iff for each $\zeta$, there is an isomorphism of sheaves such that we have the following sheaf isomorphisms
    \begin{equation}
        \phi\bullet (a_v,\gamma_e)\cong (1 \otimes U)_{\tilde\Delta(a_v,\gamma_e)}\bullet\phi,\qquad \forall (a_v,\gamma_e)\in\tilde{\cC}\label{leftreg}
    \end{equation}
    of spaces of continuous measureable sections over $X=\mathbb{G}^{\Gamma}$, where "$\cong$" denotes the appearance of module associators, which we shall suppress in the following.
    \end{definition}
\noindent Let us obtain the desired natural isomorphism $$\zeta^{-1}-\bullet -\zeta\xRightarrow{\sim} \Lambda_\zeta$$ from \eqref{leftreg} in the classical case. Here, we can write more explicitly in Sweedler notation,
    \begin{equation*}
     \phi \bullet {(a_v,\gamma_e)} \cong \bigoplus_{\substack{a^1_{v}a^2_{v}=a_v \\ \gamma^1_{e}\gamma^2_{e} = \gamma_e}} (((a^1_v,\gamma^1_{e}) \bullet- ) (U_{(a_{v}^2,\gamma^2_{e})}))\bullet \phi \cong \bigoplus_{\substack{a^1_{v}a^2_v=a_v \\ \gamma^1_{e}\gamma^2_{e} = \gamma_e}} (a^1_{v},\gamma^2_{e}) \bullet (U_{(a_{v}^2,\gamma^2_{e})}\phi).
\end{equation*}
By composing on the left by the antipode $\tilde S_h=\tilde S$ of $(a_{v}^1,\gamma^1_{e})$, the antipode axiom \eqref{antpod} leads to
\begin{equation*}
    \bigoplus_{\substack{(a_v^1,\gamma_{e_1})^{-1}\cdot(a_v,\gamma_e)=(a_v^2,\gamma_{e_2})}} (\tilde S(a^1_{v},\gamma^1_{e}))\bullet \phi\bullet (a_v,\gamma_e) \cong U_{(a_{v_2},\gamma_{e_2})}\phi,
\end{equation*}
which reads as the desired covariance condition,
\begin{equation*}
    U_{(a_v,\gamma_e)} \phi \cong \bigoplus_{\substack{(a_v^1,\gamma_{e_1})^{-1}\cdot (a_v^2,\gamma_{e_2})=(a_v,\gamma_{e})}} (\tilde S_h(a_{v}^1,\gamma_{e}^1))\bullet \phi \bullet (a_{v}^2,\gamma_{e}^2),
\end{equation*}
upon a quick change of variables, where we have neglected the module associator $\alpha^\bullet$.

\begin{proposition}\label{rightreg}
    Every left-covariant 2-graph operator $\phi$ is also right-covariant,
    \begin{equation*}
        (a_v,\gamma_e)\bullet\phi \cong \phi\bullet(1\otimes \bar U)_{\tilde\Delta(a_v,\gamma_e)},\qquad \forall~(a_v,\gamma_e)\in\tilde{\cC},
    \end{equation*}
    where $\bar U_{(a_v,\gamma_e)} = U_{\tilde S^{-1}(a_v,\gamma_e)}^\dagger$ is the dual contragredient representation.
\end{proposition}
\begin{proof}
    Let us use the shorthand $\zeta$ for transformations by elements in $\tilde{\cC}$. From the left-covariance of $\phi$ at the element $\zeta_{(1)}$, we have
    \begin{equation*}
        \phi \bullet \zeta_{(1)}\cong \bigoplus \zeta_{(1)(1)} \bullet (U_{\zeta_{(1)(2)}}\phi),
    \end{equation*}
    where we have used the Sweedler summation notation over the coproduct components of $\tilde\Delta_{\zeta_{(1)}}$. By applying from the left the operator $U_{\tilde S^{-1}\zeta_{(2)}}$ and summing, we find 
    \begin{align*}
        \bigoplus U_{\tilde S^{-1}\zeta_{(2)}}(\phi\bullet\zeta_{(1)})& \cong \bigoplus \bigoplus U_{\tilde S^{-1}\zeta_{(2)}}\left(\zeta_{(1)(1)} \bullet (U_{\zeta_{(1)(2)}}\phi)\right) \\
        &\cong \bigoplus \bigoplus\zeta_{(1)} \bullet (U_{\zeta_{(2)(1)}\tilde S^{-1}\zeta_{(2)(2)}}\phi) \\ 
        &\cong \bigoplus\zeta_{(1)} \bullet (U_{\epsilon(\zeta_{(2)})}\phi) \cong \zeta\bullet \phi,
    \end{align*}
    where we have used the coassociator of $\tilde\Delta$ and the antipode axiom \eqref{antpod}. Now going back to the top at the left-hand side, taking a conjugation in the fibre spaces that $U$ acts on, we have
    \begin{align*}
        \bigoplus U_{\tilde S^{-1}\zeta_{(2)}}(\phi\bullet\zeta_{(1)}) &\cong \bigoplus (\phi\bullet\zeta_{(1)})U_{\tilde S^{-1}\zeta_{(2)}}^\dagger \\
        & \cong \bigoplus \phi ((\bullet\zeta_{(1)})(\bar U_{\zeta_{(2)}}))= \phi\bullet (1\otimes \bar U)_{\tilde\Delta(\zeta)},
    \end{align*}
    where "$\cong$" in the second line denotes the appearance of a projective phase $c$, coming from module associativity. This proves the statement.
\end{proof}
A simple computation analogous for the left-covariance condition brings right-covariance to the form
\begin{equation*}
    \phi U_{\tilde S^{-1}\zeta} \cong \bigoplus_{\gamma_{\bar e_1}\gamma_{e_1}=\gamma_{\bar e}}(S\zeta_{(1)})^\dagger \bullet\phi\bullet \zeta_{(2)}.
\end{equation*}
which if we replace $\zeta \mapsto \tilde S\zeta$ we achieve
\begin{equation*}
    \phi U_\zeta \cong \bigoplus_{\gamma_{e_1}\gamma_{\bar e_2} = \gamma_e} (\tilde S^2\zeta_{(1)})^\dagger \bullet\phi \bullet(\tilde S\zeta_{(2)}).
\end{equation*}
Note this is {\it not} the same as left-covariance under the adjoint of $\Lambda$ (ie. a right 2-gauge action), since $\tilde S_h^2=\tilde S^2$ is in general not naturally isomorphic to the identity.


\subsubsection{Categorical definition of the lattice 2-algebra}\label{lattice2alg}
We are finally ready to define the lattice 2-algebra for 2-Chern-Simons theory. At this point, we are going to construct it categorically from a {\it semidirect product} operation on monoidal categories. We shall revisit the lattice 2-algebra later and describe it more concretely once we have understood the representation theory of $\tilde{\cC}$.

Let us first recall the definition of a semidirect product of monoidal categories, following \cite{Fuller:2015}.
\begin{definition}
    Let $\cC,\cD$ denote two monoidal categories with $\cD$ equipped with a (strong)\footnote{Here "strong" means the module associators and unitors are invertible.} $\cC$-module structure $\lhd:\cD\times\cC\rightarrow\cD$. The \textbf{semidirect product} $\cD\rtimes \cC$ is a \textit{monoidal} category consisting of pairs $(D,C)\in \cD\times\cC$ equipped with the monoidal structure
    \begin{equation*}
        (D,C)\otimes(D',C') = (D\otimes (D'\lhd C),C\otimes C').
    \end{equation*}
\end{definition}
\noindent Let $\varrho: (-\lhd -)\lhd - \Rightarrow -\lhd (-\otimes-)$ denote the module associator, and let $\varpi: (-\otimes -)\lhd- \Rightarrow (-\lhd -)\otimes (-\lhd -)$ be the module tensorator. The associator morphism on $\cD\rtimes\cC$ is given by $(\bar\alpha_\cD,\alpha_\cC)$, where $\alpha_{\cC,\cD}$ are the associators on $\cC,\cD$ respetively and
\begin{equation*}
    \bar\alpha_\cD= \varpi_{C_1}(D_2,D_3\lhd C_2)\circ {\bf 1}_{D_2\lhd C_1}\otimes \varrho_{C_1,C_2}(D_3) \circ \alpha_\cD(D_1,D_2\lhd C_1,D_3\lhd (C_1\otimes C_2)).
\end{equation*}
The fact that $\cD\rtimes\cC$ forms a monoidal category is proven in \cite{Fuller:2015}.


Taking $\cD = \mathfrak{C}(\mathbb{G}^{\Gamma})$ to be the Hopf category of the 2-graph operators and $\cC = \tilde{\cC}$ to be the 2-gauge transformations. We take the action functor $\lhd=\bullet$ to be the right-regular representation, and form the semidirect product $\cD\rtimes \cC$ equipped with the tensor product
\begin{equation*}
    (\phi,\zeta)\ostar (\phi',\zeta') \cong (\phi\ostar (\phi'\bullet\zeta),\zeta\cdot\zeta'),
\end{equation*}
where $\circ$ is the left-composition of 2-gauge transformations in $\tilde{\cC}$. By strict monoidality of $\mathbb{G}$, the module associator $\varrho_{\zeta,\zeta'}(\phi): (\phi\bullet\zeta)\bullet\zeta' \rightarrow \phi\bullet(\zeta\circ\zeta')$ is an invertible measureable morphism for each $\phi\in\mathfrak{C}(\mathbb{G}^{\Gamma})$. We shall denote by $\varrho_{\zeta,\zeta'}^* =\varrho_{\tilde S\zeta,\tilde S\zeta'}$ for the module associator under the antipode funcor $\tilde S$ on $\tilde{\cC}$.


\medskip

We now need the module tensorator $\varpi$. However, in our case, $\varpi$ takes a different form --- it must involve the coproduct $\tilde\Delta$.
\begin{proposition}\label{moduleassoc}
    Suppose $\mathfrak{C}(\mathbb{G}^{\Gamma})$ is generated by left-covariant 2-graph operators, then the derivation property \eqref{covar} provides precisely the  module tensorator $\varpi: (-\ostar-)\bullet-\Rightarrow (-\ostar-)\circ \big((-\times-)\bullet\tilde\Delta_-\big)\,.$ 
\end{proposition}
\begin{proof}
    Recall that for left-covariant 2-graph operators $\mathfrak{C}(\mathbb{G}^{\Gamma})$, the 2-gauge transformation $\Lambda$ is written equivalently in terms of the bimodule structure $\bullet$ via \eqref{leftreg}. The left-/right-module actions $\bullet$ strongly commute, hence if $\Lambda$ has the categorical quantum derivation property \eqref{covar}, so does $\bullet$.
    
    Now given $\bullet$ satisfies \eqref{covar} with $\bullet$ in place of $\Lambda$,
    $$    (-\ostar-) \bullet -\cong (-\ostar-)\circ (-\times-)\bullet \tilde\Delta,$$ then the sheaf identification underlying it
\begin{equation}
        \varpi_{\phi,\phi'}(\zeta): (\phi\ostar\phi')\bullet\zeta \cong (-\ostar-)\big((\phi\times \phi')\bullet\tilde\Delta_\zeta\big)\label{tensorator}
\end{equation}
are precisely the components of the module tensorator, where $\phi,\phi'\in\mathfrak{C}(\mathbb{G}^{\Gamma})$ and $\zeta\in\tilde{\cC}$. Naturality is clear. 
\end{proof}
\noindent Strictly speaking, it should be this bimodule structure $\bullet$ that appears in the categorical quantum derivation property.

The lattice 2-algebra (cf. \cite{Alekseev:1994pa}) is thus the semidirect product category satisfying additional conditions.
\begin{definition}
    Let $\Gamma$ denote a 2-graph embedded in a 3d Cauchy slice $\Sigma\subset M^4$. The \textbf{lattice 2-algebra} $\mathscr{B}^\Gamma$ of 2-Chern-Simons theory on $\Gamma$ is the Hopf monoidal semidirect product $(\mathfrak{C}(\mathbb{G}^{\Gamma})\rtimes\tilde{\cC},\bullet,\varrho,\varpi)$ as defined above, such that the following holds.
    \begin{enumerate}
        \item Each $\phi\in\mathfrak{C}(\mathbb{G}^{\Gamma})$ is left-covariant \eqref{leftreg} (and hence also right-covariant).
        \item As $\tilde{\cC}$-modules, we have the \textbf{braid relation}
        \begin{equation}
            \phi\times\phi' \cong (\Lambda\times \Lambda)_{\tilde R}(\phi'\times\phi)\label{braid}
        \end{equation}
        under the 2-gauge actions $U,U'$, which is natural with respect to measureable morphisms. Here, $\tilde R$ is localized on the 1-graph intersection of the supports of $\phi,\phi'$ on $\Gamma^2$.
    \end{enumerate}
\end{definition}
Let us explain briefly about the braid relation \eqref{braid}. As representations of the Hopf category $\tilde{\cC}$, the functor $\phi\ostar\phi'\rightarrow \phi'\ostar\phi$, given by swapping the tensor factors \textit{and} acting with the categorical $R$-matrix $R$, is intertwining iff the relations \eqref{quantumR} hold. Based on the definition of $\tilde R$ \eqref{inducrmat}, this then assures that both sides of \eqref{braid} furnish the same $\tilde{\cC}$-representation.

\begin{rmk}
    Recall from {\it Remark \ref{quasihopf}} that, in the weakly-associative case, $\mathfrak{C}(\mathbb{G}^{\Gamma})$ acquires witnesses for coassociativity, which can be viewed as an internal natural transformation $\alpha^\Delta:(\Delta\times 1)\Delta \Rightarrow (1\times\Delta)\Delta$ on $\mathfrak{C}(\mathbb{G}^{\Gamma})$. In this case, the aforementioned functor $c_{\phi,\phi'}$ is a $\mathbb{G}^{\Gamma^0}$-intertwiner up to homotopy.
\[\begin{tikzcd}
	{\phi\times\phi’} & {\phi\times\phi’} \\
	{\phi’\times\phi} & {\phi’\times\phi}
	\arrow[ from=1-1, to=1-2]
	\arrow[from=1-1, to=2-1]
	\arrow[from=1-2, to=2-2]
	\arrow[shorten <=4pt, shorten >=4pt, Rightarrow, from=2-1, to=1-2]
	\arrow[from=2-1, to=2-2]
\end{tikzcd}\]
    As such, both sides of the condition \eqref{braid} only determines the same $\mathbb{G}^{\Gamma^0}$-module element up to homotopy presented by a sheaf morphism. 
\end{rmk}

The braid relations \eqref{braid} can be interpreted as a an identification of sheaves over $(\mathbb{G}^{\Gamma},\mu_{\Gamma})$, which are only required to be defined \textit{$\mu_{\Gamma}$-a.e.} on the measureable global sections.


\subsubsection{Invariant 2-graph operators and the 2-Chern-Simons observables}
With the lattice 2-algebra $\mathscr{B}^\Gamma$ in hand, we can now define the obseravbles in discretized 2-Chern-Simons theory in an algebraic manner.
\begin{definition}
    The \textbf{observable 2-subalgebra} $\mathscr{O}^\Gamma\subset\mathscr{B}^\Gamma$ is the subspace generated by 2-graph operators $\phi$ satisfying the \textit{invariance} condition
    \begin{equation}
        \phi\bullet \zeta \cong \zeta\bullet\phi,\qquad \forall~ \zeta\in A\label{invarstatement}
    \end{equation}
    for all measureable subsets $A\subset\tilde{\cC}$.
\end{definition}
In other words, $\mathscr{O}^\Gamma$ is the space of (a.e.) invariants under the $\tilde{\cC}$-module structure $\Lambda$. We now prove that $\mathscr{O}^{\Gamma}$ inherits the additive cocategory structure from $\mathfrak{C}(\mathbb{G}^{\Gamma})$.

\begin{definition}
    By a \textbf{homotopy fixed point} $\cC^A$ of a cocategory $\cC= \cC_1\overset{\hat s}{\underset{\hat t}{\leftleftarrows}}\cC_0$ under the action of an additive monoidal category/algebraoid $A$, we mean the cocategory such that, for each $a\in A$ and coarrows $f\in\cC_1$, there exist an object $x_a\in\cC_0$ such that $f=\hat s(x_a)\in\cC_1$ is the cosource and $a\rhd f = \hat t(x_a)\in\cC_1$ is the cotarget. The assignment $a\mapsto x_a$ also satisfies the usual monoidality coherence axioms.
\end{definition}

\begin{proposition}\label{homfix}
    Consider the 2-graph operators $\cC=\mathfrak{C}(\mathbb{G}^{\Gamma})$ as a Hopf cocategory internal to $\mathsf{Meas}_q$ as in \S \ref{hopfcat}. The invariant states are precisely the homotopy fixed points $\cC^{\tilde{\cC}}$.
\end{proposition}
\begin{proof}
    Neglecting the module associators for the moment, let us use the covariance condition \eqref{leftreg} and \textbf{Definition \ref{concretify}} to rewrite the invariance condition in the following way
    \begin{equation*}
        U_\zeta \phi  \cong \phi.
    \end{equation*}
    It is convenient to denote by $\cV,\tilde{\cE}$ the (additive completion of the) \textit{vertex/edge parameters}, whose 2-gauge transformations $\Lambda_\zeta$ are localized respectively on the vertices $v$ and edges $e$ of the 1-graph $\Gamma^1$. These then allows us to describe $\tilde{\cC}$ as,
    \begin{equation*}
        \bar s,\bar t: \tilde{\cE}\rightrightarrows \cV,
    \end{equation*}
      in terms of the source/target structure maps $\bar s,\bar t$. Note that, if we set $\mathsf{H}=1$, the 2-gauge action $\Lambda$ recovers the notion of gauge transformations in the usual (lattice) 1-gauge theory. 

    If $\phi$ satisfies \eqref{invarstatement}, then of course $U=\id$ is the identity operator $\mu_{\Gamma}$-a.e. Hence for invariant 2-graph operators, a vertex parameter $a_v\in\cV$ acts $\phi$ as given in \eqref{2gt},
    \begin{equation*}
        (a_v\rhd \phi)(\{b_f\}_f) \cong  \phi(\{a_v^{-1}\rhd(a_{\bar v}\rhd b_f)\}_f),
    \end{equation*}
    where $\bar v$ is the target vertex of the source edge $e$ of the face $f$. The goal is therefore to find a edge state $\psi_v$ such that $\hat s^*\psi_v=\phi$ and $\hat t^*\psi_v =a_v\rhd\phi$. 
    
    If we put a "pure-gauge" $h_e = a_v^{-1}a_{\bar v}$ on the edge $e:v\rightarrow\bar v$, then we have $a_v^{-1}\rhd(a_{\bar v}\rhd b_f) = (a_v^{-1}a_{\bar v})\rhd b_f = h_e\rhd b_f$, which is nothing but a \textit{whiskering} operation \cite{Baez:2004,Baez:2004in,Chen:2024axr}. Now let $\psi_v$ denote an edge state that has support only on such pure gauge configurations $h_e = a_v^{-1}a_{\bar v}$, then the invariance condition \eqref{invarstatement} identifies $\psi_v$ to have cosource $\phi$ and cotarget $a_v\rhd\phi$, as desired. In other words, for each $v\in\cV$ and invariant coarrow $\phi\in\cF$, there exists an object
    \begin{equation*}
        a_v\rhd \phi \leftarrow \psi_v  \rightarrow\phi
    \end{equation*}
    trivializing the $\cV$-action. It is clear from the properties of whiskering that this assignment $v\mapsto \psi_v$ respects the composition of gauge transformations. 

\end{proof}

The idea that "gauging a symmetry" is the same as taking the "equivariantization/homotopy fixed points" of the given theory is well-known throughout recent literature; see eg. \cite{Delcamp:2023kew}. Here we provide a similar characterization of the 2-Chern-Simons observables: they are precisely the \textit{equivariantization} $\cC^{\tilde{\cC}}$ of the 2-graph operators $\mathcal{C}=\mathfrak{C}(\mathbb{G}^{\Gamma})$ with respect to the categorical quantum symmetry $\tilde{\cC}$.

\subsection{*-operation in the lattice 2-algebra $\mathscr{B}^\Gamma$}\label{*op}
In the final section of this paper, we now study a *-operation on $\mathscr{B}^\Gamma$. As inspired by \S 6 of \cite{Alekseev:1994pa}, this *-operation will be induced by the orientation and framing properties \cite{ferrer2024daggerncategories} of $\Gamma^2$. We now work to make this idea more precise.

Naturally, this makes the *-operations tied inherently to the antipodes introduced in \S \ref{2dagger}. As such, we will make use of 2-$\dagger$-unitarity and much of the ideas introduced there.

\subsubsection{Orientation and framing}
For each quantum 2-graph operator $\phi=\Gamma_c(H^X)[[\hbar]]\in\mathfrak{C}_q(\mathbb{G}^{\Gamma})$, we introduce natural $\bbC[[\hbar]]$-linear measureable sheaf morphisms  $\eta_\mathrm{h,v}: \Gamma_c(H^X)[[\hbar]]\rightarrow \Gamma_c(H^{\overline{X}^\mathrm{h,v}})[[\hbar]]$ induced on the 2-graph operators by the 2-$\dagger$ structure of $\Gamma$. 

\begin{definition}\label{daggerpair}
    We say the pair $(\eta_\mathrm{h},\eta_\mathrm{v})$ is a \textbf{2-$\dagger$-intertwining pair} iff for each $\zeta\in\mathbb{U}_q\G^{\Gamma^1}$, we have
    \begin{equation*}
    \eta_\mathrm{h} \circ U_\zeta   = U_{\bar \zeta}\circ  \eta_\mathrm{h} ,\qquad \eta_\mathrm{v} \circ  U_\zeta = U_{\zeta^T} \circ \eta_\mathrm{v}\phi
\end{equation*}
as operators acting any $\phi\in\C_q(\mathbb{G}^{\Gamma})$, where $\bar\zeta$ denotes the 2-gauge parameter assigned to the orientation reversal $\overline{v\xrightarrow{e}v'} = v'\xrightarrow{\bar e}v$, and $\zeta^T$ denotes that assigned to the \textit{frame rotation} $(v,e)^T = (v,e^T)$.
\end{definition}

We are finally ready to state the *-operations on the 2-graph operators and the 2-gauge transformations. Suppose the $R$-matrix $\tilde R$ on $\mathbb{U}_q\G^{\Gamma_1}$ is invertible, in the sense that the induced cobraiding natural transformations $\tilde\Delta\Rightarrow \tilde\Delta^\text{op}$ are invertible.

By locality, it suffice to define the *-operations on local pieces.
\begin{definition}\label{starop}
Let $(v,e) = v\xrightarrow{e}v'\in\Gamma^1$ denote a 1-graph, and  let $(e,f)\in\Gamma^2$ denote a 2-graph, with source and target edges $e,e':v\rightarrow v'$.
    \begin{enumerate}
        \item The \textbf{*-operations} on localized homogeneous elements in $\tilde{\cC}$ are given by
    \begin{equation}
        \zeta^{*_1}_{(v,e)} = \bar \zeta ,\qquad \bar\zeta^{*_2}_{(v,e)} = \zeta^T \label{dagger2-gau}
    \end{equation}
    where $v'\xrightarrow{\bar e}v$ is the orientation-reversal and $v\xrightarrow{e^T}v'$ is the framing rotation.
    \item Given the 2-$\dagger$-intertwining pairs in \textbf{Definition \ref{daggerpair}}, the \textbf{*-operations} on localized 2-graph operators $\phi_{(e,f)}\in\mathfrak{C}(\mathbb{G}^{\Gamma})=\mathfrak{C}_q(\mathbb{G}^{\Gamma})$ are given by 
    \begin{align*}
        & \phi_{(e,f)}^{*_1} = (\Lambda\times 1)_{\tilde R^{-1}}(\phi_{(\bar e',\bar f)}) \eta_\mathrm{h},\\
        & \phi_{(e,f)}^{*_2} = (\phi_{(e',\bar f)})^\text{op} \eta_\mathrm{v},
    \end{align*}
    where $(\bar e',\bar f) = (e,f)^{\dagger_1}$ and $(e',\bar f) = (e,f)^{\dagger_2}$. Here,  the $\tilde R$-matrix is localized on $\partial f$.
    \item The regular $\bullet$-module structure on $\mathfrak{C}(\mathbb{G}^{\Gamma})$ over $\tilde{\cC}$ is *-compatible: there exist natural measureable isomorphisms
    \begin{equation*}
        (\phi\bullet\zeta)^{*_{1,2}} \cong \zeta^{*_{1,2}} \bullet\phi^{*_{1,2}},\qquad \forall~ \phi\in\mathfrak{C}(\mathbb{G}^{\Gamma}),\quad \zeta\in\tilde{\cC},
    \end{equation*}
    satisfying the obvious coherence conditions against the $\bullet$-module associator and the tensorator \eqref{tensorator}.
    \end{enumerate}       
\end{definition}
\noindent These can be understood as a 2-dimensional version of the *-operation on the holonomies defined in \cite{Alekseev:1994pa}, (4.14). 

\begin{rmk}
    The geometry of this $*$-operation is clear: they are directly induced from the 2-$\dagger$ structure on $\Gamma^2$. However, the appearance of the $R$-matrices $\tilde R$ is a purely quantum phenomenon, as one needs to "pass" the target edge $e'$ through the source edge $e$ of the face $f$ upon an framing rotation. In the usual 3d lattice Chern-Simons case, the appearance of the $R$-matrices is kept track of by the so-called auxiliary "cilia" on the graphs \cite{Alekseev:1994pa}. Similarly, we can introduce 2-dimensional "2-cilia" as extensions on our 2-graph $\Gamma^2$, which could help visualize some of the computations below.
\end{rmk}



Note these orientation reversals are anti-homomorphisms, in the sense that\footnote{Keep in mind $\mathfrak{C}(\mathbb{G}^{\Gamma})$ is a \textit{co}category without composition!}
\begin{equation*}
    -^{*_1}: \mathfrak{C}(\mathbb{G}^{\Gamma})\rightarrow (\mathfrak{C}(\mathbb{G}^{\Gamma}))^{\text{m-op},\text{c-op}},\qquad -^{*_2}: \mathfrak{C}(\mathbb{G}^{\Gamma})\rightarrow (\mathfrak{C}(\mathbb{G}^{\Gamma}))^{\text{op}}
\end{equation*}
consistent with \eqref{antilinear}, they will swap the left- and right-actions in the $\tilde{\cC}$-bimodule structure of $\mathfrak{C}(\mathbb{G}^{\Gamma})$.
\begin{proposition}
    The *-operations on $\mathfrak{C}(\mathbb{G}^{\Gamma})$ strongly commute: there exist measureable natural isomorphisms $(\phi^{*_1})^{*_2} \cong (\phi^{*_2})^{*_1}$. 
\end{proposition}
\begin{proof}
    Let $\bar\phi^\mathrm{h,v},\bar\Lambda^\mathrm{h,v}$ denote the evaluation of $\phi$ on, and the 2-gauge transformations $\Lambda$ in the proximity of, the 2-graphs living in the horizontal/vertical orientation reversal of $\Gamma^2$. It is clear from the strong commutativty \eqref{commutedagger} of the 2-$\dagger$ structure on $\Gamma^2$ that there exists an isomorphism of measureable fields 
    \begin{equation*}
        \overline{(\bar\phi^\mathrm{h})}^\mathrm{v}\cong\overline{(\bar\phi^\mathrm{v})}^\mathrm{h},
    \end{equation*}
    and that the intertwiners $\eta$ are idempotent and strongly commute $\eta_\mathrm{h}^{*_\mathrm{v}}\circ\eta_\mathrm{v} = \eta_\mathrm{v}^{*_\mathrm{h}}\circ\eta_\mathrm{h}$. Given the concrete realization of $\Lambda$ by $U$ via \textbf{Definition \ref{concretify}}, and given that $*_\mathrm{h,v}$ are anti-homomorphisms, we have
    \begin{align*}
        (\phi^{*_1})^{*_2} &= (\bar\phi^\mathrm{h}\bullet(1\times \bar U^\mathrm{h})_{\tilde R^{-1}}\eta_\mathrm{h})^{*_2} \\
        &\cong \eta_\mathrm{h}^{*}\circ ((1\times \bar U^\mathrm{h})_{(\tilde R^{-1})^T} \bullet-)\circ \eta_\mathrm{v}\circ(\overline{\bar\phi^\mathrm{h}}^{\mathrm{v}})\\
        &\cong \eta_\mathrm{h}^{*_\mathrm{v}}\eta_\mathrm{v}\Big( (1\times \overline{\bar U^\mathrm{h}}^{\mathrm{v}})_{((\tilde R^T)^{-1}\cdot -^\text{op})} \bullet \overline{\bar\phi^\mathrm{h}}^{\mathrm{v}}\Big)\\
        &\cong \eta_\mathrm{v}^{*_\mathrm{h}}\eta_\mathrm{h}(1\times \overline{\bar U^\mathrm{v}}^\mathrm{h})_{\big((-)^\text{op}\circ (\tilde R^T\cdot-)\big)^{-1}} \bullet \overline{\bar\phi^\mathrm{v}}^{\mathrm{h}}.
    \end{align*}
    The statement then follows once we can pass $\tilde R \circ -^\text{op}\cong -^\text{op}\circ (\tilde R^T\cdot-)$, but this is precisely the naturality of $\tilde R$ against $\tilde S_v=-^\text{op}$ mentioned in \S \ref{2gthopf}. 
\end{proof}

\subsubsection{Extending the *-operations to the lattice 2-algebra}
We are now in a position to extend this *-operation to the entirety of $\mathscr{B}^\Gamma$. In order to do so, we must prove that the relations, namely the covariance condition \eqref{regrep} and the braiding relation \eqref{braid}, must be preserved. 

\begin{proposition}
    The *-operations preserves the left-covariance condition \eqref{leftreg} (and hence also the right-covariance condition).
\end{proposition}
\begin{proof}
    By {\bf Proposition \ref{rightreg}}, it suffice to prove the statement for left-covariance. We will do this for the horizontal and the vertical orientation reversals at the same time. Towards this, let us introduce the notation $-^*$ to denote either $*_1$ or $*_2$, and denote $\bar \phi,\bar\Lambda$ the evaluation of 2-graph operators on, and the 2-gauge transformations in the proximity of, the corresponding orientation reversed 2-graphs. The only caveat is that for 2-gauge transformations, the *-operation $\zeta\mapsto \zeta^*$ comes with a {\it 1-graph} orientation reversal.

    We will treat $-^*$ as a anti-homomorphism also for the semidirect product structure $\lhd$ for $\mathscr{B}^\Gamma$; see \S \ref{lattice2alg}. Recall $\tilde\Delta^\text{op} = \sigma\tilde\Delta$, we then compute for each $\zeta\in\tilde{\cC}$ and right-covariant $\phi\in\mathfrak{C}(\mathbb{G}^{\Gamma})$, using the intertwining properties \eqref{2gtR} (and neglecting the $\mathrm{h}$-subscripts),
    \begin{align*}
        ((1\times U)_{\tilde\Delta(\zeta)}\bullet \phi)^* &\cong \phi^*\bullet (1\times U)_{\tilde\Delta^\text{op}({\zeta^*})} \\
        &\cong \bar\phi((-\bullet (1\times \bar U)_{\tilde R^{-1}})\circ  \eta \circ (-\bullet (1\times U)_{(\sigma\tilde\Delta)(\zeta^*)})) \\ 
        &\cong (\bar\phi\bullet (1\times \bar U)_{\tilde R^{-1}\cdot (\sigma\tilde\Delta)({\zeta^*})})\eta \\
        &\cong (\bar\phi\bullet(1\times \bar U)_{\tilde\Delta({\zeta^*}) \cdot \tilde R^{-1}})\eta \cong (\zeta^*\bullet\bar\phi)\bullet (1\times \bar U)_{\tilde R} \eta \\ 
        &\cong \zeta^*\bullet(\bar\phi \bullet (1\times\bar  U)_{\tilde R^{-1}}\eta) = \zeta^*\bullet \phi^*,
    \end{align*}
    where in the fourth line we have used the right-covariance property \eqref{leftreg} and in the fifth line a bimodule associator $(-\bullet \phi)\bullet - \cong -\bullet(\phi\bullet -)$.
\end{proof}

We now also need to show that these *-operations preserves the braid relation \eqref{braid}. This can be done in a completely analogous way as in the latter half of the proof of Lemma 8 in \cite{Alekseev:1994pa}. To do this, we first note that the quantum $R$-matrices $\tilde R$ are compatible with the antipode. Then,  by making use of the left-covariance property to pass $\tilde R$ to the left,
\begin{equation*}
    \phi\bullet(1\times \Lambda)_{\tilde R^{-1}} \cong (1\times \Lambda\times \Lambda)_{(1\otimes\tilde\Delta)\tilde R^{-1}}\bullet \phi,
\end{equation*}
a series of computations similar to the above proposition can be performed to show that the *-operation indeed preserves the braid relations.

The compatibility between the *-operation and the bimodule structure $\bullet$ then implies that the invariance condition \eqref{invarstatement} is also preserved.
\begin{theorem}
    The above *-operation extends to strongly-commutative functors 
    $$-^{*_1}:\mathscr{B}^\Gamma\rightarrow(\overline{\mathscr{B}}^{\Gamma})^{\text{m-op},\text{c-op}},\qquad ,-^{*_2}:\mathscr{B}^\Gamma\rightarrow(\overline{\mathscr{B}}^{\Gamma})^{\text{op}}.$$ Further, they descend to the observables $\mathscr{O}^\Gamma$.
\end{theorem}
\noindent This result is important for constructing scattering amplitudes on the lattice in a future work.
 



\section{Conclusion}\label{conclusion}
Given a Lie 2-group $\mathbb{G}$, this paper lays the foundation upon which the 4d 2-Chern-Simons theory can be quantized on the lattice. Based on the notion of measureable categories, we have introduced structures which capture the kinematical lattice degrees of freedom of the theory, and categoriefied the notion of quantum groups to the context of Hopf categories. This substantiates the expectations from the categorical ladder proposal of Baez-Dolan \cite{Baez:1995ph}. 

Based on the framework introduced in this paper, we provided a categorified notion of quantum groups, and described their Hopf (co)categorical structures. These are higher structures internal to the 2-category $\mathsf{Meas}$ (or the non-commutative version $\mathsf{Meas}_q$; see \S \ref{hopfopalg}), the measureable categories of Crane-Yetter, which are the natural background for the representation theory of Lie 2-groups \cite{Crane:2003gk,Yetter2003MeasurableC,Baez:2012}. We have also shown how, under certain technical assumptions, these categorical quantum group reduces to the known Lie 2-bialgebra symmetries of the 2-Chern-Simons action. 

\medskip 

As mentioned in the introduction, this work is part of a series towards the computation of 4-simplex scattering amplitudes for 2-Chern-Simons theory, and this shall remain the central goal.  Towards this, the companion paper \cite{Chen:2025?} examines the geometric, $SO(3)$-volutive aspects of the 2-representations of the categorified quantum enveloping algebra, and show that they form, in a suitable sense, a \textit{ribbon tensor 2-category}. Future work in the series will construct lattice scattering amplitudes from these 2-representations in order to resolve \textbf{Conjecture \ref{baezconjecture}}. 

We have also made numerous comments about how our framework can be directly applied to quantize the weak/semistrict 2-Chern-Simons theory \cite{Soncini:2014} on the lattice. The resulting lattice scattering amplitudes would serve as a vast generalization beyond the familiar 4d Crane-Yetter TQFT.



\subsubsection*{Higher-integrable boundaries of 2-Chern-Simons theory}
An interesting prospect is to study the boundary theories of 2-Chern-Simons theory through this lattice theoretic, categorical approach. Recent works \cite{Chen:2024axr,Schenkel:2024dcd} have examined the analogue of the localization procedure of Costello-Yamazaki \cite{Costello:2019tri} for 2-Chern-Simons theory, and they seem to host \textit{derived} current algebras as studied in the literature \cite{FAONTE2019389,Kapranov2021InfinitedimensionalL,Garner:2023zqn,alfonsi2024raviolo,Alfonsi:2024qdr}.  It would therefore be interesting to study how representations of $\mathbb{U}_q\G$ are related to the operator algebras of these 3d boundary theories,




\newpage

\appendix

\section{Hopf 2-algebras and the necessity of categorification}\label{catfyHopf2alg}
The central goal of this paper is to construct, in the combinatorial setting, the quantized algebra of observables in the 2-Chern-Simons theory. If we begin from the semiclassical perspective as inspired by \cite{Drinfeld:1986in,Jimbo:1985zk}, then we are prompted to consider a quantization of the Lie 2-bialgebra symmetries (and its classical 2-$r$-matrix) \cite{Bai_2013,chen:2022} underlying the 2-Chern-Simons action.

This led to the development of \textbf{Hopf 2-algebras}, which had appeared in various guises in the literature \cite{Wagemann+2021,Chen:2023tjf,LU2002119,Majid:2012gy}. The formulation which we shall consider are 2-term \textit{Hopf $A_\infty$-algebras}, fitting into the following diagram
\[
\begin{tikzcd}
	{\text{Lie bialgberas}} && {\text{Hopf algberas}} \\
	{L_\infty\text{-bialgebras}} && {\text{Hopf }A_\infty\text{-algebras}}
	\arrow[from=1-1, to=1-3]
	\arrow[from=1-1, to=2-1]
	\arrow[from=1-3, to=2-3]
	\arrow[from=2-1, to=2-3]
\end{tikzcd}\]
for which one can describe an analogue of the "universal envelop" of the Lie 2-algebra $\G=\h\xrightarrow{\mu_1}\g$. 

\subsection{Universal enveloping 2-algebra}
We first pin down the structures we wish to study.
\begin{definition}
    Let $\G=\operatorname{Lie}\mathbb{G} = \h\xrightarrow{\mu_1}\g$ denote the Lie 2-algebra associated to the Lie 2-group $\mathbb{G}$. The \textbf{universal enveloping algebra} of $\G$, $$U\G=U\h\xrightarrow{D\mu_1}U\g,$$ is the tensor $A_2$-algebra (ie. a 2-term chain complex with differential graded algebra structure)
    \begin{equation*}
        \bigoplus_n \h^{\otimes n} \xrightarrow{D_{\mu_1}}\bigoplus_n \g^{\otimes n},\qquad D_{\mu_1} = \sum_i(-1)^{i-1} 1\otimes\dots\otimes \mu_1\otimes\dots\otimes 1
    \end{equation*}
    freely generated by the Lie algebras $\h,\g$, subject to the following relations on homogeneous elements
    \begin{align*}
        &x\otimes x'-x'\otimes x = [x,x'],\qquad x,x'\in U\g,\\
        &y\otimes y' - y'\otimes y = [y,y'],\qquad y,y'\in U\h,\\
        & x\otimes y - y\otimes x = \mu_2(x,y),\qquad x\in U\g,~y\in U\h,
    \end{align*}
    where $x\in \g\subset U\g$ is given the degree 1 and $y\in \h\subset U\h$ is given a degree 2. We denote the unique element with homogeneous degree-0 by "$1$".
\end{definition}
\noindent It can be seen that $U\G$ has the "minimal" amount of relations to admit a canonical injection $\iota:\G\hookrightarrow U\G$. This can be stated as a universal property.

Now since 2-term chain complexes can be equivalently thought of as categories internal to $\mathsf{Vect}$ --- aka. the so-called \textit{Baez-Crans} 2-vector spaces $\mathsf{2Vect}^{BC}=\operatorname{Cat}_\mathsf{Vect}$ \cite{Baez:2003fs}, the following construction was given in \cite{Chen:2023tjf}.
\begin{proposition}
    $U\G$ is a \textbf{Hopf 2-algebra}: namely a Hopf category internal to $\mathsf{Vect}$.
\end{proposition}
\noindent The theory of \textit{weak} $A_\infty$-algebras and their modules were studied in \cite{Chen:2023tjf}. They serve as the foundation for the universal envelopes $U\G$ of the so-called weak Lie 2-algebras \cite{Chen:2013}.

\begin{rmk}\label{Ufunctor}
    Without specifying the coproduct, the above construction has also appeared previously in \cite{Wagemann+2021} as a "universal enveloping functor"
\begin{equation*}
    U: \mathsf{Lie2Alg}\rightarrow \mathsf{2Alg},\qquad \G\mapsto U\G
\end{equation*}
which sends a Lie 2-algebra to a graded associative (ie. $A_2$-)algebra. Since (Lie) 2-algebras are equivalent to (Lie) algebra objects in $\operatorname{Cat}_\mathsf{Vect}=\mathsf{2Vect}^{BC}$ \cite{Baez:2003fs}, we can understand $U$ as a (strict) functor of internal categories.
\end{rmk}

We can prove the graded analogue of the following property of the universal enveloping algebra \cite{Hall2015}.
\begin{proposition}\label{derivationaction}
    There is a canonical isomorphism of $A_2$-algebras $U\G\xrightarrow{\sim}\operatorname{Diff}(\mathbb{G})^{\mathbb{G}}$, where $\operatorname{Diff}(\mathbb{G})^{\mathbb{G}}$ is the $A_2$-algebra of invariant (under both group and groupoid multiplication) derivations on the $A_2$-algebra of functions $C(\mathbb{G})$ on $\mathbb{G}$.
\end{proposition}
\begin{proof}
    Let $\mathfrak{X}_\mathbb{G}\in\mathsf{Lie2Alg}$ denote respectively the $L_2$-algebras of left-invariant (again, under both group and groupoid multiplication) vector fields on $\mathbb{G}$, and $\operatorname{Der}(\mathbb{G})\in\mathsf{Lie2Alg}$ the graded derivations on the functions $C(\mathbb{G})$.
    
    The fact that we have a canonical $L_2$-algebra isomorphism $\mathfrak{X}_\mathbb{G}\xrightarrow{\sim}\operatorname{Der}(\mathbb{G})$ is known in \cite{Chen:2012gz}. At the same time, $\G = \mathfrak{X}_\mathbb{G}|_{(1,1)}$ is by definition the left-invariant vectors over the 2-group unit $(1,1)=(1_\mathsf{H},1_G)\in\mathbb{G}$, hence this $L_2$-algebra isomorphism restricts to $\G\xrightarrow{\sim}\operatorname{Der}(\mathbb{G})^{\mathbb{G}}$, where the target $\operatorname{Der}(\mathbb{G})^{\mathbb{G}}$ consist of the $\mathbb{G}$-invariant derivations on $C(\mathbb{G})$.

    Define $\operatorname{Diff}(\mathbb{G})^{\mathbb{G}}\in\mathsf{2Alg}$ the $A_2$-algebra of invariant differential operators on $C(\mathbb{G})$. By applying the functor $U$, the above $L_2$-isomorphism lifts to a $A_2$-algebra homomorphism $U\G\rightarrow \operatorname{Diff}(\mathbb{G})^\mathbb{G}$. It is not hard to show that this map is bijective.
\end{proof}
\noindent This is how  $U\G$ (or $U_q\G$) acts canonically on $C(\mathbb{G})$ (or $C_q(\mathbb{G})$). This fact was used in \S \ref{categoricalquantumdeformations}.

There is a left adjoint functor to the universal envelop $U$ defined in \textit{Remark \ref{Ufunctor}}, which is the \textit{Lie-ification} functor $L: \mathsf{2Alg}\rightarrow \mathsf{Lie2Alg}$ \cite{Wagemann+2021}. It was proven in the appendix B of \cite{Chen:2023tjf} that, if $(A;R)$ is a Hopf 2-algebra equipped with a $R$-matrix (see \S \ref{coprod}), then $L(A,R) = (L(A);r)$ is a Lie 2-bialgebra equipped with a classical 2-$r$-matrix.


\subsection{Problems with Hopf 2-algebras}
The formulation above is fine by itself, but there are several problems suffered by Baez-Crans 2-vector spaces (i) in regards to its modules, and (ii) in regards to its additivity structures.

\subsubsection{Representation theory in $\mathsf{2Vect}^{BC}$ and additivity at all levels}\label{hopf2algproblems}
By definition, a module $V\in\operatorname{Mod}_{\mathsf{2Vect}^{BC}}(A)$ of an $A_2$-algebra $A$ is a Baez-Crans 2-vector space $V\in\mathsf{2Vect}^{BC}$ equipped with an $A_2$-algebra map $\rho: A\rightarrow \operatorname{End}(V)$, where $\operatorname{End}(V)$ is the endomorphism 2-algebra \cite{Angulo:2018}. It was proven in \cite{Chen:2023tjf} that the 2-category $\operatorname{Mod}_{\mathsf{2Vect}^{hBC}}(A)$ of modules of a Hopf 2-algebra $A$, equipped with a 2-$R$-matrix. is \textit{braided monoidal} \cite{BAEZ1996196,GURSKI20114225,neuchl1997representation}.

However, such 2-representations are known to not carry any non-trivial $k$-invariants \cite{heredia2016representations2groupsbaezcrans2vector}; that is, all modules of $A$ are isomorphic to a trivial one. A generalization that remedies this issue was introduced in \cite{Chen:2023tjf}, which consist of $\mathsf{Vect}$-internal categories whose algebra objects are \textit{lax} (ie. \textit{pseudo}algebras \cite{Fiore2004PseudoLB}) --- namely the (Lie) algebra objects are 2-term ($L_\infty$-)$A_\infty$-algebras \cite{Stasheff:1963}. 

This led to the definition of the \textit{homotopy refinement} $\mathsf{2Vect}^{hBC}$ of Bae-Crans 2-vector spaces, and it can be shown that the modules $V\in\operatorname{Mod}_{\mathsf{2Vect}^{hBC}}(A)$ of a 2-term Hopf $A_\infty$-algebra do indeed  carry non-trivial $k$-invariants in .

\medskip

The second problem concerning additivity, on the other hand, is a separate issue. When viewed as $\mathsf{Vect}$-internal categories in $\mathsf{2Vect}^{BC}$ (or even $\mathsf{Vect}$-internal psedo-catgories in $\mathsf{2Vect}^{hBC}$), the composition of their morphisms are given simply by vector space \textit{addition}
\begin{equation*}
    x\xrightarrow{y}x'\xrightarrow{y'}x'' = x\xrightarrow{y+y''}x''.
\end{equation*}
This has the following unsavoury consequence: that a $\mathsf{Vect}$-internal category $A= A_1\otimes A_0\rightrightarrows A_0$ cannot be made additive. Since the origin $0\in A_1$ already occupies the role of the unit morphism over $A_0$, there is no way to define a "zero morphism".

This is of particular problem for the group 2-algebra $A=\bbC[\mathbb{G}]$ of a finite 2-group $\mathbb{G}$ specifically. Indeed, one cannot view $A$ as a category internal to $\mathsf{Vect}$ while simultaneously keeping the \textit{multiplicative} groupoid prdouct in $\mathbb{G}$ as the composition in $A$. We either have to forgo the Baez-Crans 2-vector space formulation altogether (cf. the (2-)groupoid linearization in \cite{Bullivant:2019tbp}), or one must perform a highly non-canonical quotient \cite{Wagemann+2021,Chen:2023tjf}.


\subsubsection{Quantum states and correlation functions of Wilson surfaces}\label{2holstates}
The above is not just a mathematical problem, but in fact a {physical} one. Indeed, if we take inspiration from \cite{Alekseev:1994pa,WITTEN1990285,Majid:2000}, then "Hopf 2-algebras" should model the structure of the \textit{correlation functions} in 2-gauge theory. To be more explicit, consider the surface-/path-ordered 2-holonomies \cite{Baez:2004in,schreiber2013connectionsnonabeliangerbesholonomy,Yekuteli:2015,Chen:2024axr} $S\exp \int_\Sigma B, ~P\exp\int_\gamma A$, which can be written as a 2-functor (see also \cite{Kim:2019owc}) 
\begin{equation*}
    P^2M^4\rightarrow (\mathbb{G},\ast),\qquad (\Sigma,\gamma)\mapsto \left(S\exp \int_\Sigma B, ~P\exp\int_\gamma A\right)
\end{equation*}
2-path groupoid $P^2M^4$ (also called the "surface 2-groupoid") on $M^4$. Given some appropriate categorical notion of "2-trace" \cite{Ganter:2006,Ganter:2014,Bartlett:2009PhD,Huang:2024} associated to a higher irreducible representation $\rho$ of $\mathbb{G}$, the associated Wilson surface correlation functions can be formally write as a path integral (see also \cite{Wen:2019} for the finite 2-group case)
\begin{equation*}
    \langle W_{\rho}(\Sigma,\gamma)\rangle = \frac{1}{Z}\int D[A,B] \operatorname{Tr}_\rho\left[\left(S\exp \int_\Sigma B\right)\cdot \left(P\exp\int_\gamma A\right)\right]e^{i2\pi k S_{2CS}[A,B]},
\end{equation*}
which should inherit the \textit{multiplicative} gluing composition $\Sigma\circ\Sigma'=\Sigma\cup_\gamma\Sigma'$ of surfaces from the 2-holonomies. 

However, as we have seen above, one cannot keep both additivity and such multiplicative vertical products in the framework of Baez-Crans 2-vector spaces --- Wilson surface correlators \textit{cannot} be modelled by Hopf 2-algebras!

\medskip

The point of categorifying to the context of the measureable categories $\mathsf{Meas}$, as well as the technical conditions in "Hypothesis (H)" \textbf{Definition \ref{hypH}}, is to remedy this problem. The decategorification 2-functor $\lambda: \mathsf{Cat}_\mathsf{Meas}\rightarrow\mathsf{Cat}_\mathsf{Vect}=\mathsf{2Vect}^{BC}$ serves as a way to lift Hopf 2-algebras to the structure of the 2-graph operators $\mathfrak{C}_q(\mathbb{G}^{\Gamma})$, which has enough room for multiplicative composition laws without sacrificing additivity.

\medskip

We also mention that the choice $\mathsf{Meas}$ is not random. The measureable categories are the infinite-dimensional analogues of Baez's 2-Hilbert spaces $\mathsf{2Hilb}$ \cite{Baez1996HigherDimensionalAI}, and it serves as the natural backdrop for the representation theory of \textit{Lie} 2-groups \cite{Crane:2003gk,Baez:2012} and Lie groupoids \cite{TRENTINAGLIA2010750}. As we have mentioned in \S \ref{hopfopalg}, $\mathsf{Meas}$ is also closely related to the frameworks of \cite{Henriques2017-gm,Kristel:2023gus}.

\newpage

\printbibliography

\end{document}